\documentclass[11pt]{article}
\RequirePackage{etex}
\usepackage[parfill]{parskip}
\usepackage[hyphens]{url}
\usepackage[colorlinks=true, allcolors=blue, breaklinks=true]{hyperref}
\usepackage{amsfonts,mathrsfs,eucal,xspace,graphicx}
\usepackage[dvipsnames]{xcolor}
\usepackage{endnotes}
\usepackage{amssymb,latexsym}
\usepackage{enumitem}
\setlist{itemsep=0mm}

\usepackage{float}
\usepackage{amsmath, mathtools, physics}
\usepackage{complexity}

\usepackage[normalem]{ulem}
\usepackage{enumitem}
\usepackage{amssymb}
\usepackage{amsthm}
\usepackage[font=small]{caption}
\usepackage{subcaption}
\usepackage[disable]{todonotes}
\usepackage{tabularx,environ}
\usepackage{enumitem}
\usepackage{lmodern}
\usepackage{anyfontsize}

\usepackage{varioref}
\usepackage[nameinlink]{cleveref}
\crefname{case}{case}{cases}
\usepackage[noend]{algpseudocode}
\usepackage[linesnumbered,ruled,vlined,noend]{algorithm2e}
\usepackage{verbatim}

\usepackage{amsthm}
\usepackage[margin=1in]{geometry}
\usepackage[utf8]{inputenc}
\usepackage{outlines}
\usepackage{amsmath}
\usepackage{amsthm}
\usepackage{amssymb}
\usepackage{comment}
\usepackage{capt-of}
\usepackage{mathtools}
\usepackage{enumitem}
\usepackage{physics}
\usepackage{slashed}
\usepackage{braket}
\usepackage{nccmath}
\usepackage{float}
\usepackage{dcolumn}
\usepackage{tikz}
\usetikzlibrary{positioning}
\usetikzlibrary{quantikz2}
\usetikzlibrary{trees}

\usepackage{bm}
\usepackage{array}
\usepackage{color}
\usepackage{wrapfig}
\usepackage{refcount}
\definecolor{refcol}{RGB}{0,0,205}
\usepackage{microtype}
\usepackage{nicematrix}
\NiceMatrixOptions{cell-space-limits = 1pt}
\usepackage{tikzscale}
\usepackage[symbol]{footmisc}
\usetikzlibrary{arrows.meta}
\usetikzlibrary{patterns}

\makeatletter

\newcommand{\inlineitem}[1][]{%
\ifnum\enit@type=\tw@
    {\descriptionlabel{#1}}
  \hspace{\labelsep}%
\else
  \ifnum\enit@type=\z@
       \refstepcounter{\@listctr}\fi
    \quad\@itemlabel\hspace{\labelsep}%
\fi}
\makeatother
\newcommand{\ksum}{k\textsc{-sum}}

\newcommand{\blockksum}{\textsc{Block} \; k\textsc{-sum}}
\newcommand{\bkkfull}{\blockksum_N \circ \ksum_N}
\newcommand{\bkk}{\textsc{Bkk}}
\newcommand{\orfunc}{\textsc{Or}}
\newcommand{\andfunc}{\textsc{And}}
\newcommand{\andor}{\textsc{And} \circ \textsc{Or}}

\usepackage{empheq}

\usepackage{hyperref}
\hypersetup{
    colorlinks=true,
    linkcolor=blue,
    filecolor=magenta,      
    urlcolor=cyan,
    pdftitle={Quantum advantage and lower bounds in parallel query complexity},
    pdfpagemode=FullScreen,
    }

\theoremstyle{plain}
\newtheorem{theorem}{Theorem}
\newtheorem*{theorem*}{Theorem}
\newtheorem{proposition}[theorem]{Proposition}
\newtheorem{lemma}[theorem]{Lemma}

\newtheorem{corollary}[theorem]{Corollary}
\newtheorem{conjecture}[theorem]{Conjecture}
\newtheorem{definition}[theorem]{Definition}
\newtheorem{prob}[theorem]{Problem}
\theoremstyle{remark}

\newtheorem{observation}[theorem]{Observation}

\SetKwInput{KwInput}{Input}
\SetKwInput{KwOutput}{Output}

\crefname{enumi}{part}{parts}
\Crefname{enumi}{Part}{Parts}

\crefname{algocf}{algorithm}{algorithms}
\Crefname{algocf}{Algorithm}{Algorithms}
\crefname{algocfline}{algorithm}{algorithms} 
\Crefname{algocfline}{Algorithm}{Algorithms}

\newcommand{\defeq}{\coloneqq}
\newcommand{\bs}{\textsf{bs}}
\newcommand{\cert}{\textsf{C}}

\newcommand{\cheatsheet}{\textsf{cs}}
\newcommand{\canonicalcheatsheet}{\textsf{ccs}}

\newcommand{\djana}{\textsf{djana}}

\newcommand{\forrelationana}{\textsf{forana}}

\newcommand{\Det}{\textsf{D}}
\newcommand{\Rand}{\textsf{R}}
\newcommand{\Quant}{\textsf{Q}}
\newcommand{\NN}{\textsf{NN}}
\newcommand{\Adv}{\textsf{Adv}}
\newcommand{\concat}{\frown}
\newcommand{\oracle}{\mathcal{O}}

\newcommand{\dist}{\mathcal{P}}
\newcommand{\hybrid}{\text{hybrid}}
\newcommand{\correlated}{\textsc{cor}}
\newcommand{\domain}{\mathcal{D}}
\newcommand{\Add}{\textsc{add}}
\newcommand{\Dt}{\textsc{dt}}
\newcommand{\Bc}{\textsc{bc}}
\newcommand{\tg}{\textsc{tg}}
\newcommand{\In}{\textsc{in}}
\newcommand{\Ip}{\textsc{ip}}

\newcommand\floor[1]{\left\lfloor#1\right\rfloor}
\newcommand\ceil[1]{\left\lceil#1\right\rceil}

\makeatletter
\newcommand{\problemtitle}[1]{\gdef\@problemtitle{#1}}
\newcommand{\probleminput}[1]{\gdef\@probleminput{#1}}
\newcommand{\problempromise}[1]{\gdef\@problempromise{#1}}
\newcommand{\problemquestion}[1]{\gdef\@problemquestion{#1}}
\newcommand{\problempremise}[1]{\gdef\@problempremise{#1}}

\NewEnviron{problem}{
\probleminput{}\problempromise{}\problemquestion{}
  \BODY
  \par\addvspace{.5\baselineskip}
  \noindent
  \begin{tabularx}{\textwidth}{@{\hspace{\parindent}} l X c}
    \textbf{Input:} & \@probleminput \\
    \textbf{Promise:} & \@problempromise \\
    \textbf{Question:} & \@problemquestion
  \end{tabularx}
  \par\addvspace{.5\baselineskip}
}

\NewEnviron{problemTotal}{
\probleminput{}\problempremise{}\problemquestion{}
  \BODY
  \par\addvspace{.5\baselineskip}
  \noindent
  \begin{tabularx}{\textwidth}{@{\hspace{\parindent}} l X c}
    \textbf{Input:} & \@probleminput \\
    \textbf{Question:} & \@problemquestion
  \end{tabularx}
  \par\addvspace{.5\baselineskip}
}

\NewEnviron{problemPremise}{
\probleminput{}\problempromise{}\problemquestion{}
  \BODY
  \par\addvspace{.5\baselineskip}
  \noindent
  \begin{tabularx}{\textwidth}{@{\hspace{\parindent}} l X c}
    \textbf{Input:} & \@probleminput \\
    \textbf{Premise:} & \@problempremise \\
    \textbf{Question:} & \@problemquestion
  \end{tabularx}
  \par\addvspace{.5\baselineskip}
}
\makeatother

\crefname{enumi}{problem}{problems}
\crefname{enumi}{Problem}{Problems}

\crefname{enumi}{part}{parts}
\Crefname{enumi}{Part}{Parts}

\title{\bfseries\Large

Quantum advantage and lower bounds in parallel query complexity}

\author{Joseph Carolan\footnote{\texttt{jcarolan@umd.edu}}, Amin Shiraz Gilani\footnote{\texttt{asgilani@umd.edu}}, Mahathi Vempati\footnote{\texttt{mahathi@umd.edu}} \smallskip \\
\small Department of Computer Science, University of Maryland \\
\small Institute for Advanced Computer Studies, University of Maryland \\
\small Joint Center for Quantum Information and Computer Science, University of Maryland}
\date{\today}

\makeatletter
\newcommand\thefontsize{The current font size is: \f@size pt}
\makeatother

\begin{document}
\addtocontents{toc}{\protect\setcounter{tocdepth}{2}}

\date{}

\fontsize{11}{13.2}
\selectfont
\sloppy

\maketitle

\begin{abstract}

It is well known that quantum, randomized and deterministic (sequential) query complexities are polynomially related for total boolean functions. We find that significantly larger separations between the parallel generalizations of these measures are possible. In particular, 
\vspace{0.6em}

\begin{enumerate}[label=(\arabic*)]
    \item \label{item:result1} We employ the cheatsheet framework to obtain an unbounded parallel quantum query advantage over its randomized analogue for a total function, falsifying a conjecture of [\href{https://arxiv.org/abs/1309.6116}{Jeffery et al. 2017}].
    \item \label{item:result2} We strengthen \ref{item:result1} by constructing a total function which exhibits an unbounded parallel quantum query advantage despite having no sequential advantage, suggesting that genuine quantum advantage could occur entirely due to parallelism. 

    \item \label{item:result3} We construct a total function that exhibits a polynomial separation between 2-round quantum and randomized query complexities, contrasting a result of [\href{https://arxiv.org/abs/1001.0018}{Montanaro. 2010}] that there is at most a constant separation for 1-round (nonadaptive) algorithms.
    \item  We develop a new technique for deriving parallel quantum lower bounds from sequential upper bounds. We employ this technique to give lower bounds for Boolean symmetric functions and read-once formulas, ruling out large parallel query advantages for them.
\end{enumerate} 
\vspace{0.6em}

We also provide separations between randomized and deterministic parallel query complexities analogous to items \ref{item:result1}-\ref{item:result3}.

\end{abstract}
\newpage
\tableofcontents

\listoftodos

\newpage

\section{Introduction} \label{section:introduction}

Quantum query complexity is a widely studied model for understanding the capabilities and limitations of quantum computers. Many of the important quantum algorithms and quantum-classical separation results are expressed in this model. For instance, the quantum period finding algorithm, a major ingredient in Shor's factoring algorithm \cite{Shor94}, was first developed in the query model. There are many other examples of query problems for which quantum algorithms provably achieve exponential query advantage over their classical counterparts \cite{Simon97, Childs03}.

However, these problems have partial domains, meaning they are not defined on most inputs. For functions whose domain includes all inputs (total functions), it is known that the quantum, randomized and deterministic query complexities are all polynomially related \cite{Beals98,BetterbealsAaronson21}. Grover's seminal algorithm \cite{Grover96} for unstructured search achieves a quadratic speed-up over classical algorithms, and Ambainis et al. \cite{Ambainis17} and Aaronson et al. \cite{Aaronson16} construct functions that separate quantum and deterministic query complexities by a $4$th power, and quantum and randomized query complexities by a $3$rd power respectively.\footnote{\cite{Aaronson16} originally showed a $2.5$th power separation, which was boosted to a $3$rd power via tight lower bounds on $k$-fold Forrelation shown in \cite{Bansal21,Sherstov21}.} These separations cannot be significantly improved in the sequential query model \cite{Beals98,BetterbealsAaronson21}. 

A natural question is whether these algorithmic limitations are robust against generalizations of sequential query complexity. 
In this work, we consider a parallel model where an algorithm is allowed to make many queries at each step, as introduced in \cite{Jeffery17}. In particular, the $p$-parallel quantum query complexity, denoted by $\Quant^{p \parallel}$, of a function is the minimum number of steps needed to compute it with bounded error, if at each step the algorithm is allowed to make $p$ non-adaptive quantum queries.\footnote{We define the deterministic parallel query complexity $\Det^{p \parallel}$ and randomized bounded error parallel query complexity $\Rand^{p \parallel}$ similarly.} We will also refer to this quantity by \textit{$p$-query depth}, \textit{$p$-query rounds} and \textit{$p$-query layers} interchangeably. Notice that, for any function $f$ and any parallelism $p$, $\Quant(f)/p \leq \Quant^{p \parallel}(f) \leq \Quant(f)$ since any $q$-round $p$-parallel algorithm can be simulated by a $pq$-query sequential algorithm, and any $q$-query sequential algorithm can be simulated by a $q$-round $p$-parallel algorithm.\footnote{Similar bounds hold for $\Rand^{p \parallel}$ and $\Det^{p \parallel}$.}

This model is motivated by the fact that fault tolerant quantum computers must be inherently parallel: a large scale quantum computer which cannot do many gates at once cannot correct local errors faster than they occur. Therefore, quantum computers will necessarily be able to perform operations in parallel, meaning one should design algorithms to take maximum advantage of this capability. Further, near term quantum computations are limited to low depth due to short decoherence times. These factors make the parallel query model a more natural abstraction of real devices than a purely sequential model.

Prior work addressing parallel query complexity has indicated that quantum advantage tends to disappear as parallelism increases. 
Zalka \cite{Zalka99} showed the parallel complexity for searching an unstructured list of size $N$ is $\Theta\left(\sqrt{N/p}\right)$. This was generalized by Grover and Radhakrishnan \cite{Grover04}, who proved that $\Tilde{\Theta}\left(\sqrt{Nt/ (p \cdot \min \{t, p\})}\right)$ $p$-parallel queries are necessary and sufficient to find $t$ marked elements in a list of size $N$. Along with providing a systematic study of parallel quantum query complexity, Jeffery, Magniez and de Wolf \cite{Jeffery17} showed that the $p$-parallel quantum query complexity of $k$-\textsc{sum} is $\tilde{\Theta}\left((N/p)^{k/k+1}\right)$. 

Furthermore, \cite{Jeffery17} introduce the following conjecture:

\begin{conjecture} [\cite{Jeffery17}]
    For any total function $f:\{0,1\}^N \rightarrow \{0,1\}$ and any $p$, there is at most a polynomial quantum advantage for $p$-parallel query algorithms. That is, 
    $$\emph{\Det}^{p\parallel} (f) = \poly(\emph{\Quant}^{p\parallel} (f)).$$\label{conj:JefferyTotalFunctions}
    \vspace{-1pc}
\end{conjecture} 
In fact, \cite{Jeffery17} proved this conjecture for all $p$ polynomially smaller than the block sensitivity ($\bs$, defined in \Cref{section:preliminaries}) of $f$, generalizing a result of Beals et al. \cite{Beals98}.

The notion of $p$-parallel query complexity captures the minimum ``depth'' needed to compute some function for a given ``width'' (see \Cref{fig:illustrationParallelQuery}). An alternative complexity measure relevant to parallelism is the minimum ``width'' needed to compute some function for a fixed ``depth'', which we call the $k$-adaptive query complexity, and denote by $\Quant^{k \perp}(f)$ (respectively $\Rand^{k\perp}(f)$, $\Det^{k\perp}(f)$). Precisely, we define $k$-adaptive query complexity as the minimum $p$ such that there is a $p$-parallel quantum (respectively randomized, deterministic) algorithm which makes $k$ many $p$-parallel queries and computes $f$. This complexity measure is not widely studied for total functions. However, a result of Montanaro \cite{montanaro2010nonadaptive} fully characterizes the relevant $1$-adaptive (or non-adaptive) query complexities, up to a factor of $2$: for any total $f:\{0,1\}^{N} \rightarrow \{0,1\}$, $\Quant^{1\perp}(f)\geq \Det^{1\perp}(f))/2$. In the context of partial functions, a result by Girish et al. \cite{girish2023power} characterized the maximal separation between $\Quant^{k\perp}$ and $\Quant^{k+1\perp}$.

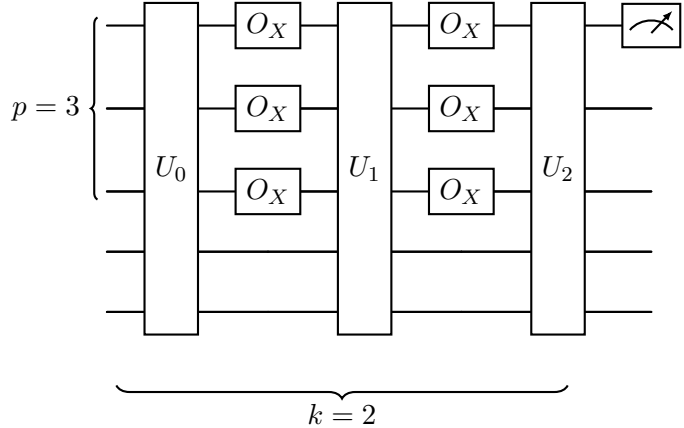
\begin{figure}
\begin{center}
        \begin{quantikz}
            \lstick[wires=3]{$p=3$} \qw & \gate[wires=5]{U_0} & \gate{O_X} & \gate[wires=5]{U_1} & \gate{O_X} & \gate[wires=5]{U_2} & \meter{} \\
            \qw & \qw & \gate{O_X} & \qw & \gate{O_X} & \qw &\qw \\
            \qw & \qw & \gate{O_X} & \qw & \gate{O_X} & \qw &\qw \\
            \qw & \qw & \qw & \qw & \qw & \qw & \qw \\
            \qw & \qw & \qw & \qw & \qw & \qw & \qw \\
        \end{quantikz}
        
        \begin{tikzpicture}[thick]
            \draw [decorate,decoration={brace,amplitude=5pt,mirror,raise=0ex}]
  (2,0) -- (8,0) node[midway,yshift=-1em]{$k=2$};
        \end{tikzpicture}
\end{center}
    \caption{A parallel quantum algorithm with $k=2$ adaptive steps and $p=3$ queries per step; the last two wires are for workspace. In general, $p$ denotes the number of non-adaptive queries per step, and $k$ denotes the number of steps. The measure $Q^{p\parallel}$ denotes the minimal $k$ needed to compute a function for a fixed $p$, while the measure $Q^{k\perp}$ denotes the minimal $p$ needed to compute a function for a fixed $k$. }
    \label{fig:illustrationParallelQuery}
\end{figure}

\subsection{Main results}
\label{subsec:MainResults}

\subsubsection{Unbounded parallel separations for total functions}
\label{sec:IntroSep}
Our first result involves adapting the cheatsheet framework of Aaronson et al. \cite{Aaronson16} to the parallel query setting, and using it to exhibit total function separations between parallel query complexity measures. In particular, we show that the canonical cheatsheet function (see \Cref{def:Canonicalcheatsheet}) exhibits unbounded separation between quantum and randomized parallel query complexities, falsifying \Cref{conj:JefferyTotalFunctions}. Further, we observe that this function demonstrates that $3$-adaptive quantum algorithms can be more efficient in terms of total queries than even fully sequential randomized algorithms. More formally, we prove the following theorem.

\begin{theorem}[Restatement of \Cref{thm:cheatsheetQpvsRpSeparation}] \label{thm:informalTotalAdaptiveSep}
    There exists a (total) Boolean function $h:\{0,1\}^N \rightarrow \{0,1\}$ such that for some $p=\Tilde O(\bs(h))$ and some constants $\epsilon, \delta > 0$, we have 
    \begin{enumerate} [label=(\roman*)]
        \item $\emph{\Rand}^{p \parallel}(h) = \tilde{\Omega}(N^{\epsilon})$, 
        \inlineitem $\emph{\Quant}^{p \parallel}(h) \leq 3$,  
        \inlineitem $\emph{\Quant}^{3 \perp}(h) = \tilde O\left(\emph{\Rand}(h)^{1-\delta}\right)$.
    \end{enumerate}
\end{theorem}

It is worth noting that the large separation in parallel query complexity occurs when the parallelism almost exactly equals the block sensitivity \(\bs(f)\), contrasting a result of Jeffery et al (\cite{Jeffery17}, Theorem 12) which eliminates this possibility for $p=O(\bs(f)^{1-\epsilon})$ for any constant $\epsilon>0$. Furthermore, the separation between $\Quant^{3\perp}(f)$ and $\Rand(f)$ contrasts Montanaro's result (\cite{montanaro2010nonadaptive}, Theorem 1) that for total functions, $1$-adaptive quantum algorithms cannot be more than twice as efficient as the trivial $1$-adaptive deterministic algorithm.

We also show a total function that achieves a similar separation between randomized and deterministic parallel query complexities (see \Cref{thm:cheatsheetRpvsDpSeparation}). We do this by modifying the canonical cheatsheet function to instead embed a partial function with a large randomized-deterministic separation.

It is worth noting that this result is also sufficient to falsify \Cref{conj:JefferyTotalFunctions}.

These results provide a recipe for achieving exponential parallel query advantages (from quantumness or randomness) for total functions by amplifying a polynomial sequential advantage for such functions. 
However, it is unclear whether an advantage can originate entirely from a quantum or randomized algorithm's ability to utilize parallelism more effectively than a randomized or deterministic algorithms. In other words, is there a function which has no sequential quantum query advantage, yet a large parallel quantum query advantage? We answer this question in the affirmative, even for total functions.

\subsubsection{Unbounded genuine parallel separations}
\label{sec:IntroGenuineSep}
Our starting point is a partial function with the desired property. That is, we construct a partial function with (almost) equal randomized and quantum sequential query complexities, that does not admit any parallel advantage for randomized algorithms, yet admits near-maximal parallel advantage for quantum algorithms. In particular, we show the following theorem.  

\begin{theorem}[Restatement of \Cref{thm:non-seq_RvsQ_parallel_separation}] \label{thm:partialexpsep}
    There exists a (partial) Boolean function $f:\mathcal{D} \rightarrow \{0,1\}$ with $\mathcal{D} \subset \{0,1\}^N$ such that for some $p = O(\emph{\Quant}(f))$ and some constant $\epsilon > 0$, we have
    \begin{enumerate} [label=(\roman*)]
        \item $\emph{\Rand}(f) = \tilde\Theta(N^\epsilon)$, 
        \inlineitem $\emph{\Quant}(f) = \tilde\Omega(\emph{\Rand}(f))$,
        \inlineitem $\emph{\Rand}^{p \parallel}(f) = \Tilde \Omega(\emph{\Rand}(f))$, 
        \inlineitem $\emph{\Quant}^{p \parallel}(f) = 1$.
    \end{enumerate}
\end{theorem}

Note that the function $f$ in the above theorem perfectly parallelizes quantumly but does not parallelize at all classically. We discuss the construction of $f$ in more detail in the technical overview. We totalize $f$ using the cheatsheet framework to give a total function with no sequential quantum advantage but unbounded parallel quantum advantage, giving rise to the following theorem. 

\begin{theorem}[Restatement of \Cref{thm:no_seq_parallel_sepa_QvsR}] \label{thm:informalTotalExpSep}
    There exists a (total) Boolean function $h:\{0,1\}^N \rightarrow \{0,1\}$ such that for some $p$ and some constants $\epsilon, \delta > 0$ (with $\delta < \epsilon$), we have 
    \begin{enumerate} [label=(\roman*)]
        \item $\emph{\Rand}(h) = \tilde\Theta(N^\epsilon)$, 
        \inlineitem $\emph{\Quant}(h) = \tilde{\Omega}(\emph{\Rand}(h))$,
        \inlineitem $\emph{\Rand}^{p \parallel}(h) = \Tilde \Omega(N^\delta)$, 
        \inlineitem $\emph{\Quant}^{p \parallel}(h) \leq 3$.
    \end{enumerate}
\end{theorem}

We also obtain analogous separations between the randomized and deterministic query complexities using the same techniques (see \Cref{thm:non-seq_DvsR_parallel_separation,thm:no_seq_parallel_sepa_RvsD}).

\subsubsection{Separations with two layers of adaptivity} \label{sec:Intro2Adaptive}
We earlier noted (in \Cref{thm:informalTotalAdaptiveSep}) that for some total function $h$, $\Quant^{3 \perp}(h)$ can be polynomially smaller than $\Rand^{3 \perp}(h)$ (and even $\Rand(h)$), however $\Quant^{1 \perp}(h) = \Theta(\Det^{1 \perp}(h))$ for all total functions $h$. A natural question to ask is whether it is possible for $\Quant^{2 \perp}(h)$ to be much smaller than $\Rand^{2 \perp}(h)$ for some total $h$? We answer this question in the affirmative by presenting a general framework that transforms an (at least) polynomial separation between $\Quant^{1 \perp}(f)$ and $\Rand^{1 \perp}(f)$ for some partial $f$ to a polynomial separation between $\Quant^{2 \perp}(h)$ and $\Rand^{2 \perp}(h)$ for some total $h$ constructed from $f$. Since there exists a partial function $f$ (such as \textsc{Forrelation} defined in \Cref{prob:Forrelation}) with $\Quant(f) = 1$ (so $\Quant^{1 \perp}(f) = 1$) and $\Rand(f) = \Omega(N^{\epsilon})$ (so $\Rand^{1 \perp}(f) = \Omega(N^\epsilon)$) for some $\epsilon > 0$, we get a total function $h$ exhibiting a polynomial separation between $\Quant^{2 \perp}(h)$ and $\Rand^{2 \perp}(h)$. Precisely, we show the following theorem:

\begin{theorem}[Restatement of \Cref{thm:twoAdaptiveRvsQ}]
\label{thm:informal_twoAdaptiveRvsQ}
Let $f: \mathcal{D} \rightarrow \{0,1\}$ with $\mathcal{D} \subseteq \{0,1\}^N$ be a (partial) function that satisfies $\emph{\Quant}^{1 \perp}(f) = \tilde O(N^\epsilon)$ and $\emph{\Rand}^{1 \perp}(f) = \tilde \Omega(N^\delta)$
for some constants $\epsilon < \delta$.
Then there exists a total \(h: \{0,1\}^{M} \rightarrow \{0,1\}\) (with $M = \Theta(N^3)$) and constants $\epsilon' < \delta'$ such that
$\emph{\Quant}^{2 \perp}(h) = \tilde O(M^{\epsilon'})$ and $\emph{\Rand}^{2 \perp}(h) = \tilde \Omega(M^{\delta'})$.
\end{theorem}

We also obtain a similar separation between $\Rand^{2 \perp}$ and $\Det^{2 \perp}$ (see \Cref{thm:twoAdaptiveDvsR}). 

\subsubsection{Quantum parallel lower bound framework and applications}
\label{sec:IntroLowerbound}
Finally, we develop a new method for deriving parallel quantum query lower bounds from sequential quantum query upper bounds. It is simpler to reason about sequential algorithms than parallel algorithms, and upper bounds are often easier to give than lower bounds. Thus, our result allows reducing the hard problem of lower bounding parallel quantum query complexity to the potentially easier problem of upper bounding its sequential analogue. Our technique is based on the parallel spectral adversary method \cite{Jeffery17}, and is especially well suited to problems with optimal adversary matrices that only distinguish input pairs that differ at a single index. We also show that, for any total function $f$, the combinatorial parallel adversary method \cite{Grover04, Burchard19} (see \Cref{thm:parallel_comb_adv}) fails to show a lower bound better than $\Omega\left(\sqrt{\ceil{\frac{\cert_0(f)}{p}}\ceil{\frac{\cert_1(f)}{p}}}\right)$ where $\cert_b$ denotes the $b$th certificate complexity of $f$ (see \Cref{thm:parallel_comb_adv_barrier}). Moreover, we show that our method surpasses this barrier and provides a better lower bound, for instance, to the $\andor$ problem.

We let $\lambda(f)$ denote the spectral sensitivity of $f$, as defined in \Cref{theorem:spectral-sens-equals-adv}. 
Let $\mathcal{F}_{p-\text{res}}^{(f)}$ denote the set of all functions obtained from restricting $f$ to $p$ input bits, and fixing the rest. We call this a $p-$restriction of $f$.

\begin{theorem}[Restatement of \Cref{theorem:parallelLowerBoundNN}]
    For any total Boolean function $f: \{0,1\}^N \rightarrow \{0,1\}$, we have 
    \begin{align*}
        \emph{\Quant}^{p\parallel}(f) = \Omega\left(\lambda(f) \cdot \frac{1}{\max_{g \in \mathcal{F}_{p-\text{res}}^{(f)}} \lambda(g)}\right)
    \end{align*}
    \label{theorem:informalNNLowerBound}
\end{theorem}
\vspace{-1pc}

Recalling that $\lambda(f) = O(\Quant(f))$, as shown by Aaronson et al. \cite{BetterbealsAaronson21}, we can write this as $\Quant^{p \parallel}(f) = \Omega\left({\lambda(f)} \cdot ({\max_{g \in \mathcal{F}_{p-\text{res}}^{(f)}} {\Quant}(g)})^{-1}\right)$, where we note that the second term can be lower bounded by giving a (sequential) quantum algorithm for any $p-$restriction of $f$. We apply \Cref{theorem:informalNNLowerBound} to the following classes of well-studied functions.
\begin{enumerate}
    \item \textit{Read-once formulas (\emph{\Cref{subsec:readOnceformulas}}}): While the combinatorial adversary method \cite{Ambainis02} is sufficient to lower bound the sequential quantum query complexity of the two-layer $\andor$ tree by $\Omega(\sqrt{N})$ \cite{Barnum04}, the best lower bound the parallel combinatorial adversary method can give for the Read-once formula problem $\textsf{And}_{\sqrt{N}} \circ \textsf{Or}_{\sqrt{N}}$ is $\Omega(\sqrt{N}/p)$ (see \Cref{corr:AND_OR_Barrier}). Using \Cref{theorem:informalNNLowerBound}, we show a $\Omega\left(\sqrt{N/p}\right)$ lower bound for any read-once formula (\Cref{theorem:lowerBoundReadOnce}).
    \item \textit{Symmetric functions (\emph{\Cref{subsec:SymmetricFunctions})}}: From \Cref{theorem:informalNNLowerBound}, we get the tight lower bound of $\Omega\left(\sqrt{Nt/p\cdot\min\{t,p\}}\right)$ where $t$ is the largest hamming weight less than $N/2$ such that the function value either differs on inputs of hamming weight $t$ and $t+1$ or differs on inputs of hamming weight $N-t$ and $N-t-1$ (see \Cref{thm:symmetric_functions}). This result recovers the combinatorial adversary lower bound implicit in \cite{Grover04} using an arguably simpler proof.
\end{enumerate}

\subsection{Open questions}
\label{subsec:ConclOpenQuestions}

\textit{A priori}, it seems unintuitive that an exponential advantage in number of rounds is possible for total functions, and also seems like lower bounding a simple function like \(\andor\) in the parallel setting should be possible by standard techniques. Our results in \Cref{sec:IntroSep} and \Cref{sec:IntroLowerbound} show that former is indeed true, and the latter is not the case and in fact requires ``parallel-specific" thinking. Thus, one of our contributions is highlighting that there are still several fundamental query complexity questions unanswered in the parallel setting. We discuss a couple of these:

\paragraph{Limitations of parallel speedup.} Can classical (randomized or deterministic) query algorithms simulate quantum query algorithms when allowed \textit{both} a polynomial overhead in parallelism as well as number of rounds (as opposed to just a polynomial overhead in number of rounds as in \Cref{conj:JefferyTotalFunctions})? We conjecture that this is true:

\begin{conjecture} For all total functions \(f: \{0,1\}^M \rightarrow \{0,1\}\) and parallelism \(p\), we have
\begin{equation*}
    \emph{\Det}^{\poly(p)\parallel} (f) = \poly(\emph{\Quant}^{p\parallel} (f)).
\end{equation*}
\label{conj:NewBeals}
\end{conjecture}
\vspace{-1em}
\begin{itemize}
\item When \(p < \bs(f)^{1-\epsilon}\) for any \( \epsilon > 0 \), \cite{Jeffery17} prove that this is true.
\item When \(p \geq \bs(f) = \Omega(M^\epsilon)\) for any \(\epsilon > 0 \), clearly this is true. This is the case in our result as well, where the canonical cheat sheet function (\Cref{def:Canonicalcheatsheet}) used to show \Cref{thm:informalTotalAdaptiveSep} has \(\bs(f) = \tilde \Omega(M^{1/6})\) and $\Det^{p^{3/2}\parallel} = O(\Quant^{p\parallel}(f)^3)$ for any \(p\).
\item Thus, the conjecture must be proven/falsified in the regime where \(\bs\) and \(p\) are subpolynomial but superconstant. Moreover, what is the smallest power of \(p\) required to show this for all total \(f\)? Analogous questions remain open for $\Det^{p\parallel}$ vs $\Rand^{p \parallel}$ as well.
\end{itemize}

\paragraph{Parallel composition theorem.} It is known that $\Quant(f \circ g) = \Theta(\Quant(f) \cdot \Quant(g))$ for boolean decision functions, a widely applicable and powerful result \cite{Reichardt11b}. Is there an analogous theorem for $p$-parallel query complexity? In particular, do we have \begin{align*}
    \Quant^{p\Vert}(f \circ g) =& \Theta\left(\min_{q \in [p]} \Quant^{q\Vert}(f) \cdot \Quant^{\lceil p/q \rceil \Vert}(g)\right)
\end{align*}
for boolean $f, g$? Such a result would suffice to reproduce our lower bound for $\andor$, and likely be widely applicable for understanding parallel quantum query complexity. A straightforward attempt to show composition of the $p$-parallel adversary quantity using the reduction in \cite{Grover04} fails. This is because the partial function $f'$ whose sequential query complexity characterizes the $p$-parallel complexity of $f$ is no longer boolean.

\paragraph{Natural functions giving unbounded separation in rounds.} Our results in \Cref{sec:IntroSep} and \Cref{sec:IntroGenuineSep} answer affirmatively whether it is possible for total functions to have unbounded separations in rounds, and whether this is still possible when no sequential separation exists. However, finding more natural functions to answer both these questions remains open and could help in understanding the interplay between quantum algorithms and parallel queries better.

\subsection{Related work}

Prior work on parallel quantum algorithms in the query model has primarily demonstrated the need for quantum depth in solving certain problems.

\vspace{-0.5pc}

\paragraph{Cryptography}
There has been a recent line of work studying problems which can only be efficiently solved by quantum algorithms with high depth \cite{Chia23, arora23depth, chia2022classical, chia23noninteractive, coudron20depth}, constructing cryptographic proofs of quantum depth. In this vein, Chung et al \cite{Chung21} and Blocki et al \cite{blocki21posw} show parallel quantum lower bounds against producing the output of an iterated hash function to give a cryptographic proof of sequential work that is secure against quantum computers.\vspace{-0.5pc}

\paragraph{Partial functions}

Burchard \cite{Burchard19} gives a characterization of parallel quantum approximate counting, showing that, approximating to multiplicative error $\epsilon$, the number $K$ of marked elements in a list of size $N$ requires $\Omega\left(\frac{\sqrt{N}}{\epsilon \sqrt{pK}}\right)$ $p$-parallel quantum queries. Girish et al \cite{girish2023power} constructs a partial function which exhibits a large total query separation between $r$-adaptive and $r-1$-adaptive quantum algorithms for any constant $r$. 
\vspace{-0.5pc}

\paragraph{Circuit complexity}
Another setting one could study the intersection of quantum and parallelism is in the circuit model. For instance, Cleve and Watrous \cite{cleve00parallel} show how to implement the quantum fourier transform on $n$ qubits in $O(\log n)$ depth, allowing for a highly parallel implementation of Shor's factoring algorithm. There are known examples of relational problems which can be computed in constant quantum depth, yet require at least logarithmic depth for classical circuits \cite{BGK18, watts19unbounded, watts2023unconditional}. More generally, quantum circuit classes such as $\textsf{QAC}_n$  and $\textsf{QNC}_n$ have been extensively studied \cite{green02circuits, hoyer05fanout, takahashi16collapse}. 
\vspace{-0.5pc}

\subsection{Organization}

Our paper consists of 4 main results, one corresponding to each point in the abstract. For each of our results: 
\begin{itemize}
\item The theorem statement is given in the Main Results subsection (\Cref{subsec:MainResults}), 
\item An intuitive explanation of the technical aspects of each result is given in the Technical Overview section (\Cref{sec:technicalOverview})
\item The required definitions are given, the main theorem is restated, more detailed explanations and the lemmas and proofs for each result are given in their respective sections (\Cref{sec:cheatsheet,sec:unbounded_genuine_separations,sec:2Adaptive,sec:LowerBounds}).
\end{itemize}

We provide links to each of the above for convenience:

\vspace{10pt}

    \begin{tabular}{|p{5cm}|c|c|c|}
    \hline
        \textbf{Result Name} & \textbf{Result Statement} & \textbf{Technical Overview} & \textbf{Full Details}\\ \hline \hline
        {\small Unbounded parallel separations for total functions} & \Cref{sec:IntroSep} & \Cref{sec:technicalOverviewUnboundedSep} & \Cref{sec:cheatsheet} \\ \hline
        {\small Unbounded genuine parallel separations } & \Cref{sec:IntroGenuineSep}& \Cref{sec:TechnicalOverviewGenuinePartial} & \Cref{sec:unbounded_genuine_separations}\\ \hline
         {\small Separations with two layers of adaptivity}& \Cref{sec:Intro2Adaptive} & \Cref{sec:TechnicalOverviewSeparationsWithTwoLayers} &  \Cref{sec:2Adaptive} \\ \hline
         {\small Parallel quantum lower bound framework and applications}& \Cref{sec:IntroLowerbound} & \Cref{sec:TechnicalOverviewLowerBounds} & \Cref{sec:LowerBounds}\\ \hline
    \end{tabular}

\vspace{10pt}

Preliminaries are provided in \Cref{section:preliminaries} and  auxiliary proofs are provided in \Cref{sec:AppendixCheatsheetFramework,sec:AppendixCompositionTheorem,sec:AppendixPointer-chasing-proofs}

\section{Technical overview}\label{sec:technicalOverview}

In this section, we outline the techniques used in showing our main results. 
For our results involving separations between query complexity measures, we will primarily describe the techniques involved in separating quantum from randomized query complexity measures unless the techniques for separating randomized and deterministic query complexity measures are significantly different.

\addtocontents{toc}{\protect\setcounter{tocdepth}{0}} 
\subsection{Unbounded parallel separations for total functions}
\addtocontents{toc}{\protect\setcounter{tocdepth}{2}} 
\label{sec:technicalOverviewUnboundedSep}

\begin{figure}[H]
\begin{center}
\begin{tikzpicture}[scale=0.47, every node/.style={transform shape}]
\tikzstyle{block} = [rectangle, draw, minimum width = 4cm, minimum height = 1cm]
\tikzstyle{data} = [rectangle, draw, minimum width = 4cm, minimum height = 1cm]
\tikzstyle{patterned_data} = [rectangle, draw, minimum width = 4cm, minimum height = 1cm, pattern = north east lines]

\foreach \n in {0, 4, 8}{
    \node[block] (a\n) at (\n, 0) {};
    \node (bkk\n) at (\n, -3) {\Large \(f \circ \andor\)};
    \draw (a\n.south east) -- (bkk\n);
    \draw (a\n.south west) -- (bkk\n);
}

\foreach \n in {4}{
\node [block] (simons\n) at (\n, -6) {};

    \pgfmathsetmacro{\previous}{int(\n-4)}
    \pgfmathsetmacro{\next}{int(\n+4)}
    
    \draw (bkk\previous) -- (simons\n);
    \draw (bkk\n) -- (simons\n);
    \draw (bkk\next) -- (simons\n);
}

\foreach \n in {11, 15, ..., 30}{
    \node[block] (data\n) at (\n+3, 0){};
}
\node[patterned_data] (pat_data) at (23+3, 0) {};
\node (cert) at (23+3, -1) {\LARGE Certificate};

\draw [
    decoration={
        brace,
        raise=0.5cm,
        amplitude=15pt
    },
    decorate,
] (a0.west) -- (a8.east);

\draw [
    decoration={
        brace,
        raise=0.5cm,
        amplitude=15pt
    },
    decorate,
] (data11.west) -- (data27.east);

\node (data_tag) at (17+ 3 + 2, 2.7) {\LARGE Data};
\node (add_tag) at (4 , 2.7) {\LARGE Address};

\draw [->] (simons4) -- (23+3, -6) -- (cert);

\end{tikzpicture}
\end{center}
\caption{The canonical cheat sheet function \(f_{ccs}\) which lifts a partial function \(f\) to a total function, while retaining some of the speedup of \(f\).}
\label{fig:unboundedSpeedup}
\end{figure}
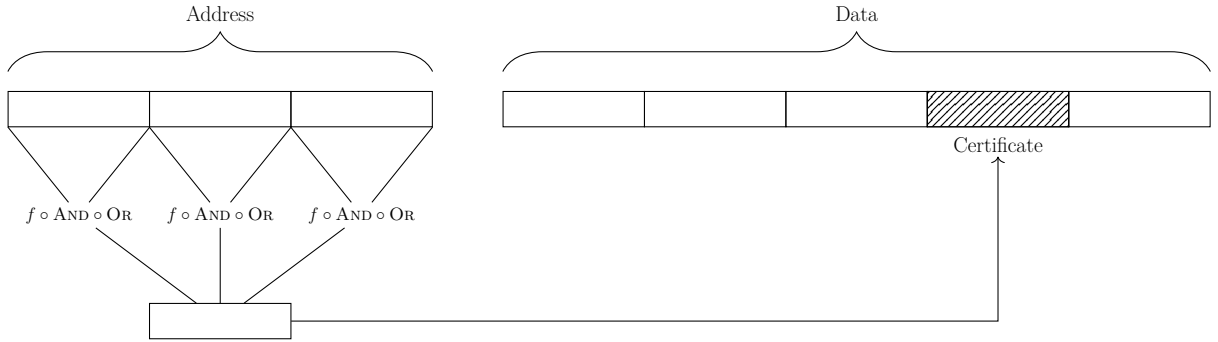

The cheatsheet framework \cite{Aaronson16} lifts a partial function \(f\), potentially with an exponential quantum query advantage to a total function that retains some of the advantage (although now polynomial). The natural way to do this is to define a total function such that when the input is in the domain of \(f\), answer as \(f\) does, and when it is not answer \(0\). The cheat sheet framework does this ``domain check" in a clever way so that the quantum advantage is not lost in the process. As shown in \Cref{fig:unboundedSpeedup}, the input for the total function \(f_\canonicalcheatsheet\) is divided into two components, the Address section and the Data section. The partial function \(f\) is composed with a total function that has low certificate complexity (defined in \Cref{section:preliminaries}), which in this case is \(\andor\). This composition \(f \circ \andor\) produces a single bit. This is repeated such that enough bits to address a block in the data component are obtained. If this block contains the certificate (or ``cheatsheet") for the \(\andor\)'s, and these certify that the input of \(f\) is in its domain, then the  output of \(f_\canonicalcheatsheet\) is \(1\), else, the output is \(0\). As the data component is large, an algorithm is forced to solve the partial function \(f\) to obtain the address of the block, and can use the certificate present there to check whether the input of \(f\) is in the domain (and thus does not need to read the full Address to perform the domain check).

Thus, computing a cheatsheet function involves 
\begin{enumerate}[label=(\arabic*)]
    \item Solving \(f \circ \andor\) to get the certificate location,
    \item Reading out the certificate,
    \item Checking the validity of the input of \(f\).
\end{enumerate}

For concreteness, consider the canonical cheatsheet function defined by Aaronson et al \cite{Aaronson16} composing \(f = \textsc{Forrelation}\) \cite{Aaronson10, Aaronson15} on $N$ bits with $\andor$ on $N^2$ bits. This function was originally constructed to exhibit a superquadratic separation between sequential quantum and randomized query complexities. Observe that with $p=N^2$ parallelism, solving any given instance of the internal $\andor$ can be done in a single $p$-parallel query by reading out all the inputs. A quantum algorithm can utilize this to quickly solve the whole composition, as Forrelation is quantumly easy. However, any classical algorithm would seem to need $\Rand(\textsc{Forrelation})$ $p$-parallel queries to solve this composition since it requires at least $p\cdot \Rand(\textsc{Forrelation})$ sequential queries. Thus, with $p$ parallelism, step (1) is quantum efficient but not classically efficient. We show that for the canonical choice of parameters, steps (2) and (3) are also possible with few $p$-parallel queries, resulting in the desired separation. 

The same technique can be used to give an exponential total function separation between randomized and deterministic parallel query complexity.

Our result, combined with an aforementioned result of \cite{Jeffery17}, which states that $\Det^{p\parallel}(h)=\poly(\Quant^{p\parallel}(h))$ for all total $h$ and $p = O(\bs(h)^{1-\epsilon})$ with $\epsilon > 0$, provides a new way of upper bounding the block sensitivity of total Boolean functions. In particular, using our separation and a lower bound argument for the canonical cheatsheet function, we characterize its block sensitivity.

\addtocontents{toc}{\protect\setcounter{tocdepth}{0}} 
\subsection{Unbounded genuine parallel separations}
\addtocontents{toc}{\protect\setcounter{tocdepth}{2}} 
\label{sec:TechnicalOverviewGenuinePartial}

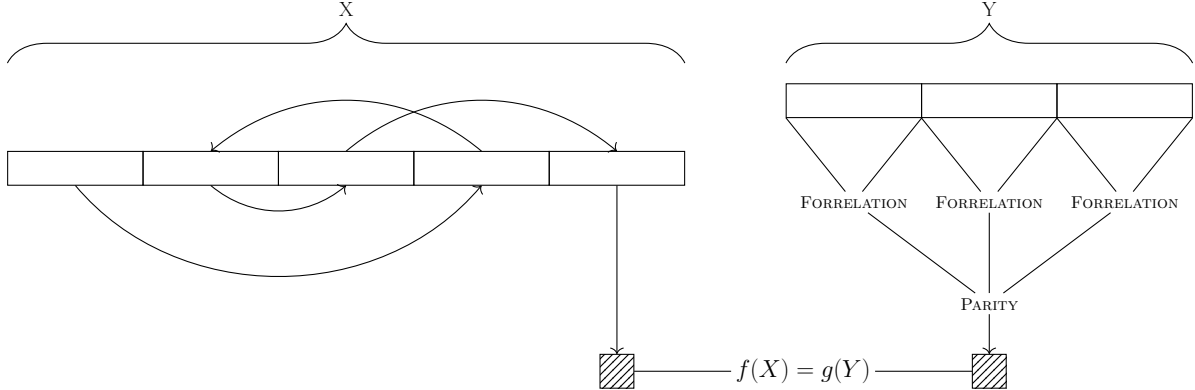
\begin{figure}[H]
\begin{center}
\begin{tikzpicture}[scale=0.45, every node/.style={transform shape}]
\tikzstyle{block} = [rectangle, draw, minimum width = 4cm, minimum height = 1cm]
\tikzstyle{data} = [rectangle, draw, minimum width = 4cm, minimum height = 1cm]
\tikzstyle{small_block} = [rectangle, draw, minimum width = 1cm, minimum height = 1cm, pattern = north east lines]

\foreach \n in {31, 35, 39}{
    \node[block] (a\n) at (\n, 0) {};
    \node (bkk\n) at (\n, -3) {\Large \(\textsc{Forrelation}\)};
    \draw (a\n.south east) -- (bkk\n);
    \draw (a\n.south west) -- (bkk\n);
}

\foreach \n in {35}{
\node (simons\n) at (\n, -6) {\Large \(\textsc{Parity}\)};
\node[small_block] (g) at (\n, -8) {};
    \pgfmathsetmacro{\previous}{int(\n-4)}
    \pgfmathsetmacro{\next}{int(\n+4)}
    
    \draw (bkk\previous) -- (simons\n);
    \draw (bkk\n) -- (simons\n);
    \draw (bkk\next) -- (simons\n);
}

\foreach \n in {5, 9, ..., 24}{
    \node[block] (data\n) at (\n+3, -2){};
}

\draw [->] (data5.south) to [out=-50,in=-130] (data17.south);
\draw [->] (data17.north) to [out=140,in=40] (data9.north);
\draw [->] (data9.south) to [out=-40,in=-140] (data13.south);
\draw [->] (data13.north) to [out=40,in=140] (data21.north);

\node[small_block] (f) at (24, -8) {};

\draw[->] (data21.south) -- (f.north);
\draw[->] (simons35.south) -- (g.north);

\draw [
    decoration={
        brace,
        raise=0.5cm,
        amplitude=15pt
    },
    decorate,
] (a31.west) -- (a39.east);

\draw [
    decoration={
        brace,
        raise=1.4cm,
        amplitude=15pt
    },
    decorate,
] (data5.west) -- (data21.east);

\node (data_tag) at (16, 2.7) {\LARGE X};
\node (add_tag) at (35 , 2.7) {\LARGE Y};

\node (equality) at (29.5, -8) {\huge \(f(X) = g(Y)\)};
\draw (f) -- (equality);
\draw (g) -- (equality);

\end{tikzpicture}
\end{center}
\caption{Partial function that admits unbounded genuine parallel separation}
\label{fig:ANA}
\end{figure}

Is some polynomial sequential separation necessary for unbounded parallel quantum speedup, or can the advantage emerge purely from a quantum algorithm's ability to utilize parallelism (which we call a genuine parallel advantage)? We argue that the latter can indeed be the case, by first describing our construction for the partial function $h$ that allows proving \Cref{thm:partialexpsep}. As shown in \Cref{fig:ANA}, the input to \(h\) here involves two components \(X\) and \(Y\). The function $f$ requires a sequence of adaptive queries to solve (in either model) whereas the function $g$ achieves quantum-randomized separation and can be fully parallelized. We are promised that the output of both will be the same, so any algorithm can choose to solve either. We then argue that any algorithm will require solving at least one of $f$ and $g$.

For concreteness, let $\epsilon > 0$ and suppose that $\Quant(f) = \Theta(\Rand(f)) = \Theta(N^{\epsilon})$ and $\Rand^{p \parallel}(f) = \Omega(N^{\epsilon})$, and $\Quant^{p \parallel}(g) = \Theta(N^{\epsilon})/p$ and $\Rand^{p \parallel}(g) = \Theta(N^{2\epsilon})/p$ for small enough $p$\footnote{For more concreteness, one may think of $f$ as the pointer chasing function (see \Cref{prob:pointer-chasing}) with chain length $N^\epsilon$ and $g$ as the composition of parity on $N^\epsilon$ bits composed with \textsc{Forrelation} on $N^{2\epsilon}$ bits as shown in \Cref{fig:ANA}.}. Then, a randomized sequential algorithm can choose to solve $f$, while a quantum sequential algorithm will not benefit from choosing to solve either $f$ or $g$. Thus, $\Rand(h) = O(N^\epsilon)$ and we would have $\Quant(h) = \Omega(N^\epsilon)$. Similarly, for $p = \Quant(g)$, a $p$-parallel quantum algorithm can choose to solve $g$, while a $p$-parallel randomized algorithm will not benefit from choosing to solve either $f$ or $g$. Therefore, $\Quant^{p \parallel}(h) = 1$ and we would have $\Rand^{p \parallel} = \Omega(N^\epsilon)$. 

To show our desired randomized parallel and quantum lower bounds, we consider the general problem, which we call $\correlated$, where we are given inputs $X$ and $Y$ to arbitrary functions $f$ and $g$ respectively along with the promise that $f(X) = g(Y)$, our goal is to output $f(X)=g(Y)$ (see \Cref{prob:correlated}). We use the hybrid argument to establish that any randomized parallel algorithm that solves $\correlated(f,g)$ will be able to either distinguish an input sampled from a hard $0$-distribution for $f$ from an input sampled from a hard $1$-distribution for $f$, or will succeed at a similar distinguishing task for $g$. Therefore, we must have $\Rand^{p \parallel}(\correlated(f,g)) = \Omega(\min(\Rand^{p \parallel}(f),\Rand^{p \parallel}(g)))$. For the quantum lower bound, we show that that for any adversary matrices $\Gamma^{(f)}$ and $\Gamma^{(g)}$ for $f$ and $g$ respectively, the matrix $\Gamma = \Gamma^{(f)} \otimes \Gamma^{(g)}$ is an adversary matrix for $\correlated(f,g)$ and $\max_i \norm{\Gamma_i} = \max\left(\max_i \norm{\Gamma^{(f)}_i}, \max_i \norm{\Gamma^{(g)}_i}\right)$\footnote{All the $\max$ are over relevant indices.} (see \Cref{subsection:adv_method_parallel} and \Cref{lem:quant_correlated_functions}). It follows that $\Adv\left(\correlated(f,g)) = \Omega(\min(\Adv(f), \Adv(g))\right)$, which implies the desired lower bound (see \cref{equation:parallelSpectralAdversary} and \Cref{theorem:parAdv}). 

Next, we describe a way to totalize the partial function $h$ that we constructed in the previous section while maintaining some of its properties. We will use the cheatsheet framework. However, we will need a function with relatively small certificate complexity that does not admit any significant quantum speed-up so $\andor$ will not work. Fortunately, $\cite{Aaronson16}$ found a function, which they called $\bkk$ (see \Cref{prob:bkk}), that satisfies $\Quant(\bkk_{N}) = \tilde \Theta(\bkk_{N}) = \tilde \Theta(N)$ and $\cert(\bkk_{N}) = \tilde{O}(\sqrt{N})$. We compose $h$ on $N$ bits with $\bkk$ on $N^2$ bits, and then plug it into the cheatsheet framework. The desired upper bounds for $\Quant^{p \parallel}$ and $\Rand$ in \Cref{thm:informalTotalExpSep} are relatively straightforward and follows from the discussion in the previous sections. The quantum lower bound follows from quantum query complexity composition theorem and results in \cite{Aaronson16}. For the randomized parallel lower bound, we show a composition theorem where the inner function is $\bkk$ (see \Cref{thm:rand_comp}).

The same lower bound techniques do not follow for the analogous separation between the randomized and deterministic query complexity measures (see \Cref{thm:no_seq_parallel_sepa_RvsD}). In particular, there is no general randomized composition theorem. Fortunately, for our constructions, the known randomized composition theorems suffice (see \Cref{thm:QueryComplexityComposition}). For the deterministic parallel lower bound, we show a general composition theorem when the degree of the inner function is almost full, which might be of independent interest.    

\addtocontents{toc}{\protect\setcounter{tocdepth}{0}} 
\subsection{Separations with two layers of adaptivity}
\addtocontents{toc}{\protect\setcounter{tocdepth}{2}} 
\label{sec:TechnicalOverviewSeparationsWithTwoLayers}

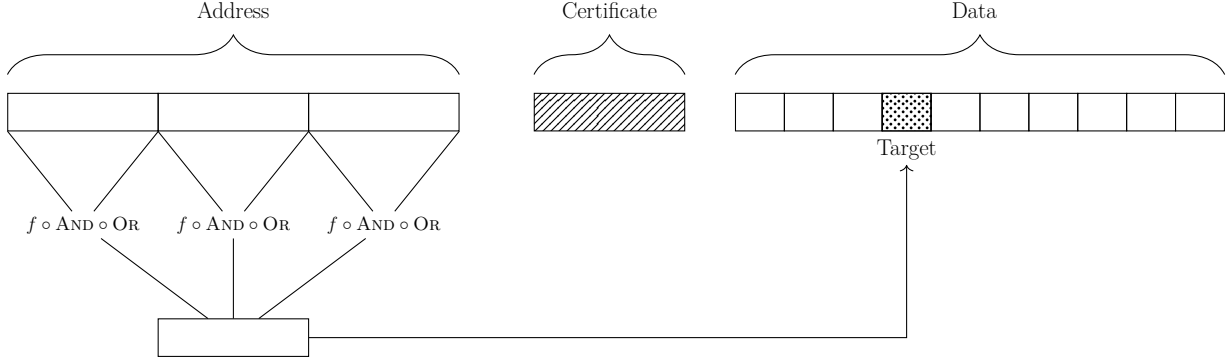
\begin{figure}[H]
\begin{center}
\begin{tikzpicture}[scale=0.5, every node/.style={transform shape}]
\tikzstyle{block} = [rectangle, draw, minimum width = 4cm, minimum height = 1cm]
\tikzstyle{data} = [rectangle, draw, minimum width = 4cm, minimum height = 1cm]
\tikzstyle{small_block} = [rectangle, draw, minimum width = 1.3cm, minimum height = 1cm]
\tikzstyle{patterned_data} = [rectangle, draw, minimum width = 4cm, minimum height = 1cm, pattern = north east lines]
\tikzstyle{patterned_small_block} = [rectangle, draw, minimum width = 1.3cm, minimum height = 1cm, pattern = crosshatch dots]

\foreach \n in {0, 4, 8}{
    \node[block] (a\n) at (\n, 0) {};
    \node (bkk\n) at (\n, -3) {\Large \(f \circ \andor\)};
    \draw (a\n.south east) -- (bkk\n);
    \draw (a\n.south west) -- (bkk\n);
}

\foreach \n in {4}{
    \node [block] (simons\n) at (\n, -6) {};

    \pgfmathsetmacro{\previous}{int(\n-4)}
    \pgfmathsetmacro{\next}{int(\n+4)}
    
    \draw (bkk\previous) -- (simons\n);
    \draw (bkk\n) -- (simons\n);
    \draw (bkk\next) -- (simons\n);
}

\node (below_brace) at (18.3, -9.3){};

\foreach \n in {0, 1, ..., 9}{
    \node[small_block] (data\n) at (15+3 +1.3*\n, 0){};
}
\node[patterned_data] (pat_data) at (11+3, 0) {};

\node[patterned_small_block] (pat_block) at (18.9+3, 0) {};
\node (target) at (18.9+3, -1) { \LARGE Target};

\draw [->] (simons4) -- (18.9+3, -6) -- (target);

\draw [
    decoration={
        brace,
        raise=0.5cm,
        amplitude=15pt
    },
    decorate,
] (a0.west) -- (a8.east);

\draw [
    decoration={
        brace,
        raise=0.5cm,
        amplitude=15pt
    },
    decorate,
] (pat_data.west) -- (pat_data.east);

\draw [
    decoration={
        brace,
        raise=0.5cm,
        amplitude=15pt
    },
    decorate,
] (data0.west) -- (data9.east);

\node (data_tag) at (15+ 3 + 4*1.30 + 0.5, 2.7) {\LARGE Data};
\node (cert_tag) at (11 + 3 , 2.7) {\LARGE Certificate};
\node (add_tag) at (4 , 2.7) {\LARGE Address};

\end{tikzpicture}
\end{center}
\caption{A framework that converts a partial function separation with one layer of adaptivity to a total function with two layers of adaptivity.}
\label{fig:TwoAdaptive}
\end{figure}

As we noted earlier (see \Cref{thm:informalTotalAdaptiveSep}), the cheat sheet framework can be employed to show a polynomial separation between \(\Quant^{3 \perp}\) and \(\Rand^{3 \perp}\). We modify this framework to show a polynomial separation between \(\Quant^{2 \perp}\) and \(\Rand^{2 \perp}\). As shown in \Cref{fig:TwoAdaptive}, our strategy is essentially to separate the cheatsheet into two parts: the first part contains the certificate, and the second just contains a long data string with one relevant index. Let $f$ be a partial function with a polynomial separation between \(\Quant^{1 \perp}(f)\) and \(\Rand^{1 \perp}(f)\). Informally, embedding $f$ in our framework results in a function that requires (1) verifying the input to the $f$ is in the domain (using the certificate) and if so, (2) outputting the bit of the data string present at the address obtained by solving the $f$.

For the formal problem description, see \Cref{prob:twoAdaptiveFunction}.

The quantum algorithm easily succeeds as follows. In the first round, it can read the certificate as well as solve $f$ to obtain the address. In the second round, it verifies the certificate and queries the bit at the right address in the data string, hence is able to output the result. The randomized algorithm is unable to succeed in the same amount of parallelism because after the first round, it cannot compute $f$, so it would not find the right address by the first round, and can only make a query to the right address by the second round if it guesses the location correctly, which happens with very low probability. We also note that there is a caveat: the structure of a zero-certificate might be easy to distinguish from the structure of a one-certificate, allowing the algorithm to infer the inputs to the partial function from just the structure of the certificates. In that case, the randomized algorithm can use this to compute the partial function in the first round itself. To prevent this from happening, we use ``bi-certificates" instead of certificates, where a ``bi-certificate" corresponds to a set of indices that could certify both a zero and a one instance. This prevents the randomized algorithm from learning any information just by knowing the certificate.

Therefore, \Cref{prob:twoAdaptiveFunction} presents a framework to lift a partial function \(\Quant^{1 \perp}\) vs \(\Rand^{1 \perp}\) separation to a total function \(\Quant^{2 \perp} \) vs \(\Rand^{2 \perp}\) separation (see \Cref{thm:twoAdaptiveRvsQ}). 
The analogous result holds for randomized vs deterministic algorithms (see \Cref{thm:twoAdaptiveDvsR}).

\addtocontents{toc}{\protect\setcounter{tocdepth}{0}} 
\subsection{Quantum parallel lower bound framework and applications}
\addtocontents{toc}{\protect\setcounter{tocdepth}{2}} 
\label{sec:TechnicalOverviewLowerBounds}

There are few known techniques for lower bounding parallel quantum query complexity. Consider for instance lower bounding the quantum parallel query complexity for the balanced \(\andor\) function, 

which is an example of a read-once formula. It is well known that this function has sequential quantum query complexity $\Theta(\sqrt{N})$ \cite{Barnum04,Reichardt11}. Since both $\Quant^{p \parallel}(\textsc{And})$ and $\Quant^{p \parallel}(\textsc{Or})$ are \(\Omega \left (\sqrt{N/p}\right ) \), a natural guess for $\Quant^{p \parallel}(\andor)$ is $\Omega\left(\sqrt{N/p}\right)$. 

The bound $\Quant(\andor) = \Omega(\sqrt{N})$ was shown using the (sequential) combinatorial adversary method. However, as mentioned in \Cref{sec:IntroLowerbound}, the corresponding parallel version fails to show a lower bound better than $\Omega \left (\sqrt { \ceil{\frac {\cert_0(f)}{p}} \ceil{\frac {\cert_1(f)}{p}}} \right )$ for any function $f$ (see \Cref{thm:parallel_comb_adv_barrier}). This means that the best lower bound this method could help prove for \(\andor\) is $\Omega(\sqrt{N}/p)$, since $\cert_0(\andor) = \cert_1(\andor) = \sqrt{N}$.

We use the parallel spectral adversary method of \cite{Jeffery17} (see \Cref{theorem:parAdv}), which is known to be optimal, to derive a method that is easier to apply and is sometimes stronger than the combinatorial adversary method (see \Cref{theorem:informalNNLowerBound}). For the case of read-once formulas, we use the adversary sets of \cite{Barnum04} to construct an adversary matrix $\Gamma$. It is easy to lower bound $\norm{\Gamma}$ by a simple counting argument, but upper bounding $\norm{\Gamma_S}$ for all $S \subseteq [N]$ with $|S| = p$ can be challenging. We show that if \(\Gamma\) satisfies the property that $\Gamma[x,y] = 0$ for all $x, y$ that differs in more than 1 bit, then all the induced \(\Gamma_S\) can be rearranged to form block-diagonal matrices with blocks of size \(2^p \times 2^p\). Moreover, each of these blocks are adversary matrices for some restricted function \(g \in \mathcal{F}_{p-\text{res}}^{(f)}\) (\Cref{fig:lower-bound-intuition}), where a restricted version of $f$ is one where all but $p$ input bits are fixed and known by the algorithm. Thus, we know that the spectral norm of these blocks is upper bounded by \(\Quant(g)\) (up to the normalizing factor of $\max_{i \in [2^p]}\norm{\Gamma_i}$, and with a max taken over all $g \in \mathcal{F}_{p-\text{res}}^{(f)}$), which we can use to upper bound the spectral norm of any \(\Gamma_S\). In our case, noting that \(g\) is a read-once formula of size $p$, we have \(Q(g) = O(\sqrt p)\) and we obtain the desired lower bound of \(\Omega \left ( \sqrt{N/p}\right )\). Hence, we are able to reduce the task of finding parallel lower bounds to the potentially easier task of finding sequential upper bounds.

\begin{figure}[H]
    \centering
    \begin{tikzpicture}[scale=0.5, every node/.style={transform shape}]
    \tikzstyle{GammaF} = [rectangle, draw, minimum width = 15cm, minimum height = 15cm]
    \tikzstyle{GammaG} = [rectangle, draw, minimum width = 3cm, minimum height = 3cm, pattern = north east lines]
    
    \node (label) at (7.5, 16) {\huge \(\Gamma(f)\)};
    \node [GammaF] (GammaF) at (7.5, 7.5){};
    \node [GammaG] (GammaG0) at (1.5+0*3, 1.5+4*3) {};
    \node [GammaG] (GammaG1) at (1.5+1*3, 1.5+3*3) {};
    \node [GammaG] (GammaG2) at (1.5+2*3, 1.5+2*3) {};
    \node [GammaG]  (GammaG3) at (1.5 + 3*3, 1.5+1*3) {};
    \node [GammaG] (GammaGn) at (1.5+4*3, 1.5+0*3) {};
    \node[fill = white] at (GammaG0.center) {\huge \(\Gamma(g_0)\)};
    \node[fill = white] at (GammaG1.center) {\huge \(\Gamma(g_1)\)};
    \node[fill = white] at (GammaG2.center) {\huge \(\Gamma(g_2)\)};
    \node[fill = white] at (GammaG3.center) {\Huge \rotatebox{-25}{\scalebox{1.3}{\(\ddots\)}} };
    \node[fill = white] at (GammaGn.center) {\huge \(\Gamma(g_n)\)};
    
    \end{tikzpicture}
    \caption{Nearest neighbor adversaries can be block-diagonalized, where the blocks are adversary matrices for restricted versions of the function $f$.}
    \label{fig:lower-bound-intuition}
\end{figure}
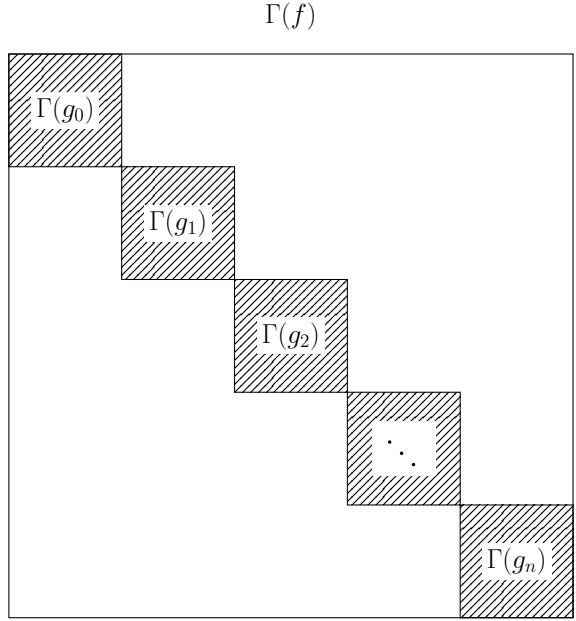

\section{Preliminaries} \label{section:preliminaries}

In this section, we will review some of the notions related to quantum query complexity, adversarial methods and the cheat sheet framework of \cite{Aaronson16}.

\subsection{Notation}
For integers $N_1 < N_2$, we use $[N_1,N_2]$ to denote the set $\{N_1, N_1+1,\dots, N_2\}$. We also use $[N]$ for $[1,N]$ for simplicity. For a set $S \subseteq [N]$, we will use $S^\complement$ to denote the set $[N] \setminus S$. For strings $x,y$, we use $x \parallel y$ to denote concatenation of $x$ with $y$, $x[i]$ to denote  the character at the $i$th position in $x$ and $x[i:]$ to denote the substring of $x$ starting from the $i$th position. For functions $f, g$, we use $f \circ g$ to denote the composition of $f$ with $g$, $\Tilde{O}(f) = O(f \cdot \polylog(f))$, $\Tilde{\Omega}(f) = \Omega(f/\polylog(f))$ and $\Tilde{\Theta}(f) = g$ iff $g = \Tilde{O}(f)$ and $g = \Tilde{\Omega}(f)$. Note that our use of $\Tilde O$ and $\Tilde \Omega$ are such that logarithmic factors in all asymptotic variables are suppressed, unless otherwise stated. We use $\mathbb{I}[\text{prop}]$ to denote the indicator variable which is $1$ if prop is true, and $0$ otherwise.

\subsection{Boolean complexity measures and query complexity}

Let $f: \mathcal{D} \rightarrow \{0,1\}$ be a Boolean function with domain $\mathcal{D} \subset \{0,1\}^N$. We say that $f$ is total if $\mathcal{D} = \{0,1\}^N$; otherwise, we say it is partial. Sometimes, instead of considering a Boolean alphabet $\{0,1\}$, we will consider larger alphabets of size up to $N^{\poly \log N}$. However, such functions always correspond to functions with Boolean alphabet, up to logarithmic overhead factors, in the relevant complexity measures by a straightforward reduction. 

We recall the combinatorial notions of certificate complexity and block sensitivity for total Boolean functions. A certificate for $f$ on input $x \in \{0,1\}^n$ is a set $S \subseteq [n]$ such that $f(y) = f(x)$ for all $y \in \{0,1\}^n$ with $y_S = x_S$. Let $\cert_x(f)$ denote the size of the smallest certificate for $f$ on input $x$. The certificate complexity of $f$, denoted $\cert(f)$, is $\max_{x \in \{0,1\}^n} \cert_x(f)$.

For any $x \in \{0,1\}^n$ and $S \subseteq [n]$, let $x^{(S)}$ be the string $x$ with the characters at positions in $S$ flipped. For a given input $x \in \{0,1\}^n$ of $f$, call a subset of indices $S$ a sensitive block if $f(x) \neq f(x^{(S)})$. Let $\bs_x(f)$ be the maximum number of disjoint sensitive blocks. The block sensitivity of $f$, denoted $\bs(f)$, is $\max_{x \in \{0,1\}^n} \bs_x(f)$.

A $p$-parallel classical query to $x$ maps $p$ bit strings $i_1, i_2, \dots, i_p \in \{0,1\}^p$ to $x_{i_1}, x_{i_2}, \dots, x_{i_p}$. Similarly, a $p$-parallel quantum oracle $\oracle_x^{p\parallel}$ is defined as
\begin{align*}
    \oracle_x^{p\parallel} \ket{i_1, i_2, \dots, i_p} \ket{b_1, b_2, \dots, b_p} \rightarrow \ket{i_1, i_2, \dots, i_p} \ket{b_1 \oplus x_{i_1}, b_2 \oplus x_{i_2}, \dots, b_p \oplus x_{i_p}}
\end{align*}
We usually call a $1$-parallel (quantum) query as a sequential (quantum) query.

A $p$-parallel quantum algorithm is a sequence of unitary operations $U_i$ interleaved with $p$-parallel queries to an input oracle $\oracle_x^{p\parallel}$ with an unbounded number of ancilla qubits. Any such algorithm $\mathcal{A}$ computes $f$ if upon measuring the first qubit of the final state of $\mathcal{A}$, one receives $f(x)$ with probability at least 2/3 for all $x \in \{0,1\}^N$. The notions of $p$-parallel deterministic and randomized algorithms are defined similarly.

The $p$-parallel deterministic (respectively randomized, quantum) query complexity of a function $f$, denoted $\Det^{p \parallel}(f)$ (respectively $\Rand^{p \parallel}(f)$, $\Quant^{p \parallel}(f)$), is the minimum number of classical (respectively classical, quantum) queries a $p$-parallel algorithm makes to compute $f$. We will usually refer to $\Det^{1 \parallel}(f)$ (respectively $\Rand^{1 \parallel}(f)$, $\Quant^{1 \parallel}(f)$) as the sequential query complexity of $f$ and denote by $\Det(f)$ (respectively $\Rand(f)$, $\Quant(f)$). 

An orthogonal notion of complexity when queries are made in parallel is what we call $k$-adaptive query complexity. The $k$-adaptive deterministic (respectively randomized, quantum) query complexity for a function $f$, denoted $\Det^{k \perp}(f)$ (respectively $\Rand^{k \perp}(f)$, $\Quant^{k \perp}(f)$) is the minimum $p$ such that there is an algorithm $\mathcal{A}$ that computes $f$ using $k$ many $p$-parallel deterministic (respectively randomized, quantum) queries. We will usually refer to $\Det^{1 \perp}(f)$ (respectively $\Rand^{1 \perp}(f)$, $\Quant^{1 \perp}(f)$) as non-adpative query complexity.

We use \(\bs(f)\) to denote the block sensitivity of a function and \(bs^f(X)\) for the block sensitivity of an input \(X\) for the function \(f\) (see \cite{Buhrman2002Complexitymeasures} for formal definition). Then, \(\bs_0\) and \(\bs_1\)  refer to \(\max \{ \bs^f(X) \; | \; f(X) = 0\}\) and \(\max \{ \bs^f(X) \; | \; f(X) = 1\}\) respectively.

\subsection{Commonly used functions}
\label{sec:CommonlyUsedFunctions}
In this section, we recall some functions that we will use in our results. 

We begin with the $\andor$ function, which is a commonly used function due to its balanced $0$-certificate and $1$-certificate sizes (i.e. $\sqrt{N}$). 

\begin{prob}[\(\andor\)]
    \hspace{0.1in}
    \begin{problemTotal} \label{prob:andor}
        \probleminput{Oracles for $N$ strings $X^i \in \{0,1\}^{N}$ for $i \in [N]$.}
        \problemquestion{Decide if there exist an $i$ such that $X^i = 0^N$.}
    \end{problemTotal}
\end{prob}

We now define $\ksum$. We will define its specific Boolean version, which will be comparable to the definition of $\blockksum$ (see \Cref{prob:blockKsum}) from \cite{Aaronson16}. For the quantum lower bound of \cite{Belovs13,Aaronson16} to go through, we would need the alphabet size to be $\Omega(N^k)$. Therefore, in the definitions of $\ksum$ and $\blockksum$, we will choose the size of each block representing an integer to be $10k \log N$.

\begin{prob}[\(\ksum\)] Let $M = N^{10k} = \Omega(N^k)$ be the alphabet size.  
\hspace{0.1in}
\begin{problemPremise}\label{prob:ksum}
    
    \probleminput{An oracle for a string \(X \in \{0,1\}^N.\)}
    \problempremise{Divide \(X\) into blocks of size \(10k \log N\), where each block represent a number in \([M]\).}
    \problemquestion{Decide if there are \(k\) blocks whose corresponding numbers sum to \(0 \bmod M\). }
\end{problemPremise}
\end{prob}

Next, we define $\blockksum$. Unlike $\ksum$ where every block represents a number in the alphabet, only the balanced blocks represent numbers in $\blockksum$, while an additional check is necessary on the unbalanced blocks to compute the function.

\begin{prob}[\(\blockksum\) \cite{Aaronson16}] We will have $k = \log N$ unless otherwise mentioned. Let $M = N^{10k} = \Omega(N^k)$ be the alphabet size. 
\hspace{0.1in}
\begin{problemPremise} \label{prob:blockKsum}
   
    \probleminput{An oracle for a string \(X \in \{0,1\}^N.\)}
    \problempremise{Divide \(X\) into blocks of size \(10k \log N\). A block is balanced if it has an equal number of $0$s and $1$s. Each balanced block represents a number in \([M]\).}
    \problemquestion{Decide if there are \(k\) balanced blocks whose corresponding numbers sum to \(0 \bmod M\), and all other blocks have at least as many $1$s as $0$s.}
\end{problemPremise}
\end{prob}

\cite{Aaronson16} constructed \(\bkk\) to show a maximal (quadratic) separation between quantum query complexity and certificate complexity. We state these results in a theorem below for easy referencing. 

\begin{prob}[\((\bkk)\) \cite{Aaronson16}]
\hspace{0.1in}
\begin{problemTotal} \label{prob:bkk}
    
    \probleminput{An oracle for a string \(X \in \{0,1\}^{N^2}\).}
    \problemquestion{Compute \(\bkkfull(X)\).}
   
\end{problemTotal}
\end{prob}

\begin{theorem}[\cite{Aaronson16},Section 4]
\label{thm:bkk}
    $\Quant(\bkk_{N^2}) = \tilde{\Omega}(N^2)$ and $\cert(\bkk_{N^2}) = \tilde O(N)$.
\end{theorem}

Deutsch and Jozsa \cite{Deutsch92} found a function $f:\{0,1\}^N\rightarrow\{0,1\}$ with $\Quant(f) = O(1), \Rand(f) = O(1)$ and $\Det(f) = \Theta(N)$. We modify this function slightly to have $\Quant(f) = \Rand(f) = 1$ but still have $\Det(f) = \Theta(N)$. 
\begin{prob}[\textsc{Deutsch Jozsa modified}]
\hspace{0.1in}
\begin{problem} \label{prob:djvar}
    \probleminput{An oracle for a string $X \in \{0,1\}^N$.}
    \problempromise{Either $X_i = 0$ for all $i \in \{0,1\}^n$ or $|X| = N/2$ ( i.e. $X$ is balanced between $1$ and $0$)}
    \problemquestion{Decide which is the case.}
\end{problem}
\end{prob}

The $1$-query randomized algorithm queries a random $i$, outputs $0$ with probability $2/3$ if $X_i = 0$ and outputs $1$ otherwise. We will sometimes refer to the \textsc{Deutsch Jozsa modified} function as \textsf{dj}.

The Forrelation problem, introduced by Aaronson \cite{Aaronson10}, can be solved with bounded error using $1$ quantum query yet requires $\Tilde\Omega(\sqrt{N})$ classical randomized queries to be solved with bounded error \cite{Aaronson15}.

\begin{prob}[\textsc{Forrelation}] Let $N=2^n$ and $X,Y: \{0,1\}^n \rightarrow \{-1,1\}$ be Boolean functions. Let $$\Phi_{X,Y} = \frac{1}{2^{3N/2}} \sum_{i,j \in \{0,1\}^n} X(i)(-1)^{i \cdot j} Y(j)$$ denote the ``Forrelation'' (\textit{fourier correlation}) of $X$ and $Y$.
\hspace{0.1in}

\begin{problem} \label{prob:Forrelation}
    \probleminput{An oracle for functions $X,Y : \{0,1\}^n \rightarrow \{-1,1\}$}
    \problempromise{Either $|\Phi_{X,Y}| \leq \frac{1}{100}$ or $\Phi_{X,Y} \geq \frac{3}{5}$.}
    \problemquestion{Decide which is the case.}
\end{problem}
\end{prob}

We will sometimes refer to the \textsc{Forrelation} function as \textsf{for}.

The pointer chasing problem is a canonical example of a problem which does not benefit from parallel queries. It has been studied in classical communication complexity and cryptography \cite{ponzio01pointers, cohen18posw}, and in similar contexts quantumly \cite{Klauck01,blocki21posw, Chung21}.  The input consists of \(N = 2^n\) blocks of size \(n\) each. Each block contains a pointer to a block. \(X(i)\) refers to the block pointed to by block \(i\), and \(X^k(i)\) refers to the block arrived at after following the pointer \(k\) times. Formally, the pointer chasing function is: 
\begin{prob}[\textsc{Pointer Chasing}] 
Let $N = 2^n$ be the instance size and $k$ be the chain length.
\begin{problemTotal}
    \label{prob:pointer-chasing}
    \probleminput{An oracle for $X:\{0,1\}^n \rightarrow \{0,1\}^n$ (represented by a string of length $nN$)}
    \problemquestion{Decide the last bit of $X^k(0^n)$.}
\end{problemTotal}
\end{prob}

We will need the following results about the \textsc{Pointer Chasing} function, which we prove in \Cref{sec:AppendixPointer-chasing-proofs} for completeness. We will sometimes use the notation $\textsc{pointer}_{N,k}$ to refer to the \textsc{Pointer Chasing} function with instance size $N$ and chain length $k$.

\begin{theorem} \label{thm:pointer_chasing}
    Let $f$ be $\textsc{pointer}_{N,k}$ for some integers $N,k$ and let $p \in [N]$. Then, 
    \begin{enumerate}[label=(\roman*)]
        \item $\emph{\Det}^{p \parallel}(f) = \tilde O(\min(k, N/p))$ (so $\emph{\Rand}^{p \parallel}(f) = \tilde O(\min(k, N/p))$ and $\emph{\Quant}^{p \parallel}(f) =\tilde O(\min(k, N/p))$),
        \item $\emph{\Det}^{p \parallel}(f) = \Omega(\min(k, N/p))$, 
        \inlineitem $\emph{\Rand}^{p \parallel}(f) = \Omega(\min(k, N/pk))$,
        \inlineitem $\emph{\Quant}(f) = \Omega(k)$.
    \end{enumerate}
\end{theorem}

\subsection{Adversary method for parallel algorithms}\label{subsection:adv_method_parallel}
The (negative weighted) adversary quantity \cite{Hoyer07} is known to characterize (bounded-error) quantum sequential query complexity \cite{Reichardt11b}. This result was generalized to $p$-parallel quantum queries by \cite{Jeffery17}. For any function $f:\mathcal{D} \rightarrow \mathcal{M}$, a matrix $\Gamma \in \mathbb{R}^{|\mathcal{D}| \times |\mathcal{D}|}$ is an adversary matrix for $f$ if $\Gamma$ is symmetric and for any $x,y$ if $f(x) = f(y)$, then $(\Gamma)_{xy} = 0$. Alternatively, we say that any $\Gamma \in \mathbb{R}^{X \times Y}$ is an adversary matrix for $f$ where $X = f^{-1}(0)$ and $Y = f^{-1}(1)$. For any adversary matrix $\Gamma$ and any set $S \subset [N]$, let $\Gamma_S$ be defined as  
\begin{align*}
    (\Gamma_S)_{xy} \defeq \begin{cases}
        (\Gamma)_{xy} & \text{ if } x_S \neq y_S \\
        0 & \text{otherwise}
    \end{cases}
\end{align*}
for all row indices $x$ and column indices $y$ of $\Gamma$.

The $p$-parallel adversary quantity for $f$ is defined as
\begin{align}
    \Adv^{p \parallel}(f) \defeq \max_\Gamma \frac{\norm{\Gamma}}{\max_{S \subset [N], |S| = p}\norm{\Gamma_S} } \label{equation:parallelSpectralAdversary}
\end{align}
where the maximization is over all adversary matrices for $f$.

We state the correspondence between the $p$-parallel quantum query complexity and the $p$-parallel adversary quantity for later reference.
\begin{theorem} [Parallel quantum query complexity characterization \cite{Jeffery17}]
    For any function $f:\{0,1\}^N \rightarrow \mathcal{M}$, $\emph{\Quant}^{p\parallel}(f) = \Theta(\emph{\Adv}^{p\parallel}(f))$.
    \label{theorem:parAdv}
\end{theorem}

The positive weighted adversary method was originally formulated by \cite{Ambainis02} combinatorially, which was later generalized to allow for negative weights \cite{Hoyer07}. The latter version characterizes bounded error quantum query complexity up to constant factors \cite{Reichardt11b}, and both can be generalized to $p$-parallel quantum queries (where once again the negative weighted version is tight up to constants). We note that the combinatorial version is not generally tight (and is known to be not tight in the sequential case \cite{Zhang04}), but can be easier to work with.

\begin{theorem} [Parallel combinatorial adversary method \cite{Grover04,Burchard19}] 
\label{thm:parallel_comb_adv}
Let $f:\{0,1\}^N \rightarrow \mathcal{M}$ be any function and \(S\) be a subset of \([N]\) such that \(|S| \leq p\). For any \(X\) and \(Y\) be such that \(f(x) \neq f(y) \) for all \( x \in X, y \in Y\) and any relation \(R \subseteq X \times Y\), define 
\begin{equation*}
w(x, y) \defeq \begin{cases}
1 & \text{ if } (x, y) \in R \\
0 &\text{ otherwise }
\end{cases}
\end{equation*}
\begin{equation*}
w_x \defeq \sum_y w(x, y) \qquad w_{x, S} \defeq \sum_{x, S \; \vert \; x_S \neq y_S}  w(x,y) \qquad
w_y \defeq \sum_x w(x, y) \qquad w_{y, S} \defeq \sum_{y, S \; \vert \; x_S \neq y_S}  w(x,y)
\end{equation*}
\begin{equation*}
m \defeq \min_{x \in X} \; w_x \qquad l \defeq \max_{x, S} \; w_{x, S} \qquad
m' \defeq \min_{y \in Y} \; w_y \qquad l' \defeq \max_{y, S} \; w_{y, S}
\end{equation*}

Then, $\emph{\Quant}^{p \parallel}(f) = \Omega\left(\max_{X,Y,R} \sqrt{\frac{mm'}{\ell \ell'}}\right)$.
\end{theorem}

Sequential deterministic and quantum query complexities of composed Boolean functions is easy to express in terms of the complexity of the functions being composed due to the work of \cite{Reichardt11b,Kimmel12,Tal13,Montanaro13}. No such general result is known for randomized query complexity. In fact, when the relevant functions are partial, it is known to be not true \cite{Ben-David20}. However, when the outer function is $\textsf{And}$ or $\textsf{Or}$ (or a composition of the two), such a result is known \cite{Aaronson16}. We summarize these results in the following theorem.

\begin{theorem}[Query complexity composition theorem 
]
\label{thm:QueryComplexityComposition}
    Let $f:\mathcal{D}_1 \rightarrow \{0,1\}$ and $g:\mathcal{D}_2 \rightarrow \{0,1\}$ be any Boolean functions with $\mathcal{D}_1 \subseteq \{0,1\}^{N_1}$ and $\mathcal{D}_2 \subseteq \{0,1\}^{N_2}$ for some integers $N_1,N_2$. Then,
    \begin{enumerate}
        \item $\emph{\Det}(f \circ g) = \Theta\left(\emph{\Det}(f) \cdot \emph{\Det}(g)\right)$ \emph{\cite{Tal13,Montanaro13}},
        \item $\emph{\Rand}(f \circ g) = \tilde{O}\left(\emph{\Rand}(f) \cdot \emph{\Rand}(g)\right)$, and if $g \in \{\emph{\textsf{And}}, \emph{\textsf{Or}}\}$ or $\emph{\Rand}(f) =\Theta(N_1)$, then $\emph{\Rand}(f \circ g) = \Omega\left(\emph{\Rand}(f) \cdot \emph{\Rand}(g)\right)$ \emph{\cite{Aaronson16,Ben-David20,Chakraborty23}},
        \item $\emph{\Quant}(f \circ g) = \Theta\left(\emph{\Quant}(f) \cdot \emph{\Quant}(g)\right)$ \emph{\cite{Reichardt11b,Kimmel12}}.
    \end{enumerate}
 
\end{theorem}

\subsection{Cheat sheet framework}
\label{subsec:cheatsheetPrelims}

The cheat sheet framework of Aaronson et al. \cite{Aaronson16} was primarily designed to construct separations between various Boolean complexity measures for total functions \cite{Aaronson16, Anshu16}, especially a super-quadratic separation between the quantum and randomized query complexities. 

This framework allows for constructing a total Boolean function $f_{\cheatsheet}$ from a partial function $f$ that preserves some of the desirable properties of $f$. 

\begin{definition}[Cheat sheet framework \cite{Aaronson16}] \label{def:cheatsheet}
    Let $f:\mathcal{D} \rightarrow \{0,1\}$ be any Boolean function with $\mathcal{D} \subseteq \{0,1\}^F$ and let $c$ be an integer. We define a total function $f^{c}_{\cheatsheet}: \{0,1\}^{cF} \times \{0,1\}^{M2^c} \rightarrow \{0,1\}$ constructed from $f$ as follows where $M$ is the size of the ``cheatsheet" and is the number of bits required to certify if \(c\) inputs to \(f\) are in the domain of \(f\). On input $z = (x^{(1)},x^{(2)}, \cdots, x^{(c)}, y^{(1)},y^{(2)}, \cdots y^{(2^c)})$, $f^c_{\cheatsheet} $ where $|x^{(i)}| = F$ and $|y^{(i)}| = M$, $f^{c}_{\cheatsheet}$ outputs $1$ iff the following conditions are satisfied (where the condition \ref{part:cscondition2} is only relevant when the condition \ref{part:cscondition1} is true).
    \begin{enumerate}
        \item \label{part:cscondition1} For each $i \in [c]$, $x^{(i)} \in \mathcal{D}$.
        \item \label{part:cscondition2} Let $\ell \in \{0,1\}^c$ be a binary string such that $\ell_i = f(x^{(i)})$. Then, $y^{(\ell)}$ certifies condition \ref{part:cscondition1}, and that $\ell_i = f(x^{(i)})$ for all $i \in [c]$.  
    \end{enumerate}
\end{definition}

For any input $z = (x^{(1)},x^{(2)}, \cdots, x^{(c)}, y^{(1)},y^{(2)}, \cdots y^{(2^c)})$ to $f^c_{\cheatsheet}$, we will usually refer to $x^{(1)},x^{(2)}, \cdots, x^{(c)}$ and $y^{(1)},y^{(2)}, \cdots y^{(2^c)}$ as the \textit{address} and \textit{data} parts of $z$ respectively. Moreover, we will refer to $y^{(i)}$ as the $i^{th}$ cheat sheet for input $z$. The cheat sheet framework is usually embedded with a function $f$ that has small certificate complexity so that the complexity of checking the condition \ref{part:cscondition2} in \Cref{def:cheatsheet} is bounded by the complexity of computing $f$. Moreover, we would usually want the value of $c$ to be small enough that so that the complexity of computing $c$ instances of $f$ is not much more than the complexity of computing one instance but large enough so that any algorithm is forced to compute all $c$ instances of $f$ before querying any information in the cheat sheet. For these reasons, we define a canonical version of cheat sheet function that will be very handy for us.

\begin{definition}[Canonical cheat sheet function] \label{def:Canonicalcheatsheet}
    Let $f:\mathcal{D} \rightarrow \{0,1\}$ be any Boolean function with $\mathcal{D} \subseteq \{0,1\}^N$ and let $f':\mathcal{D}' \rightarrow \{0,1\}$ be the function $f \circ \emph{\textsf{And}}_N \circ \emph{\textsf{Or}}_N$ with $\mathcal{D}'\subseteq \{0,1\}^{N^3}$. Let $c = 10 \log N$, $c' = O(\log N)$ be the space needed to store a query address-value pair for $f$, and $V$ be the size of the cheatsheet. Then, define the canonical cheat sheet function $f_{\canonicalcheatsheet}:\{0,1\}^{cN^3} \times \{0,1\}^{V2^c} \rightarrow \{0,1\}$ as $(f')^{c}_{\cheatsheet}$ 
\end{definition}

Notice that $f_{\canonicalcheatsheet}$ is a total Boolean function and $M$ is \(O(cc'N^2) = O(N^2 \log^2N)\) since the certificate complexity of $\textsf{And}_N \circ \textsf{Or}_N$ is only $N$. 
Intuitively, the cheat sheet framework would not help any algorithm perform more efficiently since it would need to compute $c$ independent instances of $f$ before it could access the relevant cheat sheet.

\begin{theorem}[\cite{Aaronson16}] \label{thm:cheatsheetcomplexities}
    Let $f:\mathcal{D} \rightarrow \{0,1\}$ be any Boolean function with $\mathcal{D} \subseteq \{0,1\}^N$ and let $c = \poly \log N$ be an integer. Then, $\emph{\Det}(f^c_\cheatsheet) = \Tilde\Omega \left(\emph{\Det}(f)\right)$, $\emph{\Rand}(f^c_\cheatsheet) = \Tilde\Omega \left(\emph{\Rand}(f)\right)$ and $\emph{\Quant}(f^c_\cheatsheet) = \Tilde\Omega \left(\emph{\Quant}(f)\right)$. 
\end{theorem}

\section{Unbounded parallel separations for total functions} \label{sec:cheatsheet}

In this section, we demonstrate the power of the cheat sheet framework by constructing exponential separations between parallel query algorithms for total functions. Our separations use the cheat sheet (\Cref{def:Canonicalcheatsheet}) framework, which takes as input a partial function \(f\) and returns a total function. Substituting different \(f\) gives us the separations for quantum vs randomized and randomized vs deterministic. First in \Cref{sec:CheatsheetUpperbounds} we give upper bounds for  this canonical cheat sheet function. We use these upper bounds in \Cref{sec:CheatsheetRandDetSep} to show the unbounded parallel separation between randomized and deterministic complexities, and in \Cref{sec:CheatsheetQuantRandSep} to show the unbounded parallel separation between quantum and randomized complexities. Finally, in \Cref{subsec:blockSensitivity} we prove that the block sensitivity of the canonical cheat sheet function is nearly equal to the parallelism at which unbounded speedup occurs, thus refuting \Cref{conj:JefferyTotalFunctions} in the strongest possible way.  

\subsection{Upper bounds in the cheatsheet framework}
\label{sec:CheatsheetUpperbounds}

We start off by proving deterministic and randomized sequential upper bounds for any partial function \(f\) composed with a total function \(g\) plugged into a cheat sheet. 

\begin{proposition} \label{prop:cheatsheet_seq_upperbound}
    Let $f:\mathcal{D} \rightarrow \{0,1\}$ be any (partial) Boolean function with $\mathcal{D} \subseteq \{0,1\}^N$, let $g:\{0,1\}^{N^2} \to \{0,1\}$ be any total Boolean function with $\cert(g) = \tilde{O}(N)$, and let $c = 10 \log N$. Let $(f \circ g)^c_{\cheatsheet}$ be the cheatsheet version of $f \circ g$ (see \emph{\Cref{def:cheatsheet}}). Then,
    \begin{enumerate} [label=(\roman*)]
        \item \label{itm:det_cheatsheet_upperbound} $\emph{\Det}((f \circ g)^c_{\cheatsheet}) = \tilde{O}(\max(\emph{\Det}(f \circ g),N^2))$,
        \item \label{itm:rand_cheatsheet_upperbound} $\emph{\Rand}((f \circ g)^c_{\cheatsheet}) = \tilde{O}(\max(\emph{\Rand}(f \circ g),N^2))$.
    \end{enumerate}
\end{proposition}

\begin{proof}
    We will only prove \cref{itm:det_cheatsheet_upperbound} since the proof of \cref{itm:rand_cheatsheet_upperbound} is very similar.  
    We describe a deterministic algorithm that computes $(f \circ g)^c_{\cheatsheet}$ using at most $\tilde{O}(\max({\Det}(f \circ g),N^2))$ deterministic queries as follows. Let $z = (x^{(1)},x^{(2)}, \cdots, x^{(c)}, y^{(1)},y^{(2)}, \cdots y^{(2^c)})$ be an input to $(f \circ g)^c_{\cheatsheet}$.

    Notice that we can compute the outputs of each of the $c$ instances of $f \circ g$ on respective inputs $x^{(1)},x^{(2)}, \cdots, x^{(c)}$ using $\tilde{O}(\Det(f \circ g))$ queries to $(f \circ g)^c_{\cheatsheet}$.

    As in \Cref{def:cheatsheet}, let $\ell$ be the length $c$ binary string such that $\ell_i = (f \circ g)(x^{(i)})$. As $\cert(g) = \tilde{O}(N)$, the length of each cheatsheet $y^{(i)}$ is $\tilde{O}(N^2)$ so, in particular, we can query the contents of $y^{(\ell)}$ with $\tilde{O}(N^2)$ deterministic queries.

    Given $y^{(\ell)}$, we only need to query $\tilde{O}(N^2)$ indices (i.e. those pointed by $y^{(\ell)}$) of $(f \circ g)^c_{\cheatsheet}$ to verify that $x^{(i)}$ is in the domain of $f \circ g$. Furthermore, since $y^{(\ell)}$ certifies each input bit to each of the $c$ instances of $f$, it automatically certifies that $\ell_i = (f \circ g)(x^{(i)})$ for all $i \in [c]$, and that each \(x^{(i)}\) is in the domain of \(f \circ g\) so both conditions \ref{part:cscondition1} and \ref{part:cscondition2} in \Cref{def:cheatsheet} are satisfied. The desired result follows.
\end{proof}

Next, we prove randomized and quantum upper bounds for any partial function \(f\) composed with a total function \(g\) plugged into a cheat sheet. 

\begin{proposition} \label{prop:cheatsheet_parallel_upperbound}
    Let $f:\mathcal{D} \rightarrow \{0,1\}$ be any (partial) Boolean function with $\mathcal{D} \subseteq \{0,1\}^N$, let $g:\{0,1\}^{N^2} \to \{0,1\}$ be any total Boolean function with $\cert(g) = \tilde{O}(N)$, and let $c = 10 \log N$. Let $(f \circ g)^c_{\cheatsheet}$ be the cheatsheet version of $f \circ g$ (see \emph{\Cref{def:cheatsheet}}) with cheatsheet size $V$\footnote{Since $\cert(g) = \tilde{O}(N)$, the input to $f$ has $N$ bits and specifying each location takes $\polylog(N)$ bits, $V = \tilde{\Theta}(N^2)$}. 
    \begin{enumerate} [label=(\roman*)]
        \item \label{itm:rand_parallel_cheatsheet_upperbound} If $\emph{\Rand}^{q \parallel}(f) = k$ and $p = \max(cqN^2 \log N,V) = \tilde{\Theta}(q N^2)$, then $\emph{\Rand}^{p \parallel}((f \circ g)^c_{\cheatsheet}) \leq k+2$ and $\emph{\Rand}^{k+2 \perp}(f \circ g)^c_{\cheatsheet}) \leq p$.
        \item \label{itm:quant_parallel_cheatsheet_upperbound} If $\emph{\Quant}^{q \parallel}(f) = k$ and $p = \max(cqN^2  \log N,V) = \tilde{\Theta}(q N^2)$, then $\emph{\Quant}^{p \parallel}((f \circ g)^c_{\cheatsheet}) \leq k+2$ and $\emph{\Quant}^{k+2 \perp}(f \circ g)^c_{\cheatsheet}) \leq p$.
    \end{enumerate}
\end{proposition}

\begin{proof}
     We will only prove \cref{itm:rand_parallel_cheatsheet_upperbound} since the proof of \cref{itm:quant_parallel_cheatsheet_upperbound} is very similar. We describe a randomized algorithm that computes $(f \circ g)^c_{\cheatsheet}$ using at most $k+2$ $p$-parallel queries, which is sufficient to prove both ${\Rand}^{p \parallel}((f \circ g)^c_{\cheatsheet}) \leq k+2$ and ${\Rand}^{k \perp}(f \circ g)^c_{\cheatsheet}) = \tilde{O}(p)$. Let $z = (x^{(1)},x^{(2)}, \cdots, x^{(c)}, y^{(1)},y^{(2)}, \cdots y^{(2^c)})$ be an input to $(f \circ g)^c_{\cheatsheet}$.

    Let $p' = \floor{p/cq} \geq N^2 \log N$. Notice that we can compute the whole input to some $g$ in $f \circ g$ using $1$ $p'$-parallel query to $(f \circ g)^c_{\cheatsheet}$. Thus, by running the $k$-query $q$-parallel randomized algorithm for $f$ where each query to $f$ makes $1$ $p'$-parallel query to $(f \circ g)^c_{\cheatsheet}$ allows us to compute $f \circ g$ using $k$ many $qp'$-parallel queries with probability $1-\Theta(1/c)$. Therefore, as $cqp' \leq p$, we can compute $c$ instances of $f \circ g$ on respective inputs $x^{(1)},x^{(2)}, \cdots, x^{(c)}$ in parallel using $k$ many $p$-parallel queries with probability $2/3$.

    As in \Cref{def:cheatsheet}, let $\ell$ be the length $c$ binary string such that $\ell_i = (f \circ g)(x^{(i)})$. Since $p \geq V$ and $\cert(g) = \tilde{O}(N)$, we can query the contents of $y^{(\ell)}$ with a single $p$-parallel query.

    Given $y^{(\ell)}$, we only need to query at most $cN \cdot \cert(g)$ indices of $(f \circ g)^c_{\cheatsheet}$ to verify that $x^{(i)}$ is in the domain of $f \circ g$. But this could be done using only $1$ $p$-parallel query. Furthermore, since $y^{(\ell)}$ certifies each input bit to each of the $c$ instances of $f$, it automatically certifies that $\ell_i = (f \circ g)(x^{(i)})$ for all $i \in [c]$, and also that each \(x^{(i)}\) is in the domain of \(f \circ g\) so both conditions \ref{part:cscondition1} and \ref{part:cscondition2} in \Cref{def:cheatsheet} are satisfied. It follows that $\Rand^{p \parallel}(g_{\cheatsheet}) \leq k+2$.
\end{proof}

As a corollary, note that the above upper bounds hold for the canonical cheatsheet function, which we will use to show separations in the next two sections. 

\begin{corollary} \label{cor:can_cheatsheet}
    Let $f:\mathcal{D} \rightarrow \{0,1\}$ be any (partial) Boolean function with $\mathcal{D} \subseteq \{0,1\}^N$. Let $f_{\canonicalcheatsheet}$ be the canonical cheatsheet function associated with $f$ (see \emph{\Cref{def:Canonicalcheatsheet}}) with cheatsheet size $V$. 
    \begin{enumerate} [label=(\roman*)]
        \item \label{itm:rand_parallel_can_cheatsheet_upperbound} If $\emph{\Rand}(f) = k$ and $p = \max(cN^2 \log N,V) = \tilde{\Theta}(N^2)$, then $\emph{\Rand}^{p \parallel}(f_{\canonicalcheatsheet}) \leq k+2$ and $\emph{\Rand}^{k+2 \perp}(f_{\canonicalcheatsheet}) \leq p$.
        \item \label{itm:quant_parallel_can_cheatsheet_upperbound} If $\emph{\Quant}(f) = k$ and $p = \max(cN^2  \log N,V) = \tilde{\Theta}(N^2)$, then $\emph{\Quant}^{p \parallel}(f_{\canonicalcheatsheet}) \leq k+2$ and $\emph{\Quant}^{k+2 \perp}(f_{\canonicalcheatsheet}) \leq p$.
    \end{enumerate}
\end{corollary}

\subsection{Separation between randomized and deterministic complexities}
\label{sec:CheatsheetRandDetSep}

Using the parallel upper bound for the canonical cheatsheet function in the previous section, we show the desired unbounded separation between randomized and deterministic complexities. 

\begin{prob}[\textsc{cheatsheet deutsch-josza}] \label{prob:djcancs}
    Let $f:\{0,1\}^N \to \{0,1\}$ be $\textsf{dj}_N$ (see \emph{\Cref{prob:djvar}}).
    \begin{problemTotal}
        \probleminput{An oracle for a string $X \in \{0,1\}^{M}$ (with $M = \tilde{\Theta}(N^{12})$\footnote{Since a certificate of $\andor$ on $N^2$ bits could be represented using $\tilde{\Theta}(N)$ bits, we will have the number of input bits for the function $(f \circ g)^c_{\cheatsheet}$ to be $\tilde{\Theta}(N^{12})$.}).}
        \problemquestion{Compute $f^c_{\canonicalcheatsheet}(X)$ (see \emph{\Cref{def:Canonicalcheatsheet}}).}
    \end{problemTotal}
\end{prob}

\begin{theorem}
    Let $h$ be the \textsc{cheatsheet deutsch-josza} problem. Then there exists a  $p = \tilde \Theta(N^{2})$ such that 
    \begin{enumerate} [label=(\roman*)]
        \item $\emph{\Det}^{p \parallel}(h) = \Omega(N)$,
        \inlineitem $\emph{\Rand}^{p \parallel}(h) \leq 3$,
        \inlineitem $\emph{\Rand}^{3 \perp}(h) \leq p$.
    \end{enumerate}
    \label{thm:cheatsheetRpvsDpSeparation}
\end{theorem}

\begin{proof}
    We prove each item one by one. Recall that $f = \textsf{dj}_N$ in \Cref{prob:djcancs} and let $g = (\andor)_{N^2}$.  
    \begin{enumerate}[label=(\roman*)]
        \item It is well-known that $\Det(f) = \Omega(N)$ so $\Det(f \cdot g) = \Omega(N^3)$ by \Cref{thm:QueryComplexityComposition}. It follows, from \Cref{thm:cheatsheetcomplexities}, that $\Det(h) = \Omega(N^3)$. Therefore, $\Det^{p \parallel} = \Omega(\Det(h)/p) = \tilde \Omega(N)$.
        \item and\!\!\!\inlineitem The desired result directly follows from \cref{itm:rand_parallel_can_cheatsheet_upperbound} of \Cref{cor:can_cheatsheet}. 
    \end{enumerate}
\vspace{-1pc}
\end{proof}

The previous theorem gives us the separation we want, but note that the input is of size \(\tilde \Theta (N^{12})\) there. To make the actual separation more evident (and to emphasize the relatively small size of the exponent), we restate it with an input size of \(M\) in the corollary below. To make clear where this separation occurs, we also mention the block sensitivity which comes from our results in \Cref{subsec:blockSensitivity}.

\begin{corollary}
    There exists a total Boolean function $h: \{0,1\}^M \to \{0,1\}$ and an integer $p \in [M]$ such that 
    \begin{enumerate}[label=(\roman*)]
        \item $\emph{\Det}^{p \parallel}(h) = \tilde\Omega(M^{1/12})$,
        \inlineitem $\emph{\Rand}^{p \parallel}(h) \leq 3$,
        \inlineitem $\emph{\Rand}^{3 \perp}(h) \leq p$,
        \inlineitem $\emph{\bs}(h) = \tilde{\Omega}(p)$
    \end{enumerate}
\end{corollary}

\begin{proof}
    It follows from \Cref{thm:cheatsheetRpvsDpSeparation} and \Cref{cor:bsForCanonicalCheatsheet}.
\end{proof}

\subsection{Separation between quantum and randomized complexities}
\label{sec:CheatsheetQuantRandSep}

Using the parallel upper bound for the canonical cheatsheet function in \Cref{sec:CheatsheetUpperbounds}, we show the desired unbounded separation between quantum and randomized complexities.

\begin{prob}[\textsc{cheatsheet forrelation}] \label{prob:forcancs}
    Let $f:\{0,1\}^N \to \{0,1\}$ be $\textsf{for}_N$ (see \emph{\Cref{prob:Forrelation}}).

    \begin{problemTotal}
        \probleminput{An oracle for a string $X \in \{0,1\}^{M}$ (with $M = \tilde{\Theta}(N^{12})$\footnote{Since a certificate of $\andor$ on $N^2$ bits could be represented using $\tilde{\Theta}(N)$ bits, we will have the number of input bits for the function $(f \circ g)^c_{\cheatsheet}$ to be $\tilde{\Theta}(N^{12})$.}).}
        \problemquestion{Compute $f^c_{\canonicalcheatsheet}(X)$ (see \emph{\Cref{def:Canonicalcheatsheet}}).}
    \end{problemTotal}
\end{prob}

\begin{theorem}
    Let $h$ be the \textsc{cheatsheet forrelation} problem. Then, there exists $p = \Theta(N^{2})$ such that
    \begin{enumerate} [label=(\roman*)]
        \item $\emph{\Rand}^{p \parallel}(h) = \tilde{\Omega}(\sqrt{N})$,
        \inlineitem $\emph{\Quant}^{p \parallel}(h) \leq 3$,
        \inlineitem $\emph{\Quant}^{3 \perp}(h) \leq p$.
    \end{enumerate}
    \label{thm:cheatsheetQpvsRpSeparation}
\end{theorem}

\begin{proof}
    We prove each item one by one. Recall that $f = \textsf{for}_N$ in \Cref{prob:djcancs} and let $g = (\andor)_{N^2}$.  
    \begin{enumerate}[label=(\roman*)]
        \item It is well-known that $\Rand(f) = \tilde{\Omega}(\sqrt{N})$ so $\Rand(f \cdot g) = \tilde{\Omega}(N^{5/2})$ by \Cref{thm:QueryComplexityComposition}. It follows, from \Cref{thm:cheatsheetcomplexities}, that $\Rand(h) = \tilde{\Omega}(N^{5/2})$. Therefore, $\Rand^{p \parallel} = \Omega(\Rand(h)/p) = \tilde \Omega(\sqrt{N})$.
        \item and\!\!\!\inlineitem The desired result directly follows from \cref{itm:quant_parallel_can_cheatsheet_upperbound} of \Cref{cor:can_cheatsheet}. 
    \end{enumerate}
\vspace{-1pc}
\end{proof}

Analogous to the randomized vs deterministic case, the previous theorem gives us the separation we want, but note that the input is of size \(\tilde \Theta (N^{12})\) there. To make the actual separation more evident (and to emphasize the relatively small size of the exponent), we restate it with an input size of \(M\) in the corollary below. To make clear where this separation occurs, we also mention the block sensitivity which comes from our results in \Cref{subsec:blockSensitivity}.
\begin{corollary}
    There exists a total Boolean function $h: \{0,1\}^M \to \{0,1\}$ and an integer $p \in [M]$ such that 
    \begin{enumerate}[label=(\roman*)]
        \item $\emph{\Rand}^{p \parallel}(h) = \tilde\Omega(M^{1/24})$,
        \inlineitem $\emph{\Quant}^{p \parallel}(h) \leq 3$,
        \inlineitem $\emph{\Quant}^{3 \perp}(h) \leq p$,
        \inlineitem $\emph{\bs}(h) = \tilde{\Omega}(p)$
    \end{enumerate}
\end{corollary}

\begin{proof}
    It follows from \Cref{thm:cheatsheetQpvsRpSeparation} and \Cref{cor:bsForCanonicalCheatsheet}.
\end{proof}

\subsection{Discussion of block sensitivity}
\label{subsec:blockSensitivity}

The block sensitivity of a function \(\bs(f)\) is intimately related with parallel exponential advantage. Jeffery et al \cite{Jeffery17} show that if \(p = O(bs(f)^{1-\epsilon})\) for a constant $\epsilon>0$, then \(\Quant^{p\parallel}(f)\) and \(\Det^{p\parallel}(f)\) are polynomially related. We use \Cref{thm:blockSensitivity} to give a lower bound on \(\bs(f)\) for functions that exhibit exponential parallel advantage and show that this phenomenon occurs exactly when \(p = \tilde O(bs(f))\). Moreover, we observe that showing an exponential parallel separation has the unexpected application of providing an upper bound for \(\bs(f)\), and we exploit this to characterize \(\bs(f)\) for cheat sheet functions in \Cref{cor:bsThetaforCanonicalCheatsheet}. 

Below is a general theorem to lower bound the block sensitivity of cheat sheet functions.

\begin{theorem}
\label{thm:blockSensitivity}
Let \(g: \mathcal D \to \{0,1\}\) be a partial function where \(\mathcal D \subseteq \{0,1\}^G\) and let \(h: \{0,1\}^H \to \{0,1\}\) be a total function with \(\min(\bs_0(h), \bs_1(h)) = \tilde \Omega(B)\). Thus, \(f = g \circ h : \{0,1\}^{GH} \to \{0,1\}\) is a partial function. Construct the cheat sheet function \(f^c_{cs}\) as described in \Cref{def:cheatsheet}. Then, \(\bs(f^c_{cs}) = \tilde \Omega (BG)\).
\end{theorem}
\begin{proof}
We prove this by constructing an input \(\mathcal X\) to \(f^c_{cs}\) that attains block sensitivity \(\tilde \Omega(BG)\) as follows:
\begin{enumerate}
\item For each of the \(c\) copies of \(g\), select some input \(z^{(i)}\) where \(|z^{(i)}| = G\) such that \(z^{(i)} \in \mathcal D\). Denote the bits of \(z^{(i)}\) as \(z^{(i)}_j\).
\item For each \(z^{(i)}_j\), select string \(x^{(i)}_j\) such that  \(|x^{(i)}_j| = H\), \(h(x^{(i)}_j) = z^{(i)}_j\) and \(\bs^h(x^{(i)}_j) = \tilde \Omega(B)\). The concatenation of all the \(x^{(i)}_j\)'s constitutes the address section of our input string.
\item Let \(\ell \in \{0,1\}^c\) be such that \(\ell_i = f(x^{(i)})\) for all \(i\). Then let \(y^{(l)}\) contain a correct certificate. At all other \(y^{(k)}\), place some spurious data.
\end{enumerate}

Observe that each \(x^{(i)}_j\) has \(\tilde \Omega(B)\) disjoint sensitive blocks for the function \(h\), and there are \(\tilde \Omega(G)\) many \(x^{(i)}_j\)'s. We will argue that all of these \(\tilde \Omega(BG)\) disjoint blocks of \(\mathcal X\) are sensitive for \(f^c_{cs}\). First note that by construction, \(f_{cs}^c(\mathcal X) = 1\). Suppose one of these blocks, say one in \(x^{(i)}_j\), is flipped. Because this block is sensitive for \(h\), this causes \(z^{(i)}_j\) to flip to \(z^{(i)'}_j\), and we denote the new string by \(z^{(i)'}\), and the new input as \(\mathcal X'\). Then, one of three cases could occur: 
\begin{enumerate}
\item \(z^{(i)'} \notin \mathcal D\), which means by definition, \(f_{cs}^c(\mathcal X') = 0\).
\item \(z^{(i)'} \in \mathcal D\) but \(g(z^{(i)'}) \neq g(z^{(i)})\). This results in a new \(l'\) but \(y^{(l')}\) contains spurious data by construction, so \(f_{cs}^c(\mathcal X') = 0\).
\item \(z^{(i)'} \in \mathcal D\) and \(g(z^{(i)'}) = g(z^{(i)})\). This time, the certificate in \(y^{(l)}\) is wrong, hence \(f_{cs}^c(\mathcal X') = 0\).
\end{enumerate}

Hence, \(\bs(f_{cs}^c) = \tilde \Omega (BG)\).

\end{proof}

The above theorem immediately implies a lower bound for the block sensitivity of the canonical cheat sheet function.

\begin{corollary}
\label{cor:bsForCanonicalCheatsheet}
    For any partial function \(f: \mathcal D \to \{0,1\}\) where \(\mathcal D \subseteq \{0,1\}^N\), the canonical cheatsheet function \(f_{\canonicalcheatsheet}\) satisfies \(\emph{\bs}(f_{\canonicalcheatsheet}) = \tilde \Omega(N^2)\). 
\end{corollary}

\begin{proof}
    We have \(h = f \circ \andor \), where \(\andor\) satisfies \(\min(\bs_0, \bs_1) = \tilde \Omega(N)\). Thus, \(g_{\canonicalcheatsheet}:= f^c_{\cheatsheet}\) satisfies the conditions in \Cref{thm:blockSensitivity} with \(B = N\) and \(G = N\). Hence, \(\bs(f_{\canonicalcheatsheet}) = \tilde \Omega(N^2)\).
\end{proof}

Combining the above with the upper bound we get from Jeffery et al \cite{Jeffery17} and our \Cref{thm:cheatsheetQpvsRpSeparation}, we characterize the block sensitivity for cheat sheet functions. 

\begin{corollary}
\label{cor:bsThetaforCanonicalCheatsheet}
Let \(f: \mathcal D \to \{0,1\}\) be a partial function such that  \(\mathcal D \subseteq \{0,1\}^N\) and \(\emph{\Quant}(f) = k\) and \(\emph{\Rand}(f) = \tilde \Theta (N^\alpha)\) hold for some \(0 < \alpha \leq 1\). Then, the canonical cheatsheet function \(g_{\canonicalcheatsheet}\) satisfies \(\emph{\bs}(g_{\canonicalcheatsheet}) = \tilde \Omega(N^2)\) and $\emph{\bs}(g_{\canonicalcheatsheet}) = O(N^{2+\epsilon})$ for all $\epsilon > 0$. 
\end{corollary}

\begin{proof}
Fix an $\epsilon$ and let $\delta = 1-\frac{2}{2+\epsilon}$. Using the same argument as in \Cref{thm:cheatsheetQpvsRpSeparation}, we can show that at \(p = \tilde \Theta(N^2)\), \(\Quant^{p\parallel}(f_{\canonicalcheatsheet}) = O(1)\) whereas \(\Rand^{p\parallel}(f_{\canonicalcheatsheet}) = \Omega(N^\alpha)\). However, from (\(\cite{Jeffery17}\), Theorem 12), we know that for all \(p = O(\bs(f)^{1-\delta})\), \(\Rand^{p\parallel}(h) = O \left ( \poly \left (\Quant^{p||}(h) \right )\right)\). Hence, it must be the case that \(p = \Omega(\bs(f_{\canonicalcheatsheet})^{1-\delta})\) so \(\bs(f_{\canonicalcheatsheet}) = O(N^{2+\epsilon})\). Combining it with \Cref{cor:bsForCanonicalCheatsheet} proves the desired result.
\end{proof}

\section{Unbounded genuine parallel separations} \label{sec:unbounded_genuine_separations}
In \Cref{sec:cheatsheet}, we demonstrated that cheatsheet functions can show an unbounded parallel speedup. However, these functions show a considerable sequential speedup too (in fact they were originally constructed in \cite{Aaronson16} for this very purpose). One might wonder then if a sequential speedup is always necessary to see a parallel speedup. Our goal is to prove this is not the case, and we do this over this section and the next (\Cref{sec:unbounded_genuine_separations}). In this section, we construct a \textit{partial} function that satisfies our properties (parallel advantage with no sequential advantage). We will use this partial function in the next section to construct a total function that accomplishes our goal. Our partial function uses the \textit{adaptive non-adaptive}, or \textsf{ANA} construction, which provides a framework for accomplishing these separations.

At a high level, the \textsf{ANA} function constructions will consist of two components: an adaptive function $f$ which does not parallelize, and a non-adaptive function $g$ which parallelizes, but admits an advantage for the stronger query model. In particular, an instance of the \textsf{ANA} function will consist of an instance of $f$ and an instance of $g$, promised that the evaluation of both will be the same. Therefore, an algorithm can choose which function it wants to solve, depending on the amount of parallelism it gets. By tuning the instance sizes, the ability of an algorithm in the stronger model to solve $g$ allows it to parallelize, while the inability of the algorithm in the weaker model to efficiently solve $g$ requires it to compute the \textsf{ANA} function only via computing $f$, which bottlenecks its ability to effectively use the given parallelism.

\subsection{Correlated functions}

Before we describe our separations, we will prove generic results relating the query complexity of a function with components promised to produce the same output with the query complexities of the components, which may be of independent interest.

Precisely, we consider the following problem of computing functions with input components correlated in their outputs.  

\begin{prob}[\correlated$(f,g)$] \label{prob:correlated}
Let $f:\domain_f \to \{0,1\}$ and $g:\domain_g \to \{0,1\}$ be any (partial) Boolean functions. 

\begin{problem}
    \probleminput{An oracle for the strings $x \in \domain_f$ and $y \in \domain_g$.}
    \problempromise{$f(x) = g(y)$.}
    \problemquestion{Output $f(x)=g(y)$.}
\end{problem}
\end{prob}

For an input $(x,y)$ to \correlated$(f,g)$, we will also refer to $x$ and $y$ as the $f$-component and $g$-component of the input to \correlated$(f,g)$ respectively. Also, we define $X_0 = f^{-1}(0)$, $X_1 = f^{-1}(1)$, $Y_0 = g^{-1}(0)$ and $Y_1 = g^{-1}(1)$.  

In the rest of this section, we aim to show that the deterministic and randomized $p$-parallel and quantum sequential query complexity of \correlated$(f,g)$ is the minimum of the relevant query complexities of computing $f$ and $g$. We begin with the result for deterministic $p$-parallel query complexities.

\begin{lemma} \label{lem:det_correlated_functions}
    Let $f:\domain_f \to \{0,1\}$ and $g:\domain_g \to \{0,1\}$ be any (partial) Boolean functions. Then, $\emph{\Det}^{p \parallel}(\correlated(f,g)) = \min(\emph{\Det}^{p \parallel}(f), \emph{\Det}^{p \parallel}(g))$. 
\end{lemma}

\begin{proof}
    Let $\mathcal{A}_f$ and $\mathcal{A}_g$ be optimal deterministic $p$-parallel algorithms for $f$ and $g$. Given an input $(x,y)$ for $\correlated(f,g)$, running either $\mathcal{A}_f$ on $x$ or $\mathcal{A}_g$ on $y$ is sufficient to compute it. Thus, $\Det^{p \parallel}(\correlated(f,g)) \leq 2 \cdot \min(\Det^{p \parallel}(f), \Det^{p \parallel}(g))$.

    We now argue for the lower bound using the adversary argument. That is, for any deterministic $p$-parallel algorithm for $\correlated(f,g)$, we answer queries to the $f$-component of the input according to an adversary strategy for $f$ and queries to the $g$-component of the input according to an adversary strategy for $g$. Now, suppose that $\mathcal{A}$ is a deterministic $p$-parallel algorithm for $\correlated(f,g)$ that makes $t < \min(\Det^{p \parallel}(f), \Det^{p \parallel}(g))$ queries. Then, it must have made less than $\Det^{p \parallel}(f)$ $\leq p$-parallel queries to the $f$-component of the input and less than $\Det^{p \parallel}(g)$ $p$-parallel queries to the $g$-component of the input. Therefore, after $t$ $p$-parallel queries, both the $f$-component and the $g$-component of the input are not determined to be $0$-inputs or $1$-inputs, meaning that there is a consistent $0$ or $1$ input to $\correlated(f, g)$. It follows that $t$ $p$-parallel queries are not sufficient to compute $\correlated(f,g)$ so $\Det^{p \parallel}(\correlated(f,g)) \geq \min(\Det^{p \parallel}(f), \Det^{p \parallel}(g))$.  
\end{proof}

Before we consider the randomized analog of \Cref{lem:det_correlated_functions}, we will prove the following useful proposition. For distributions $\dist$ and $\dist'$, we will usually say distinguishing $\dist$ and $\dist'$ to mean distinguishing an input $x$ distributed according to the distribution $\dist$ from an input $x'$ distributed according to the distribution $\dist'$ with probability $1/2 + \Omega(1)$.

\begin{proposition} \label{prop:hybrid_dist}
    Let $f:\domain_f \to \{0,1\}$ and $g:\domain_g \to \{0,1\}$ be any (partial) Boolean functions. Let $\dist^f_0$, $\dist^f_1$, $\dist^g_0$ and $\dist^g_1$ be any distributions over $X_0$, $X_1$, $Y_0$ and $Y_1$ respectively. Let $r_f = \emph{\Rand}^{p \parallel}(\dist^f_0, \dist^f_1)$ and $r_g = \emph{\Rand}^{p \parallel}(\dist^g_0, \dist^g_1)$ be the randomized $p$-parallel query complexities of distinguishing $\dist^f_0$ and $\dist^f_1$ and distinguishing $\dist^g_0$ and $\dist^g_1$ respectively. Then, the randomized $p$-parallel query complexity of distinguishing the distributions $\dist_0 = \dist^f_0 \times \dist^g_0$ and $\dist_1 = \dist^f_1 \times \dist^g_1$ is $\emph{\Rand}^{p \parallel}(\dist_0, \dist_1) = \Omega(\min(r_f, r_g))$. 
   
\end{proposition}

\begin{proof}
    Let $\dist_{\hybrid} = \dist^f_0 \times \dist^g_1$. Then, since one can sample from the distribution $\dist^f_0$ without any queries, distinguishing $\dist_0$ from $\dist_{\hybrid}$ allows one to distinguish $\dist^g_0$ from $\dist^g_1$. Thus, $\Rand^{p \parallel}(\dist_0, \dist_{\hybrid}) = \Omega(r_g)$. Similarly, $\Rand^{p \parallel}(\dist_{\hybrid}, \dist_1) = \Omega(r_f)$. Suppose that a randomized $p$-parallel algorithm $\mathcal{A}$ making $t$ queries could distinguish $\dist_0$ and $\dist_{\hybrid}$ with probability $1/2 + \alpha(t)$ and distinguish $\dist_{\hybrid}$ and $\dist_1$ with probability $1/2 + \beta(t)$. Then, $\mathcal{A}$ can distinguish $\dist_0$ and $\dist_1$ with probability at most $1/2 + \alpha(t) + \beta(t)$. Therefore, if we want $\alpha(t) + \beta(t) = \Omega(1)$, then we must choose $t$ so that either $\alpha(t) = \Omega(1)$ or $\beta(t) = \Omega(1)$ (or both). It follows that $\Rand^{p \parallel}(\dist_0, \dist_1) = \min(\Rand^{p \parallel}(\dist_0, \dist_{\hybrid}), \Rand^{p \parallel}(\dist_{\hybrid}, \dist_1)) = \Omega(\min(r_f,r_g))$.
\end{proof}

We are now ready to prove the following relationship between the randomized $p$-parallel query complexities of $f$, $g$ and $\correlated(f,g)$.

\begin{lemma} \label{lem:rand_correlated_functions}
    Let $f:\domain_f \to \{0,1\}$ and $g:\domain_g \to \{0,1\}$ be any (partial) Boolean functions. Then, $\emph{\Rand}^{p \parallel}(\correlated(f,g)) = \Theta(\min(\emph{\Rand}^{p \parallel}(f), \emph{\Rand}^{p \parallel}(g))$. 
\end{lemma}

\begin{proof}
    Similar to the deterministic case, computing $f$ or $g$ is sufficient to compute $\correlated(f,g)$. Thus, $\Rand^{p \parallel}((\correlated(f,g)) \leq \min(\Rand^{p \parallel}(f), \Rand^{p \parallel}(g))$.

    Now, we argue for the lower bound. Let $\dist^f_0$ and $\dist^f_1$ be respective distributions over $X_0$ and $X_1$ such that distinguishing them requires any randomized $p$-parallel algorithm to make $\Omega(\Rand^{p \parallel}(f))$ queries. Similarly, let $\dist^g_0$ and $\dist^g_1$ be respective distributions over $Y_0$ and $Y_1$ such that distinguishing them requires any randomized $p$-parallel algorithm to make $\Omega(\Rand^{p \parallel}(g))$ queries. That is, $r_f = \Rand^{p \parallel}(\dist^f_0, \dist^f_1) = \Omega(\Rand^{p \parallel}(f))$ and $r_g = \Rand^{p \parallel}(\dist^g_0, \dist^g_1) = \Omega(\Rand^{p \parallel}(g))$. By \Cref{prop:hybrid_dist}, we know that the randomized $p$-parallel query complexity of distinguishing the distributions $\dist_0 = \dist^f_0 \times \dist^g_0$ and $\dist_1 = \dist^f_1 \times \dist^g_1$ is $\Rand^{p \parallel}(\dist_0, \dist_1) = \Omega(\min(r_f, r_g)) = \Omega(\min(\Rand^{p \parallel}(f), \Rand^{p \parallel}(g))$. But, since $\dist_0$ is a distribution over $X_0 \times Y_0$ and $\dist_1$ is a distribution over $X_1 \times Y_1$, the task of distinguishing $\dist_0$ and $\dist_1$ is at most as hard as computing $\correlated(f,g)$. It follows that $\Rand^{p \parallel}(\correlated(f,g)) \geq \Rand^{p \parallel}(\dist_0, \dist_1) = \Omega(\min(\Rand^{p \parallel}(f), \Rand^{p \parallel}(g))$. 
\end{proof}

Last, we will prove a relationship between the quantum query complexity of $f$, $g$ and $\correlated(f,g)$. 

\begin{lemma} \label{lem:quant_correlated_functions}
    Let $f:\domain_f \to \{0,1\}$ and $g:\domain_g \to \{0,1\}$ be any (partial) Boolean functions on $n_f$ and $n_g$ bits respectively. Then, $\emph{\Quant}(\correlated(f,g)) = \Theta(\min(\emph{\Quant}(f), \emph{\Quant}(g))$. 
\end{lemma}

\begin{proof}
    Similar to the deterministic and the randomized case, computing $f$ or $g$ is sufficient to compute $\correlated(f,g)$. Thus, $\Quant(\correlated(f,g)) = O(\min(\Quant(f), \Quant(g)))$.

    Now, we argue for the lower bound. Let $\Gamma^{(f)} \in \mathbf{R}^{|X_0| \times |X_1|}$ and $\Gamma^{(g)} \in \mathbf{R}^{|Y_0| \times |Y_1|}$ be an optimal adversary matrices for the functions $f$ and $g$ respectively. That is, $\Adv(f) = \frac{\norm{\Gamma^{(f)}}}{\min_{i \in [n_f]} \norm{\Gamma^{(f)}_i}}$ and $\Adv(g) = \frac{\norm{\Gamma^{(g)}}}{\min_{i \in [n_f+1,n_g]} \norm{\Gamma^{(g)}_i}}$.  Let $\Gamma^{(\correlated(f,g))} = \Gamma{(f)} \otimes \Gamma{(g)} \in \mathbf{R}^{|X_0| |Y_0| \times |X_1| |Y_1|}$, where \(\otimes\) refers to the Kronecker product. It is easy to see that $\Gamma^{(\correlated(f,g))}$ is an adversary matrix for $\correlated(f,g)$. Moreover, for any $i \in [n_f]$, $\Gamma^{(\correlated(f,g))}_i = \Gamma^{(f)}_i \otimes \Gamma^{(g)}$. Similarly, for any $i \in [n_f+1,n_g]$, $\Gamma^{(\correlated(f,g))}_i = \Gamma^{(f)} \otimes \Gamma^{(g)}_i$. Therefore,
    \begin{align*}
        \Quant(\correlated(f,g)) &= \Omega(\Adv(\correlated(f,g)) \\ 
        &= \Omega\left(\min_{i \in [n_f + n_g] }\frac{\norm{\Gamma^{(\correlated(f,g))}}}{\norm{\Gamma^{(\correlated(f,g))}_i}}\right) \\
        &= \Omega\left(\min\left(\min_{i \in [n_f]}\frac{\norm{\Gamma^{(f)}}}{\norm{\Gamma^{(f))}_i}}, \min_{i \in [n_f+1, n_g]} \frac{\norm{\Gamma^{(g)}}}{\norm{\Gamma^{(g)}_i}}\right)\right) \\
        &= \Omega\left(\min\left(\Adv(f), \Adv(g) \right)\right) \\
        &= \Omega(\min(\Quant(f), \Quant(g)))
     \end{align*}
     \vspace{-1pc}
\end{proof}

For our applications, it is sufficient to consider the $p=1$ case. Moreover, it is not clear how we would generalize the above proof for $p > 1$ since we will need to consider sets of indices that contains inputs to $f$ and inputs to $g$. Therefore, we leave extending \Cref{lem:quant_correlated_functions} to $p$-parallel quantum query complexity as future work.

\subsection{Separations between randomized and deterministic complexities} \label{subsec:rand_det_separation_partial}
We begin by showing an unbounded separation for partial functions between randomized and deterministic parallel query complexities, while maintaining no sequential separation. 

\subsubsection{Partial function separation}
We will choose $f$ to be the pointer chasing function (\Cref{prob:pointer-chasing}) with instance size $N$ and chain length $k=\sqrt{N}$ and $g$ to be the balanced composition of \textsc{parity} and the \textsc{Deutsch Jozsa modified} function (\Cref{prob:djvar}). Then, we define \textsc{deutsch-josza ana function} as the function $\correlated(f,g)$ (\Cref{prob:correlated}). Precisely,

\begin{prob}[\textsc{deutsch-josza ana}] Let $N=2^n$, $f = \textsc{pointer}_{N,\sqrt{N}}$ and $g=\bigoplus_{\sqrt{N}} \circ \textsc{dj}_{\sqrt{N}}$.\label{prob:ANA-D-vs-R}

\begin{problem}
        \probleminput{Oracles for strings $X \in \{0,1\}^{Nn}$ and $Y \in \{0,1\}^{N}$.}
        \problempromise{$Y$ satisfies the promise of \textsc{Deutsch Jozsa modified} \emph{(\Cref{prob:djvar})}, and $f(X)=g(Y)$.}
        \problemquestion{Decide $f(X)=g(Y)$.}
    \end{problem}
\end{prob}

We will sometimes refer to the \textsc{deutsch-josza ana} function on parameter $N$ as $\djana_N$.

Below, we show that the \textsc{deutsch-josza ana} function satisfies our criteria of a parallel separation without a sequential separation for randomized vs deterministic complexities. In the sequential case, a strategy for both the deterministic and the randomized algorithm is to follow the pointer to compute \textsc{deutsch-jozsa ana} in \(\tilde O(\sqrt N)\) queries, and lower bounds for \(f\), \(g\) and \(\correlated(f, g)\) show neither can do better. In the parallel case with \(\tilde \Theta(\sqrt N)\) parallelism, the randomized algorithm can solve \(g\) in $1$ round because the \textsc{parity} component of \(g\) can be computed in parallel and \textsc{deutsch jozsa modified} requires $1$ query. However, the deterministic algorithm is still forced to spend \(\sqrt N\) queries for \textsc{deustsch jozsa modified} even if it can solve parity in parallel. Formally, we prove the theorem:

\begin{theorem} \label{thm:non-seq_DvsR_parallel_separation}
    Let $h$ be the \textsc{deutsch-josza ana} function as defined in \emph{\Cref{prob:ANA-D-vs-R}}. Then for some $p = \tilde{\Theta}(\sqrt{N})$,
    \begin{enumerate}[label=(\roman*)]
        \item \label{itm:non-seq_DvsR_parallel_separation_1} $\emph{\Det}(h) = O\left(\sqrt{N}\right)$,
        \inlineitem \label{itm:non-seq_DvsR_parallel_separation_2} $\emph{\Rand}(h) = \Omega\left(\sqrt{N}\right)$, \inlineitem \label{itm:non-seq_DvsR_parallel_separation_3} $\emph{\Det}^{p \parallel}(h) = \tilde{\Omega}\left(\sqrt{N}\right)$, \inlineitem \label{itm:non-seq_DvsR_parallel_separation_4} $\emph{\Rand}^{p \parallel}(h) = 1.$
    \end{enumerate} 
\end{theorem}

\begin{proof}
    We prove each of the items one by one as follows.
    \begin{enumerate} [label=(\roman*)]
        \item By \Cref{thm:pointer_chasing}, we know that $\Det(f) = O(\sqrt{N})$. Since $\Det(h) = \min(\Det(f), \Det(g))$ from \Cref{lem:det_correlated_functions}, we get $\Det(h) = O(\sqrt{N})$.
        \item From \Cref{thm:pointer_chasing}, we have $\Rand(f) = \Omega(\sqrt{N})$ and from \Cref{thm:QueryComplexityComposition}, we have $\Rand(g) = \Omega(\sqrt{N})$. It follows, from \Cref{lem:rand_correlated_functions}, that $\Rand(h) = \Omega(\min(\Rand(f), \Rand(g))) = \Omega(\sqrt{N})$.
        \item From \Cref{thm:pointer_chasing}, we know that $\Det^{p \parallel}(f) = \Omega(\sqrt{N})$ and from \Cref{thm:QueryComplexityComposition}, we have that $\Det(g) = \Omega(N)$ so $\Det^{p \parallel}(g) = \Omega(N/p) = \tilde{\Omega}(\sqrt{N})$. Since $\Det^{p \parallel}(h) = \min(\Det^{p \parallel}(f), \Det^{p \parallel}(g))$ from \Cref{lem:det_correlated_functions}, we get $\Det^{p \parallel}(h) = \tilde{\Omega}(\sqrt{N})$.
        \item Since \textsc{Deutsch Jozsa modified} is defined to have randomized query complexity $1$ and \textsc{parity} can be perfectly parallelized, we have $\Rand^{p \parallel}(g) = 1$. Solving $g$ is sufficient to solve $h$ so $\Rand^{p \parallel}(h) = 1$.
    \end{enumerate}
    \vspace{-1pc}
\end{proof}

\subsubsection{Total function separation}
We totalize the partial function obtained in the previous section to show an unbounded total function parallel separation without sequential separation (between randomized and deterministic complexities).

\begin{prob}[\textsc{cheatsheet deutsch-josza ana}] \label{prob:djcs}
    Let $f:\{0,1\}^N \to \{0,1\}$ and $g:\{0,1\}^{N^2} \to \{0,1\}$ be  defined as $f = \djana_N$ (see \emph{\Cref{prob:ANA-D-vs-R}}) and $g=(\andor)_{N^2}$. Let $c = 10 \log N$.
    \begin{problemTotal}
        \probleminput{An oracle for a string $X \in \{0,1\}^{M}$ (with $M = \tilde{\Theta}(N^{12})$\footnote{Since a certificate of $\andor$ on $N^2$ bits could be represented using $\tilde{\Theta}(N)$ bits, we will have the number of input bits for the function $(f \circ g)^c_{\cheatsheet}$ to be $\tilde{\Theta}(N^{12})$.}).}
        \problemquestion{Compute $(f \circ g)^c_{\cheatsheet}(X)$ (see \emph{\Cref{def:cheatsheet}}).}
    \end{problemTotal}
\end{prob}

Next, we show that this function demonstrates unbounded parallel advantage without sequential advantage.

\begin{theorem} \label{thm:no_seq_parallel_sepa_RvsD}
    Let $h$ be the \textsc{cheatsheet deutsch-josza ana} function as defined in \emph{\Cref{prob:djcs}}. Then for some $p = \Theta(N^{5/2} \log^3 N)$,
    \begin{enumerate}[label=(\roman*)]
        \item $\emph{\Det}(h) = \tilde{O}\left(N^{5/2}\right)$,
        \inlineitem $\emph{\Rand}(h) = \tilde{\Omega}\left(N^{5/2}\right)$, \inlineitem $\emph{\Det}^{p \parallel}(h) = \tilde{\Omega}\left(\sqrt{N}\right)$, \inlineitem $\emph{\Rand}^{p \parallel}(h) \leq 3.$
    \end{enumerate}
\end{theorem}

\begin{proof}
    We will prove each of the statements one by one. Recall that $f = \djana_N$ and $g = (\andor)_{N^2}$ in \Cref{prob:djcs}.
    \begin{enumerate}[label=(\roman*)]
        \item By \cref{itm:non-seq_DvsR_parallel_separation_1} of \Cref{thm:non-seq_DvsR_parallel_separation} and \Cref{thm:QueryComplexityComposition}, we know that $\Det(f \circ g) = O(N^{5/2})$. It follows, by \cref{itm:det_cheatsheet_upperbound} of \Cref{prop:cheatsheet_seq_upperbound}, that $\Det(h) = \tilde{O}(N^{5/2})$.
        \item By \cref{itm:non-seq_DvsR_parallel_separation_2} of \Cref{thm:non-seq_DvsR_parallel_separation} and \Cref{thm:QueryComplexityComposition}, we know that $\Rand(f \circ g) = \Omega(N^{5/2})$. It follows, by \Cref{thm:cheatsheetcomplexities}, that $\Rand(h) = \tilde{\Omega}(N^{5/2})$. 
        \item By \cref{itm:non-seq_DvsR_parallel_separation_3} of \Cref{thm:non-seq_DvsR_parallel_separation}, we know that $\Det^{\floor{p/N^2} \parallel}(f) = \tilde{\Omega}(\sqrt{N})$. Since $\Det(g) = \Theta(N^2)$, it follows, by  \Cref{thm:det_parallel_comp}, that $\Det^{p \parallel}(f \circ g) = \Omega\left(\Det^{\floor{p/N^2} \parallel}(f)\right) = \Omega(\sqrt{N})$. Therefore, using \Cref{thm:cheatsheetdetparallelimpact}, we have that $\Det^{p \parallel}(h) = \tilde{\Omega}(\sqrt{N})$. 
        \item The desired result directly follows from \cref{itm:rand_parallel_cheatsheet_upperbound} of \Cref{prop:cheatsheet_parallel_upperbound}.
    \end{enumerate}
    \vspace{-1pc}
\end{proof}

The previous results are shown with an input of size \(\tilde \Theta (N^{12})\). We state the corollary of the above theorem below using inputs of size \(M\) to make clear the exponents in the separation.

\begin{corollary}
    There exists a total Boolean function $h: \{0,1\}^M \to \{0,1\}$ and an integer $p \in [M]$ such that 
    \begin{enumerate}[label=(\roman*)]
        \item $\emph{\Det}(h) = \tilde{O}\left(M^{5/24}\right)$,
        \inlineitem $\emph{\Rand}(h) = \tilde{\Omega}\left(\emph{\Det}(h)\right)$, 
        \inlineitem $\emph{\Det}^{p \parallel}(h) = \tilde{\Omega}\left(M^{1/24}\right)$, \inlineitem $\emph{\Rand}^{p \parallel}(h) \leq 3$.
    \end{enumerate}
\end{corollary}

\subsection{Separations between quantum and randomized complexities}
 \label{subsec:quant_rand_separation_partial}
Using the same framework, we can separate quantum and randomized parallel query complexities, while maintaining no sequential separation. We first give partial function separations, and then totalize them. 

\subsubsection{Partial function separation}
We will choose $f$ to be the pointer chasing function (\Cref{prob:pointer-chasing}) with instance size $N$ and chain length $k=N^{1/3}$ and $g$ to be the composition of \textsc{parity} with instance size $N^{1/3}$ and the \textsc{Forrelation} function (\Cref{prob:Forrelation}) with instance size $N^{2/3}$. Then, we define \textsc{forrelation ana function} as the function $\correlated(f,g)$ (\Cref{prob:correlated}). Precisely,

\begin{prob}[\textsc{forrelation ana function}] Let $N=2^n$, $f = \textsc{pointer}_{N,N^{1/3}}$ and $g=\bigoplus_{N^{1/3}} \circ \textsc{for}_{N^{2/3}}$.\label{prob:ANA-R-vs-Q}
\begin{problem}
        \probleminput{Oracles for strings $X \in \{0,1\}^{Nn}$ and $Y \in \{0,1\}^{N}$.}
        \problempromise{$Y$ satisfies the promise of \textsc{Forrelation} \emph{(\Cref{prob:Forrelation})}, and $f(X)=g(Y)$.}
        \problemquestion{Decide $f(X)=g(Y)$.}
    \end{problem}
\end{prob}
Below, we show that the \textsc{forrelation ana} function satisfies our criteria of a parallel separation without a sequential separation for quantum vs randomized complexities. In the sequential case, a strategy for both the quantum and the randomized algorithm is to follow the pointer to compute \textsc{forrelation ana} in \(\tilde O(N^{1/3})\) queries, and lower bounds for \(f\), \(g\) and \(\correlated(f, g)\) show neither can do better. In the parallel case with \(\tilde \Theta(N^{1/3})\) parallelism, the quantum algorithm can solve \(g\) in 1  query because the \textsc{parity} component of \(g\) can be computed in parallel and \textsc{forrelation} is takes just \(1\) query. However, the randomized algorithm is still forced to spend \(N^{1/3}\) queries for \textsc{forrelation} even if it can solve parity in parallel. Formally, we prove the theorem:
\begin{theorem} \label{thm:non-seq_RvsQ_parallel_separation}
    Let $h$ be the \textsc{forrelation ana function} as defined in \emph{\Cref{prob:ANA-R-vs-Q}}. Then for some $p = \tilde{\Theta}(N^{1/3})$,
    \begin{enumerate}[label=(\roman*)]
        \item \label{itm:non-seq_RvsQ_parallel_separation_1} $\emph{\Rand}(h) = \tilde O\left(N^{1/3}\right)$,
        \inlineitem \label{itm:non-seq_RvsQ_parallel_separation_2} $\emph{\Quant}(h) = \Omega\left(N^{1/3}\right)$, \inlineitem \label{itm:non-seq_RvsQ_parallel_separation_3} $\emph{\Rand}^{p \parallel}(h) = \tilde{\Omega}\left(N^{1/3}\right)$, \inlineitem \label{itm:non-seq_RvsQ_parallel_separation_4}$\emph{\Quant}^{p \parallel}(h) = 1.$
    \end{enumerate}
\end{theorem}

\begin{proof}
    We prove each of the items one by one as follows.
    \begin{enumerate} [label=(\roman*)]
        \item By \Cref{thm:pointer_chasing}, we know that $\Rand(f) = O(N^{1/3})$. Since $\Rand(h) = O(\min(\Rand(f), \Rand(g)))$ from \Cref{lem:rand_correlated_functions}, we get $\Rand(h) = O(N^{1/3})$.
        \item From \Cref{thm:pointer_chasing}, we have $\Quant(f) = \Omega(N^{1/3})$ and from \Cref{thm:QueryComplexityComposition}, we have $\Quant(g) = \Omega(N^{1/3})$. It follows, from \Cref{lem:rand_correlated_functions}, that $\Quant(h) = \Omega(\min(\Quant(f), \Quant(g))) = \Omega(N^{1/3})$.
        \item From \Cref{thm:pointer_chasing}, we know that $\Rand^{p \parallel}(f) = \Omega(N^{1/3})$ and from \Cref{thm:QueryComplexityComposition}, we have that $\Rand(g) = \Omega(N^{1/3} \cdot \sqrt{N^{2/3}}) = \Omega(N^{2/3})$ so $\Rand^{p \parallel}(g) = \Omega(N^{2/3}/p) = \tilde{\Omega}(N^{1/3})$. Since $\Rand^{p \parallel}(h) = \Omega(\min(\Rand^{p \parallel}(f), \Rand^{p \parallel}(g)))$ from \Cref{lem:det_correlated_functions}, we get $\Rand^{p \parallel}(h) = \tilde{\Omega}(N^{1/3})$.
        \item Since \textsc{Forrelation} has quantum query complexity $1$ and \textsc{parity} can be perfectly parallelized, we have $\Quant^{p \parallel}(g) = 1$. Solving $g$ is sufficient to solve $h$ so $\Quant^{p \parallel}(h) = 1$.
    \end{enumerate}
    \vspace{-1pc}
\end{proof}

\subsubsection{Total function separation}
We totalize the partial function obtained in the previous section to show an unbounded total function parallel separation without sequential separation (between quantum and randomized complexities).

\begin{prob}[\textsc{cheatsheet forrelation ana}] \label{prob:forcs}
    Let $f:\{0,1\}^N \to \{0,1\}$ and $g:\{0,1\}^{N^2} \to \{0,1\}$ be  defined as $f = \forrelationana_N$ (see \emph{\Cref{prob:ANA-R-vs-Q}}) and $g=\bkk_{N^2}$ (see \emph{\Cref{prob:bkk}}). Let $c = 10 \log N$.
    \begin{problemTotal}
        \probleminput{An oracle for a string $X \in \{0,1\}^{M}$ (with $M = \tilde{\Theta}(N^{12})$\footnote{Since a certificate of $\andor$ on $N^2$ bits could be represented using $\tilde{\Theta}(N)$ bits, we will have the number of input bits for the function $(f \circ g)^c_{\cheatsheet}$ to be $\tilde{\Theta}(N^{12})$.}).}
        \problemquestion{Compute $(f \circ g)^c_{\cheatsheet}(X)$ (see \emph{\Cref{def:cheatsheet}}).}
    \end{problemTotal}
\end{prob}

We use the above function to prove an unbounded advantage in parallel query complexities but no sequential advantage between randomized and quantum algorithms.

\begin{theorem} \label{thm:no_seq_parallel_sepa_QvsR}
    Let $h$ be the \textsc{cheatsheet forrelation ana} function as defined in \emph{\Cref{prob:forcs}}. Then for some $p = \Theta(N^{7/3} \log^3 N)$,
    \begin{enumerate}[label=(\roman*)]
        \item $\emph{\Rand}(h) = \tilde{O}\left(N^{7/3}\right)$,
        \inlineitem $\emph{\Quant}(h) = \tilde{\Omega}\left(N^{7/3}\right)$, \inlineitem $\emph{\Rand}^{p \parallel}(h) = \tilde{\Omega}\left(N^{1/3}\right)$, \inlineitem $\emph{\Quant}^{p \parallel}(h) \leq 3$.
    \end{enumerate}
\end{theorem}

\begin{proof}
    We will proof each of the statements one by one. Recall that $f = \djana_N$ and $g = (\andor)_{N^2}$ in \Cref{prob:djcs}.
    \begin{enumerate}[label=(\roman*)]
        \item By \cref{itm:non-seq_RvsQ_parallel_separation_1} of \Cref{thm:non-seq_RvsQ_parallel_separation} and \Cref{thm:QueryComplexityComposition}, we know that $\Rand(f \circ g) = O(N^{7/3})$. It follows, by \cref{itm:rand_cheatsheet_upperbound} of \Cref{prop:cheatsheet_seq_upperbound}, that $\Rand(h) = \tilde{O}(N^{7/3})$.
        \item We know, from \Cref{thm:bkk}, that $\Quant(g) = \tilde\Omega(N^2)$. By \cref{itm:non-seq_RvsQ_parallel_separation_2} of \Cref{thm:non-seq_RvsQ_parallel_separation} and \Cref{thm:QueryComplexityComposition}, we have that $\Quant(f \circ g) = \tilde{\Omega}(N^{7/3})$. It follows, by \Cref{thm:cheatsheetcomplexities}, that $\Quant(h) = \tilde{\Omega}(N^{7/3})$. 
        \item By \cref{itm:non-seq_RvsQ_parallel_separation_3} of \Cref{thm:non-seq_RvsQ_parallel_separation}, we know that $\Rand^{\floor{p/N^2} \parallel}(f) = \Omega(N^{1/3})$. It follows, by \Cref{thm:rand_comp}, that $\Rand^{p \parallel}(f \circ g) = \Omega\left(\Rand^{\floor{p/N^2} \parallel}(f)\right) = \tilde{\Omega}(N^{1/3})$. Therefore, using \Cref{thm:cheatsheetrandparallelimpact}, we have that $\Rand^{p \parallel}(h) = \tilde{\Omega}(N^{1/3})$. 
        \item The desired result directly follows from \cref{itm:quant_parallel_cheatsheet_upperbound} of \Cref{prop:cheatsheet_parallel_upperbound}.
    \end{enumerate}
    \vspace{-1pc}
\end{proof}

The above theorem was proven for a function using input size \(\tilde \Theta(N^{12})\). We restate it as a corollary below using input size \(M\) to make clear the size of the exponents in the separation.

\begin{corollary}
    There exists a total Boolean function $h: \{0,1\}^M \to \{0,1\}$ and an integer $p \in [M]$ such that 
    \begin{enumerate}[label=(\roman*)]
        \item $\emph{\Rand}(h) = \tilde{O}\left(M^{7/36}\right)$,
        \inlineitem $\emph{\Quant}(h) = \tilde{\Omega}\left(\emph{\Det}(h)\right)$, 
        \inlineitem $\emph{\Rand}^{p \parallel}(h) = \tilde{\Omega}\left(M^{1/36}\right)$, \inlineitem $\emph{\Quant}^{p \parallel}(h) \leq 3$.
    \end{enumerate}
\end{corollary}

\section{Separations with two layers of adaptivity}
\label{sec:2Adaptive}
So far we have largely been dealing with separations in the depth (or number of layers or rounds) of an algorithm given a fixed width (or parallelism), that is, separations in \(\Det^{p \parallel}(f)\), \(\Rand^{p \parallel}(f)\) and \(\Quant^{p \parallel}(f)\). The relevant dual question to ask is given a fixed depth \(k\), what total function separation is possible for the widths, denoted as \(\Det^{k \perp}(f)\), \(\Rand^{k \perp}(f)\) and \(\Quant^{k \perp}(f)\)? Moreover, what is the least depth \(k\) for which a polynomial separation is possible between the widths? From \cite{montanaro2010nonadaptive}, we know that for \(k = 1\), for all total \(f\), there can be at most a constant factor of separation, that is \(\Det^{1 \perp}(f) = \Theta(\Rand^{1 \perp}(f)) = \Theta(\Quant^{1 \perp}(f))\). From \Cref{thm:cheatsheetQpvsRpSeparation} and \Cref{thm:cheatsheetRpvsDpSeparation} in this work, we know that at \(k = 3\), there is a polynomial separation. In particular,  there exist functions \(f\) and \(g\) such that \(\Rand^{3 \perp}(f) = \tilde O \left (D^{3 \perp}(f)^{1-\delta} \right )\) and \(\Quant^{3 \perp} (g) = \tilde O \left (\Rand^{3 \perp}(g)^{1-\delta'} \right )\) for some \(\delta, \delta' > 0\). Thus, the pertinent question is: is it possible to have a polynomial separation for total functions at \(k = 2\)? In this section we show that it is indeed possible---two layers of adaptivity are sufficient for a polynomial separation in parallelism or width between deterministic and randomized, and randomized and quantum algorithms. More generally, we show:

\begin{theorem}
\label{thm:twoAdaptiveInformalTheorem}
(Informal) Any partial function \(f\) that shows a polynomial separation in parallelism with one layer of adaptivity ($\Rand$ vs $\Det$ or $\Quant$ vs $\Rand$) can be converted into a total function \(h\) that shows a polynomial separation in parallelism with two layers of adaptivity. 
\end{theorem}

Before showing how we can obtain \(h\) given \(f\), we first give some intuition. We discuss here the $\Quant$ vs $\Rand$ separation but the same idea holds for $\Rand$ vs $\Det$. Let's start with the canonical cheat sheet function (\Cref{def:Canonicalcheatsheet}) used to show the separations in \Cref{thm:cheatsheetQpvsRpSeparation} and \Cref{thm:cheatsheetRpvsDpSeparation}. The separation here requires three rounds of adaptivity. At a high level, the upper bound for \(Q\) is given by doing the following in the three rounds:
\begin{enumerate}
\item Solve the hard partial function to obtain an address.
\item Read the certificate at the address in the cheat sheet. 
\item Use the certificate to go back and verify whether the input was in the domain of the hard partial function.
\end{enumerate}

\(R\) is unable to solve the hard partial function and obtain the location of the cheat sheet in the same amount of parallelism, hence the separation.

\paragraph{} Our challenge is to compress the above into two rounds. To achieve this, we remove the data dependence from step (1) to step (2). In particular, we will give the certificates of the hard partial function in a given fixed input position. However, we would still like to force the algorithm to solve the problem before (or in parallel to) reading the certificates: to accomplish this, we will have a third ``data'' section of the input. The output of the function will depend on a certain bit of the data register, specifically the one pointed to by the output of the hard partial function. In this way, an algorithm capable of solving the hard partial function without the certificates can find the data bit in the second round of queries. However, an algorithm that needs the certificates will not be able to learn the data bit until the third round. This function is still total, because the certificates can---also in the second round---be used to verify that the partial function input is from the domain.

\paragraph{} Note a subtle point in the prior intuition: an algorithm can learn a certificate of the hard partial function before it solves said function, in seeming contrast to the cheatsheet construction. If the \textit{form} of the certificate reveals the output of the function (e.g. for $\andor$, a minimal YES certificate points to a single position in each block, whereas a minimal NO certificate points to a full block), then the algorithm doesn't need to actually solve the function in the first round. We avoid this using so-called ``bi-certificates'' for the $\andor$ function. These are pointers to both an entire block, and a single index in each block, such that either a YES or NO instance could be certified. In this way, even after learning all the bi-certificates, an algorithm must still query the hard partial function to learn its output. 

\paragraph {} We now define the function \(h\) given the partial function \(f\). Then, in the following sections, we use \(h\) to show separations between \(\Det\) and \(\Rand\), and \(\Rand\) and \(\Quant\).

\begin{prob}
\label{prob:twoAdaptiveFunction}
\textsc{2-Adaptive-F} Let \(f:\{0,1\}^N \rightarrow \{0,1\}\) be a partial function. The input bit string to \(h: \{0,1\}^{\tilde \Theta(N^3)} \rightarrow \{0,1\}\) is divided into three components:  
\begin{enumerate}
\item \Add: The address component of size = \(\tilde \Theta(N^3)\)
\item \Bc: The bicertificate component of size = \(\tilde \Theta(N^2)\)
\item \Dt: The data component of size = \(N^3\)
\end{enumerate}
We then make the following indexing/ definitions:
\begin{enumerate}
\item Divide \(\Add\) into \(3 \log N\) segments of size \(N^3\) each, calling each segment \(\Add[i]\). Further divide each segment into \(N\) sub segments of size \(N^2\) each, calling each sub segment \(\Add[i,j]\). Each sub segment is further divided into \(N\) sized blocks \(\Add[i, j, k]\).
\item Similarly, divide \(\Bc\) into \(3 \log N\) segments of size \(\tilde \Theta(N^2)\) each, calling each segment \(\Bc[i]\). Further divide each segment into \(N\) sub segments of size \(\tilde \Theta(N)\)\footnote[1]{The \(\tilde \Theta\) hides a \(\log\) and a constant factor. It comes from the number of bits required to specify 2\(N\) locations of the \(\Add\) bit string.} calling each sub segment \(\Bc[i, j]\). \(\Bc[i, j]\) is called the ``bicertificate" corresponding to \(\Add[i, j]\).
\item Define the ``f-input bits" \(\In[i,j] = \andor(\Add[i,j])\), and the ``f-input" \(\In[i] = \In[i,0] \; \concat \; \In[i,1] \; \concat \; \In[i,2] \; ... \; \In[i,N-1]\)
\item If each \(\In[i]\) is in the domain of \(f\), define the `` target address bits" \(\tg[i] = f(\In[i])\) and the ``target address" \(\tg = \tg[0] \; \concat \; \tg[1] \; \concat \;  \tg[2] \; ... \;  \tg[3 \log N-1]\)
\item A valid bicertificate is a set of locations that could form a zero-certificate concatenated with a set of locations that could form a one-certificate (the actual instance is not taken into account). In the case of \(\andor\), a valid bicertificate \(\Bc[i,j]\) for \(\Add[i, j]\) is locations corresponding to any block \(\Add[i, j, k]\) (the zero certificate) concatenated with a set of one location in \(\Add[i, j, k]\) for each \(k\) (the one certificate).   
\item The intersection point \(\Ip[i, j]\) defined for a valid \(\andor\) bicertificate \(\Bc[i, j]\) refers to the single location that is part of both the zero- and the one- certificate. 
\end{enumerate}

Then, a given input is a YES instance to \(h\) if and only if:
\begin{enumerate}
\item For all \(i \in [0...3 \log N-1]\) and \(j \in [0 ... N-1]\), we have \(\Bc[i, j]\) is a valid \(\andor\),  
\item Every \(\Bc[i, j]\) certifies the output of \(\andor(\Add[i, j])\),
\item All \(\In[i]\) are within the domain of \(f\), and
\item \(\Dt[\tg] = 1\)
\end{enumerate}
\end{prob}

We can now use \textsc{2-Adaptive-F} (depicted in \Cref{fig:TwoAdaptive}) to show the desired separations. In \Cref{sec:2AdaptiveDetvsRand}, we show the deterministic vs randomized separation and in 
\Cref{sec:2AdaptiveRandvsQuant}, we show the randomized vs quantum separation. 

\subsection{Separation between randomized and deterministic complexities}
\label{sec:2AdaptiveDetvsRand}

\begin{theorem}
\label{thm:twoAdaptiveDvsR}
Let \(f: \mathcal D \rightarrow \{0,1\}\) for $\mathcal D \subset \{0,1\}^N$ be a partial function that satisfies the following where \(0 \leq s, l \leq 1\): 
\begin{align*}
\emph{\Rand}^{1 \perp}(f) = \tilde O(N^s)  \hspace{80pt} \emph{\Det}^{1 \perp}(f) = \tilde \Omega(N^l)
\end{align*}
Then \(h: \{0,1\}^{\tilde \Theta (N^3)} \rightarrow \{0,1\}\) as constructed in \Cref{prob:twoAdaptiveFunction} is a total function that satisfies:
\begin{align*}
\emph{\Rand}^{2 \perp}(h) = \tilde O(N^{s+2})   \hspace{80pt} \emph{\Det}^{2 \perp}(h) = \tilde \Omega(N^{l+2})
\end{align*}
\end{theorem}

\begin{proof}
The statement follows from the deterministic lower bound proved in \Cref{lem:twoAdaptiveDLowerBound} and the randomized upper bound proved in \Cref{lem:twoAdaptiveRUpperBound}.

\end{proof}

By taking \(f\) as the \textsc{Deutsch Jozsa} problem, for which \(\Rand^{1 \perp}(f) = k\) where \(k\) is constant, and \(\Det^{1 \perp}(f) = \tilde \Omega(N)\), we have the desired polynomial separation. 

\begin{lemma}
\label{lem:twoAdaptiveDLowerBound}
Let \(f: \mathcal D \rightarrow \{0,1\}\) for $\mathcal D \subset \{0,1\}^N$  be a partial function that satisfies \(\emph{\Det}^{1 \perp}(f) = \tilde \Omega(N^l)\) where \(0\leq l \leq 1\). Then \(h: \{0,1\}^{\tilde \Theta (N^3)} \rightarrow \{0,1\}\) as constructed in \Cref{prob:twoAdaptiveFunction} is a total function that satisfies \(\emph{\Det}^{2 \perp}(h) = \tilde \Omega(N^{l+2})\).
\end{lemma}
\begin{proof}
It suffices to show that for any 2 round deterministic algorithm that makes \(\tilde O \left ( N^{l+2-\epsilon} \right )\) queries for any \(\epsilon > 0\), it is always possible to adversarially answer the queries such that the algorithm fails.

\paragraph{Round 1} The algorithm submits \(\tilde O (N^{2+l-\epsilon})\) queries each to \(\Add\), \(\Bc\) and \(\Dt\). Answer the queries using the following strategy: 
\begin{enumerate}
\item For each segment \(\Add[i]\), the algorithm can fully query at most \(\tilde O({N^{2+l-\epsilon}}/N^2) = \tilde O(N^{l-\epsilon})\) \(\Add[i, j]\)'s, and there remain \(\tilde \Omega (N^\epsilon) \) \(\Add[i, j]\)'s with at least one bit unqueried.
\begin{enumerate}
\item For \(\Add[i, j]\) with at least one unqueried bit: Select a valid bicertificate \(\Bc[i, j]\) such that the intersection point \(\Ip[i, j]\) points to at an unqueried bit. For the queried bits, if \(l\) belongs to the one-certificate of \(\Bc[i, j]\), set \(\Add[i, j](l) = 1\), else set \(\Add[i, j](l) = 0\). This ensures that both \(\andor(\Add[i, j]) = 1\) and \(\andor(\Add[i, j]) = 0\) are still consistent with answered queries.  
\item For the \(\Add[i, j]\) which are fully queried: There are at most \(\tilde O(N^{l-\epsilon})\) such \(\Add[i, j]\)'s, therefore this fixes at most \(\tilde O(N^{l-\epsilon})\) bits of \(\In[i]\). Answer the queries to the \(\Add[i, j]\)'s such that \(\In[i]\) can be completed to both a \(0\)-instance and a \(1\)-instance of \(f\). This should always be possible, else it contradicts \(\Det^{1\perp}(f) = \tilde \Omega (N^l)\).
\end{enumerate}
\item For queries to \(\Dt\), answer anything. 
\end{enumerate}

\paragraph{Round 2} At this stage, note that the algorithm can fully learn \(\Bc\). The algorithm once again submits \(\tilde O (N^{2+l-\epsilon})\) queries to \(\Add\), \(\Bc\) and \(\Dt\). Answer the queries using the following strategy: 
\begin{enumerate}
\item Designate any location of \(\Dt\) that is unqueried in both rounds as the target address \(t'\). \(\Vert \Dt \Vert = N^3\) and only \(\tilde O(N^{2+l-\epsilon})\) queries have been made to it, so there should be several such locations. 
\item Setting \(\tg = t'\), for each \(i\), select an \(\In[i]\) consistent with the queries such that \(f(\In[i]) = \tg[i]\). In the first round we ensured this is always possible. Then answer queries to the \(\Add[i, j]\) such that \(\andor(\Add[i, j]) = \In[i, j]\).
\item For queries to \(\Dt\), answer anything.
\end{enumerate}
After Round 2, as the target location has not been queried, both 0 and 1 outputs of \(h\) are consistent with the queries made. Hence the deterministic algorithm fails, proving the claim. 
\end{proof}

\begin{lemma}
\label{lem:twoAdaptiveRUpperBound}
Let \(f: \mathcal D \rightarrow \{0,1\}\) for $\mathcal D \subset \{0,1\}^N$ be a partial function that satisfies \(\emph{\Rand}^{1 \perp}(f) = \tilde O(N^s)\) where \(0 \leq s \leq 1\). Then \(h: \{0,1\}^{\tilde \Theta (N^3)} \rightarrow \{0,1\}\) as constructed in \Cref{prob:twoAdaptiveFunction} is a total function that satisfies \(\emph{\Rand}^{2 \perp}(h) = \tilde O(N^{s+2})\).
\end{lemma}
\begin{proof}
The following randomized algorithm proves the statement.
\paragraph{Round 1} We have \(\Rand^{1 \perp}(f) = \tilde O(N^s)\). Therefore, by composition,  \(\Rand^{1 \perp}(f \circ \andor ) = \tilde O(N^{2+s})\). Thus in \(\tilde O(N^{s+2})\) queries the algorithm can determine each \(\tg[i] := f \circ \andor(\Add[i])\) with probability \(1/2 + c\) for any constant \(c\). By repeating the algorithm non-adaptively \(O(\log N)\) times and taking majority vote for each \(\tg[i]\), the algorithm can determine \(\tg\) with constant probability \(c'\). The algorithm can also query the full \(\Bc\) and check condition (1) in \Cref{prob:twoAdaptiveFunction}.
\paragraph{Round 2} The algorithm queries \(\Add\) at locations learned from \(\Bc\) to verify conditions (2) and (3) of \Cref{prob:twoAdaptiveFunction}. It also queries \(\Dt\) at the \(\tg\) determined in Round 1. If conditions 1, 2 and 3 are satisfied, it answers with \(\Dt[\tg]\), else it answers 0. With probability \(c'\) we have \(\Dt[\tg]\) is correct, else it is a uniformly random bit. Hence the algorithm succeeds with probability \(\geq \frac 1 2 + c'\).
\end{proof}

\subsection{Separation between quantum and randomized complexities}
\label{sec:2AdaptiveRandvsQuant}

\begin{theorem}
\label{thm:twoAdaptiveRvsQ}
Let \(f: \mathcal D \rightarrow \{0,1\}\) for $\mathcal D \subset \{0,1\}^N$ be a partial function that satisfies the following where \(0 \leq s, l \leq 1\): 
\begin{align*}
\emph{\Quant}^{1 \perp}(f) = \tilde O(N^s)  \hspace{80pt} \emph{\Rand}^{1 \perp}(f) = \tilde \Omega(N^l)
\end{align*}
Then \(h: \{0,1\}^{\tilde \Theta (N^3)} \rightarrow \{0,1\}\) as constructed in \Cref{prob:twoAdaptiveFunction} is a total function that satisfies
\begin{align*}
\emph{\Quant}^{2 \perp}(h) = \tilde O(N^{s+2})   \hspace{80pt} \emph{\Rand}^{2 \perp}(h) = \tilde \Omega(N^{l+2-o(1)})
\end{align*}
\end{theorem}

\begin{proof}
The statement follows from the randomized lower bound proved in \Cref{lem:twoAdaptiveRLowerBound} and the quantum upper bound proved in \Cref{lem:twoAdaptiveQUpperBound}.
\end{proof}
By taking \(f\) as the \textsc{Forrelation} problem, for which \(\Quant^{1 \perp}(f) = k\) where \(k\) is constant, and \(\Rand^{1 \perp}(f) = \tilde \Omega(\sqrt N)\), we have the desired polynomial separation. 

\begin{lemma}
\label{lem:twoAdaptiveRLowerBound}
Let \(f: \mathcal D \rightarrow \{0,1\}\) for $\mathcal D \subset \{0,1\}^N$ be a partial function that satisfies \(\emph{\Rand}^{1 \perp}(f) = \tilde \Omega(N^l)\) where \(0 \leq l \leq 1\). Then \(h: \{0,1\}^{\tilde \Theta (N^3)} \rightarrow \{0,1\}\) as constructed in \Cref{prob:twoAdaptiveFunction} is a total function that satisfies \(\emph{\Rand}^{2 \perp}(h) = \tilde \Omega(N^{l+2-o(1)})\).
\end{lemma}
\begin{proof} Using Yao's Lemma, observe that it is sufficient to show that there exists a distribution over inputs to \(h\) such that any deterministic algorithm making \(\tilde O(N^{l+2-\epsilon})\) queries in 2 rounds for any \(\epsilon > 0\) can succeed with probabilility at most \(1/2+o(1)\) over the distribution. We construct such a distribution below:

\begin{enumerate}
\item For \(i \in [0 \; ... \; 3 \log N - 1]\), select \(\In[i]\) from a hard distribution for \(f\).
\item For each \(\Bc[i, j]\), select valid bicertificates for \(\andor\) uniformly at random. 
\item Fixing each \(\In[i]\) and \(\Bc[i,j]\) as above fixes all the \(\Add[i,j]\) as follows. For each bit location \(l\) of \(\Add[i, j]\),
\begin{enumerate} 
\item If \(l = \Ip[i, j]\): set \(\Add[i, j](l) = \In[i, j]\)
\item Else if \(l\) belongs to the one-certificate of \(\Bc[i, j]\): set \(\Add[i, j](l) = 1\)
\item Else: set \(\Add[i, j](l) = 0\)
\end{enumerate}
\item For \(\Dt\), select a bitstring of length \(N^3\) uniformly at random. 
\end{enumerate}

\paragraph{Round 1} The deterministic algorithm submits \(\tilde O (N^{l+2-\epsilon})\) queries to \(\Add\), \(\Bc\) and \(\Dt\). Thus after the first round, it learns \(\tilde O (N^{l+2-\epsilon})\) bits in each \(\Add[i]\), all of \(\Bc\), and \(\tilde O (N^{l+2-\epsilon})\) bits of \(\Dt\). We will show that at this stage, among the inputs consistent with the queries, the probability distribution over all target addresses is almost uniform. That is \(\Pr[\tg = t'] \leq \frac{2} {N^3}\) for all \(t' \in [0, N^3 - 1]\). 
\paragraph{}
Recall that given \(\Bc[i, j]\), every bit of \(\Add[i, j]\) is fixed except \(\Add[i, j](\Ip[i, j])\) which is equal to \(\In[i, j]\). Thus a query to \(\Add[i, j](\Ip[i, j])\) fixes \(\In[i, j]\) whereas a query to any \(\Add[i, j](l)\) where \(l \neq \Ip[i, j]\) does not learn anything more about \(\In[i, j]\). The number of queries to \(\Add[i]\) is \(\tilde O (N^{l+2-\epsilon})\), and each of the contained \(N\) \(\Add[i, j]\)'s of length \(N^2\) have one intersection point placed uniformly at random. Thus, by linearity of expectation, the expected number of intersection points queried are \(\tilde O \left (\frac {N^{l+2-\epsilon}}{N^2}\right ) = \tilde O(N^{l-\epsilon})\). Note also that the probability that any given intersection point is queried is independent of whether any other is queried: this follows from their independent uniform random placement. Therefore, using Markov's inequality,
\begin{equation}
\label{eqn:det_tg_prob}
\Pr[\#\text{Intersection points queried in } \Add[i] \geq N^{l-\epsilon/2}] = \tilde O\left ( \frac {1}{N^{\epsilon/2}} \right )
\end{equation}

Thus with the remaining probability \(1 - \tilde O\left ( \frac {1}{N^{\epsilon/2}} \right )\), the number of intersection points queried in \(\Add[i]\) is \(< N^{l-\epsilon/2}\), or in other words the number of bits of \(\In[i]\) learned is  \(< N^{l-\epsilon/2}\). Thus, we have the probability of guessing \(\tg[i] := f(\In[i]) \) is 
\begin{equation}
\label{eqn:det_tg_condition}
\begin{aligned}
\Pr[\frac{\text{Determining }\; \tg[i]}{< N^{l-\epsilon/2} \text{ bits of } \In[i] \text{ known}}] & \leq   \frac 1 2 + \sqrt {\frac{\# \text{bits learned of } \In[i]}{\Rand^{1 \perp}(f)}} \\ &= \frac 1 2 + \sqrt{\tilde O \left ( \frac {N^{l-\epsilon/2}}  {N^l} \right)} \\ &= \frac 1 2 + \tilde O \left ( \frac 1  {N^{\epsilon/4}} \right). \\
\end{aligned}
\end{equation}
where the first inequality must be true by the following argument, where we recall that the location of the input bits learned to $f$ are independently, uniformly random. If it did not hold, then by repeating the process \(\frac{\Rand^{1 \perp}(f)}{\# \text{bits learned of } \In[i]}\) times, \(f\) can be computed with \(1/2 + const.\) probability in a total number of queries \(< R^{1\perp}(f)\) which is a contradiction. Thus from \Cref{eqn:det_tg_prob} and \Cref{eqn:det_tg_condition}, we have
\begin{equation}
\begin{aligned}
\Pr[\text{Determining }\; \tg[i]]
& \leq \tilde O \left ( \frac 1 {N^{\epsilon/2}} \right ) (1) +  \left ( 1 - \tilde O \left (\frac 1 {N^{\epsilon/2}} \right ) \right ) \left (\frac 1 2 + \tilde O \left ( \frac 1 {N^{\epsilon/4}} \right) \right ) \\
& = \frac 1 2 + \tilde O \left ( \frac 1 {N^{\epsilon/4}} \right )
\end{aligned}
\end{equation}

Thus, for all \(t' \in [0 \; ... \; N^3 - 1]\) and large enough $N$ and any constant $\epsilon > 0$,
\begin{equation}
\begin{aligned}
\Pr[\tg = t']
\leq \left (\frac 1 2 + \tilde O \left ( \frac 1 {N^{\epsilon/4}} \right ) \right )^ {3 \log N} 
\leq \frac {2} {N^3} 
\end{aligned}
\label{eqn:uniformity-target-bound}
\end{equation}

\paragraph{Round 2} Once again, the deterministic algorithm submits \(\tilde O (N^{l+2-\epsilon})\) queries to \(\Add\), \(\Bc\) and \(\Dt\). Let \(S\) be the set of all of the queries made to \(\Dt\), both in the first round and the second round. We know that \(\vert S \vert = \tilde O(N^{l+2 - \epsilon})\). By the union bound, we have
\begin{equation}
\begin{aligned}
\Pr[t' \in S]  \leq  \tilde O(N^{l+2 - \epsilon}) \left ( \frac 2 {N^3}\right ) 
= \tilde O \left (\frac 1 {N^\epsilon} \right ).
\end{aligned}
\end{equation}

Note that the first inequality holds for any constant $\epsilon > 0$, but not $\epsilon = 0$, as in this case \Cref{eqn:uniformity-target-bound} fails to hold. Because \(\Dt\) is selected uniformly at random, if \(t'\) is not queried, \(h\) can be answered with only probability \(1/2\). Thus, the probability of success over this distribution for the deterministic algorithm is \(\frac 1 2 + \tilde O \left (\frac 1 {N^\epsilon} \right )  = \frac 1 2 + o(1) \) as desired. This proves the statement.

\end{proof}

\begin{lemma}
\label{lem:twoAdaptiveQUpperBound}
Let \(f: \mathcal D \rightarrow \{0,1\}\) for $\mathcal D \subset \{0,1\}^N$ be a partial function that satisfies \(\emph{\Quant}^{1 \perp}(f) = \tilde O(N^s)\) where \(0 \leq s \leq 1\). Then \(h: \{0,1\}^{\tilde \Theta (N^3)} \rightarrow \{0,1\}\) as constructed in \Cref{prob:twoAdaptiveFunction} is a total function that satisfies \(\emph{\Quant}^{2 \perp}(h) = \tilde O(N^{s+2})\).
\end{lemma}
\begin{proof}
The following quantum algorithm proves the statement: 
\paragraph{Round 1} We have \(\Quant^{1 \perp}(f) = \tilde O(N^s)\). Therefore, by composition,  \(\Quant^{1 \perp}(f \circ \andor ) = \tilde O(N^{s+2})\). Thus in \(\tilde O(N^{s+2})\) queries the algorithm can determine each \(\tg[i] := f \circ \andor(\Add[i])\) with probability \(1/2 + c\) for any constant \(c\). By repeating the algorithm non-adaptively \(O(\log N)\) times and taking majority vote for each \(\tg[i]\), the algorithm can determine \(\tg\) with constant probability \(c'\). The algorithm can also query the full \(\Bc\) and check condition (1) in \Cref{prob:twoAdaptiveFunction}.
\paragraph{Round 2} The algorithm queries \(\Add\) at locations learned from \(\Bc\) to verify conditions (2) and (3) of \Cref{prob:twoAdaptiveFunction}. It also queries \(\Dt\) at the \(\tg\) determined in Round 1. If conditions 1, 2 and 3 are satisfied, it answers with \(\Dt[\tg]\), else it answers 0. With probability \(c'\) we have \(\Dt[\tg]\) is correct, else it is a uniformly random bit. Hence the algorithm succeeds with probability \(\geq \frac 1 2 + c'\).
\end{proof}

\section{A parallel quantum lower bound framework and applications}
\label{sec:LowerBounds}
In this section we prove a modified version of the parallel quantum adversary theorem which is weaker, yet easy to work with, giving explicit lower bounds for read once formulas and symmetric functions. This theorem allows one to derive parallel lower bounds from sequential upper bounds. To motivate this result, we first give a barrier result for applying the parallel combinatorial adversary for lower bounding the parallel query complexity of read once formulas.

\subsection{A parallel combinatorial adversary method barrier}
In this section, we exhibit a general certificate barrier for the parallel combinatorial adversary method (\Cref{thm:parallel_comb_adv}), which in particular implies that the method cannot give a tight lower bound for the simple two layer \(\andor\) tree. Notably, in the sequential case the combinatorial adversary method suffices to give a tight lower bound \cite{Barnum04}. This motivates our lower bound method, which overcomes this barrier while still being straightforward to apply. In particular, we show a lower bound for read-once formulas in \Cref{subsec:NNAdversaries}, of which a tight lower bound for the two layer \(\andor\) tree is a special case. The barrier is described in the following theorem, and can be seen as a parallel analog of a theorem by Zhang (\cite{Zhang04}, Theorem 10). Following the terminology therein, we use the term \(Alb_{p\parallel}\) to refer to the parallel combinatorial adversary bound.

\begin{theorem}
    \label{thm:parallel_comb_adv_barrier}
    Let \(f:\{0,1\}^N \rightarrow \{0,1\}\) be a total function, and let \(Alb_{p\parallel}(f), C_0(f) \text{ and } C_1(f)\) refer to the parallel combinatorial adversary method bound (\Cref{thm:parallel_comb_adv}), the \(0\)-certificate complexity and the \(1\)-certificate complexity of \(f\) respectively. Then, we have \(Alb_{p\parallel}(f) \leq \sqrt{\ceil{\frac {C_0} {p}} \ceil{\frac {C_1} {p}}}\).
\end{theorem}
\begin{proof}
Let \(X, Y, R, S, w(x, y), w_x, w_{x, S}, w_y \text{ and } w_{y,S}\) be defined as in \Cref{thm:parallel_comb_adv}. Thus, \begin{align*}
Alb_{p \parallel}(f) &= \max_{X, Y, R} \sqrt { \left ( \frac{\min_x{w_x}}{\max_{x, S} w_{x, S}} \right 
 ) \left (\frac{\min_y{w_y}}{\max_{y, S'} w_{y, S'}} \right )} \\
& \leq  \max_{X, Y, R} \sqrt { \min_{x, S} \left (\frac{w_x}{w_{x, S}} \right ) \min_{y, S'} \left (\frac{w_y}{w_{y, S'}} \right )}.
\end{align*}
Therefore, to prove the statement, it is sufficient to show that for any choice of \(X, Y, R\) there exists an \(x, y, S, S'\) such that \(\frac {w_x w_y}{w_{x, S} w_{y, S'}} \leq \ceil {\frac{C_0(f)}{p}} \ceil{\frac{C_1(f)}{p}}\). Let us assume this is false. That is for some choice of \(X, Y, R\), 
\begin{equation}
\label{eqn:parallel_comb_adv_barrier_proof_cont}
\forall x, y, S, S', \; \; w_x w_y > \ceil {\frac{C_0(f)}{p}} \ceil{\frac{C_1(f)}{p}} w_{x, S} \; w_{y, S'}
\end{equation}
Let \(C_{x,p}\) be a partition of the set of indices belonging to any certificate of \(x\) consisting of sets of size \(p\), as well as one set of size \(\leq p\) for the remaining indices. We have \(\vert C_{x,p} \vert \leq \ceil{\frac{C_0(f)}{p}}\). Define \(C_{y,p}\) likewise, and we have \(\vert C_{y,p} \vert \leq \ceil{\frac{C_1(f)}{p}}\). Thus, \Cref{eqn:parallel_comb_adv_barrier_proof_cont} implies \(\forall \; x, y\)
\begin{equation*}
  w_x w_y > \sum_{S \in C_{x,p}} w_{x, S} \; \sum_{S' \in C_{y, p}} w_{y, S'}.
\end{equation*}

Summing over \(x\) and \(y\), we have
\begin{equation}
\begin{aligned}
\label{eqn:parallel_comb_adv_barrier_proof_sum}
\sum_x w_x \sum_y w_y &> \sum_{x, S \in C_{x,p}} w_{x, S} \; \sum_{y, S' \in C_{y, p}} w_{y, S'} \\
&= \sum_{x, y} \left (\sum_{S \in C_{x,p} \; | \; x_S \neq y_S} w(x, y) \right ). \sum_{x, y} \left ( \sum_{S' \in C_{y,p} \; | \; x_{S'} \neq y_{S'}} w(x, y) \right )
\end{aligned}
\end{equation}

Now because \(f\) is a total function, for each \(x, y\) there must exist some \(S \in C_{x,p}\) such that \(x_S \neq y_S\). Likewise, there must exist some \(S' \in C_{y, p}\) such that \(x_{S'} \neq y_{S'}\). Thus, for each \(x, y\), the summation over \(S\) and \(S'\) in \Cref{eqn:parallel_comb_adv_barrier_proof_sum} contains at least one term. Thus, 
\begin{align*}
\sum_x w_x \sum_y w_y &> \sum_{x,y} w(x, y) . \sum_{x,y} w(x,y) \\
&> \sum_x w_x \sum_y w_y
\end{align*}
which is a contradiction. Hence, the statement follows.
\end{proof}

\begin{corollary}
    \label{corr:AND_OR_Barrier}
    The parallel combinatorial adversary method \Cref{thm:parallel_comb_adv} cannot provide a tight lower bound for \(\andor\).
\end{corollary}

\begin{proof}
The $\andor$ function is a special case of a read once formula, so by \Cref{theorem:lowerBoundReadOnce} we have $\Quant^{p\parallel}(\andor) = \Omega\left(\sqrt{\frac{{N}}{p}}\right)$. Furthermore, $\Quant^{p\parallel}(\andor) = \tilde O\left(\sqrt{\frac{N}{p}}\right)$ by parallel Grover search, so this is tight up to log factors in $N, p$. However, we have \(C_0(\andor) = C_1(\andor) = \sqrt N\). Thus, by \Cref{thm:parallel_comb_adv_barrier}, we have \(Alb_{p||}(\andor) \leq \ceil {\frac {\sqrt N} {p}} \).
\end{proof}

We note that the above proof only acts as a barrier for the parallel combinatorial adversarial method, but not the parallel version of the general positive weighted adversary method \cite{Ambainis02, Zhang04}. However, it is unclear how to assign weights using the general method to prove parallel lower bounds for read once formulas.

\subsection{Lower bounds from nearest neighbor adversaries}
\label{subsec:NNAdversaries}

For any total $f:\{0,1\}^N \rightarrow \{0,1\}$, call $\Gamma^{\NN}$ a nearest neighbor adversary matrix of $f$ if it is an adversary matrix that only places non-zero weight on pairs that differ at a single entry. Formally: \begin{align}
    \Gamma^{\NN}_{xy} &= \Gamma^{\NN}_{yx} \nonumber \\
    f(x)=f(y) \Rightarrow& \Gamma^{\NN}_{xy} = 0 \nonumber \\
    d(x,y) > 1 \Rightarrow& \Gamma^{\NN}_{xy} = 0 \label{eqn:adversaryMatrixNNReqs}
\end{align}
where $d(\cdot,\cdot)$ denotes hamming distance. Nearest neighbor adversary matrices were studied by Aaronson et al \cite{BetterbealsAaronson21}, where they are shown to be closely related to the so-called spectral sensitivity, $\lambda(f)$. To define this quantity for a boolean function $f:\{0,1\}^N \rightarrow \{0,1\}$, let $G_f$ be a subgraph of the boolean hypercube on $\{0,1\}^N$ with edges only between inputs distinguished by $f$. In particular, the vertex set of $G_f$ is $\{0,1\}^N$, and there is an edge between strings $x$ and $y$ if and only if (a) $d(x, y)=1$ and (b) $f(x) \neq f(y)$. Then let $A_f$ be the adjacency matrix of $G_f$; we define $\lambda(f) = \norm{A_f}$. The following theorem was established by Aaronson et al.

\begin{theorem}[\cite{BetterbealsAaronson21}, Theorem 10]
    For any total boolean function $f:\{0,1\}^N \rightarrow \{0,1\}$ and nearest neighbor adversary matrices $\Gamma^{NN}$ for $f$, we have \begin{align*}
        \max_{\Gamma^{NN}} \frac{\norm{\Gamma^{NN}}}{\max_{i \in [N]} \norm{\Gamma^{NN}}} = \lambda(f)
    \end{align*}
    \label{theorem:spectral-sens-equals-adv}
\end{theorem}
This quantity will play an important role in our parallel lower bound. A straightforward corollary of the above is that spectral sensitivity is upper bounded by quantum query complexity, $\lambda(f) = O(\Quant(f))$ for all total $f$.

Furthermore, let $\mathcal{F}_{p-\text{res}}^{(f)}$ denote the set of all functions obtained from $f$ by fixing all but $p$ bits in the input. Formally, a function $g: \{0,1\}^p \rightarrow \{0,1\}$ is in $\mathcal{F}_{p-\text{res}}^{(f)}$ if there is a subset $S \subset [N]$ of size $|S|=p$ and an assignment $A: ([N] \setminus S) \rightarrow \{0,1\}$ of the remaining bits such that for all $X$ satisfying $X_i=A(i)$ when $i \in ([N] \setminus S)$, we have $f(X)=g(X_S)$ where $X_S$ is the concatenation of $X_j$ for all $j \in S$.

\begin{theorem}
    \label{theorem:parallelLowerBoundNN}
    For any total boolean function $f: \{0,1\}^N \rightarrow \{0,1\}$, the parallel quantum query complexity satisfies the following. \begin{align*}
        \emph{\Quant}^{p\parallel}(f) &= \Omega\left(\lambda(f) \cdot \frac{1}{\max_{g \in \mathcal{F}_{p-\text{res}}^{(f)}}\lambda(g)}\right)
    \end{align*}
\end{theorem}

We defer the proof of this theorem to \Cref{sec:nnLowerboundProof}. A slightly weaker but sometimes easier to reason about formulation of this lower bound is given by recalling that $\lambda(f)=O(\Quant(f))$. This form, stated below, is the one we will apply to read-once formulas.

\begin{corollary}
    For any total boolean function $f: \{0,1\}^N \rightarrow \{0,1\}$, the parallel quantum query complexity satisfies the following. \begin{align*}
        \emph{\Quant}^{p\parallel}(f) &= \Omega\left(\lambda(f) \cdot \frac{1}{\max_{g \in \mathcal{F}_{p-\text{res}}^{(f)}}\emph{\Quant}(g)}\right)
    \end{align*}
    \label{corollary:parallelLowerBoundNN}
\end{corollary}

The first term in this expression is a standard boolean complexity measure that is well understood for many functions. The second term is the inverse of the sequential query complexity of a worst-case restriction of $f$, and so the second term can be lower bounded by giving an algorithm for worst-case restrictions of $f$ (upper bounding the denominator). This will allow us to prove parallel lower bounds with only sequential reasoning.

\subsection{Read once formulas}
\label{subsec:readOnceformulas}

Recall that a read-once formula is a boolean formula of variables $x_1,...,x_N$ that can be written using and $(\land)$, or $(\lor)$, and not $(\neg)$ in such a way that each variable appears exactly once. For example, the $\andor$ tree corresponding to \begin{align*}
    \psi_{\andor} &= (x_1 \lor ... \lor x_{\sqrt{N}}) \land ... \land (x_{N-\sqrt{N}+1} \lor ... \lor x_N)
\end{align*}
is a read once formula. It is well known that the $\andor$ tree has sequential quantum query complexity $\Omega(\sqrt{N})$ \cite{Barnum04,Reichardt11}. We give a lower bound for the parallel quantum query complexity of $\andor$ and more broadly any read-once formula, eliminating the possibility of large parallel query advantage in this case.
\begin{theorem}
    The $p$-parallel quantum query complexity of any read-once formula with $N$ variables, $f:\{0,1\}^N \rightarrow\{0,1\}$ satisfies the following. \begin{align*}
        \emph{\Quant}^{p\parallel}(f) &= \Omega\left(\sqrt{\frac{N}{p}}\right)
    \end{align*}
    \label{theorem:lowerBoundReadOnce}
\end{theorem}

\begin{proof}
    We will apply \Cref{corollary:parallelLowerBoundNN}, bounding the first and second term separately. A result of Barnum and Saks \cite{Barnum04} provides a $\Gamma$ for any read once formula $f$ that satisfies $\parallel \Gamma\parallel =\Omega(\sqrt{N})$ and $\parallel \Gamma_i\parallel =1$ for any $i\in[N]$ and is nearest neighbor. Hence we have $\lambda(f) = \Omega(\sqrt{N})$. To bound the second term we can simply observe that a $p$-restriction of a read-once formula is another read-once formula of size $p$. Consider a disjunction of the form \begin{align*}
        c&=x_1 \lor x_2 \lor ... \lor x_k
    \end{align*}
    If $x_i=0$ is fixed then it can simply be removed from $c$, if $x_i=1$ is fixed then $c$ can be removed and replaced with a $1$ literal (or vice versa if $\neg x_i$ appears). The rule for eliminating literals in a conjunction follows similarly. This procedure can be applied recursively until all fixed variables have been eliminated, leaving only $p$ many free variables for a $p$-restriction. The quantum query complexity of a size $p$ read-once formula is $O(\sqrt{p})$ \cite{Reichardt11}, so the second term is $\Omega(1/\sqrt{p})$.
\end{proof}

\subsection{Symmetric functions}
\label{subsec:SymmetricFunctions}

We can also use \Cref{corollary:parallelLowerBoundNN} to reproduce parallel quantum lower bounds for symmetric functions. These bounds were implicit in the work of Grover and Radhakrishnan \cite{Grover04}, but are arguably simpler to show once our theorem has been established. A symmetric function is a function $f: \{0, 1\}^N\rightarrow \{0, 1\}$ which satisfies \begin{align}
    f(x) &= f(\pi(x)) \text{ \qquad \qquad for all permutations } \pi. \nonumber
\end{align}
Or equivalently, $f(x)$ depends only on the hamming weight of $x$, $|x|$. We denote by $f_{|x|}$ the value of $f$ on inputs of hamming weight $|x|$. We recall a standard combinatorial adversary lower bound for symmetric functions.

Let $f$ be a non-constant symmetric total function from $\{0,1\}^N$ to $\{0, 1\}$, and let $t_f\leq N/2$ be chosen such that \begin{enumerate}[label=(\arabic*)]
    \item $f_{t_f}\neq f_{t_f+1}$ or $f_{N-t_f} \neq f_{N-t_f-1}$
    \item $t_f$ is chosen to minimize $|N/2-t_f|$
\end{enumerate}
WLOG let us assume that $f_{t_f} \neq f_{t_f+1}$ (if this is not the case then consider the complement function $\tilde{f}$ defined by $\tilde{f}(x)=f(\tilde{x})$, with $\tilde{x}$ being the bitwise complement). Consider the adversary matrix given by \begin{align}
    \Gamma_{xy} = \begin{cases}
        1 & \text{ if } |x|=t_f, |y|=t_f+1, |x \oplus y| = 1 \\
        0 & \text{ otherwise}
    \end{cases} \nonumber 
\end{align}
This provides a tight $\Omega(\sqrt{Nt_f})$ sequential lower bound \cite{Beals98}, and it is worth noting that this is also a nearest-neighbor adversary matrix. We are now ready to show the parallel lower bound.

\begin{theorem} \label{thm:symmetric_functions}
    For any non-constant symmetric function $f:\{0,1\}^N \rightarrow \{0, 1\}$, define $t_f\leq N/2$ such that \begin{enumerate}[label=(\arabic*)]
        \item $f_{t_f}\neq f_{t_f+1}$ or $f_{N-t_f} \neq f_{N-t_f-1}$
        \item $t_f$ minimizes $|N/2-t_f|$
    \end{enumerate}
    Then the parallel quantum query complexity of $f$ satisfies \begin{align*}
        \Quant^{p\parallel}(f) &= \Tilde{\Theta}\left(\sqrt{\frac{Nt_f}{p\min\{p, t_f\}}}\right)
    \end{align*}
\end{theorem}

\begin{proof}
    We will again apply \Cref{corollary:parallelLowerBoundNN}, bounding each term separately. From the preceding paragraph, we have \begin{align*}
        \lambda(f) &= \Omega\left(\sqrt{Nt}\right)
    \end{align*}
    bounding the first term. We will now upper bound $\max_{g \in \mathcal{F}_{p-\text{res}}^{(f)}}\Quant(g)$ by giving an efficient quantum algorithm for a $p$-restriction of $f$. Let $g$ be some $p$-restriction of $f$ and $A:[N] \setminus S \rightarrow \{0,1\}$ be the assignment to the remaining $N-p$ variables. Let there be $k$ many $1$'s in the assignment $A$.  $g$ is now a symmetric function ``shifted'' by $k$, i.e. we have \begin{align*}
        f_l = f_{l+1} (\Leftrightarrow) g_{l-k} = g_{l-k+1}
    \end{align*}
    it follows that $t_g \leq t_f$, and we obviously have $t_g \leq p$ as the domain is size $p$. Hence there is a quantum algorithm for $g$ with complexity $O(\sqrt{p \min \{p, t_f\}})$, and therefore \begin{align*}
        \frac{1}{\max_{g \in \mathcal{F}_{p-\text{res}}^{(f)}}\Quant(g)} &= \Omega\left(\frac{1}{\sqrt{p \min\{p, t_f\}}}\right) \\
         \Quant^{p\parallel}(f) &= \Omega\left(\sqrt{\frac{Nt_f}{p\min\{p, t_f\}}}\right) & \text{(\Cref{corollary:parallelLowerBoundNN})}
    \end{align*}

    This is tight by a small modification of the algorithm presented by Grover and Radhakrishnan \cite{Grover04}.
\end{proof}

\subsection{Proof of nearest neighbour lower bound}
\label{sec:nnLowerboundProof}
Now, we give a proof of the nearest neighbour lower bound stated in \Cref{theorem:parallelLowerBoundNN}. By \Cref{theorem:spectral-sens-equals-adv}, the following formulation is equivalent: 
\begin{theorem*}[Reformulation of \Cref{theorem:parallelLowerBoundNN}]
    For any total boolean function $f: \{0,1\}^N \rightarrow \{0,1\}$, and nearest neighbor adversary $\Gamma^{NN}$ for $f$, the parallel quantum query complexity satisfies the following. \begin{align*}
        \emph{\Quant}^{p\parallel}(f) &= \Omega\left(\frac{\norm{\Gamma^{NN}}}{\max_i \norm{\Gamma^{NN}_i}} \cdot \frac{1}{\max_{g \in \mathcal{F}_{p-\text{res}}^{(f)}}\lambda(g)}\right)
    \end{align*}
\end{theorem*}

\begin{proof}
Whenever  we \(\max\) over \(S\) in this proof, it is understood that \(S \subseteq [N], |S| = p\), and when we \(\max\) over \(i\), it is understood that \(i \in [N]\). From the parallel adversary method in \Cref{equation:parallelSpectralAdversary}, we have
\begin{align}
        \Quant^{p\parallel}(f) &= \Omega\left(\frac{\parallel \Gamma^{\NN}\parallel }{\max_{S} \parallel \Gamma^{\NN}_S\parallel } \right)
 \label{eqn:redefParallelAdvQuant}       
\end{align}
By \Cref{lem:NearestNeighbourBlockDiagonal}, we have that \(\Gamma^{\NN}_S\) is a block diagonal matrix. Labelling each block by \(b\) such that \(\left (\Gamma^{\NN}_S \right )^{(b)}\) refers to the \(b^{th}\) block matrix, we get 
\begin{equation}
\left \Vert \Gamma_S^{\NN} \right \Vert = \max_b \left \Vert \left (\Gamma_{S}^{\NN}\right )^{(b)}\right \Vert = \left \Vert \left (\Gamma_{S}^{\NN}\right )^{(b^*)}\right \Vert 
\label{eqn:adversaryMaxBlockNN}
\end{equation}
where $b^*$ is the block index which achieves the maximum in \Cref{eqn:adversaryMaxBlockNN}.

By \Cref{lem:NearestNeighbourBlockDiagonal}, we have that each of the block matrices \(\left (\Gamma_S^{\NN} \right )_b\) are in fact adversary matrices for some function \(g \in \mathcal{F}_{p-\text{res}}^{(f)}\). Therefore, for all \(b\), 
\begin{equation}
\frac{\left \Vert \left (\Gamma_{S}^{\NN} \right )^{(b)} \right \Vert} {\max_i \left \Vert \left (\Gamma^{\NN}_{S} \right )_i^{(b)} \right \Vert} \leq \max_{g \in \mathcal{F}_{p-\text{res}}^{(f)}} \Quant(g)
\label{eqn:adversaryInequalityNN}
\end{equation}
Since $\Gamma^{\NN}_S$ is block diagonal, $\Gamma^{\NN}_S[x,y] = 0$ whenever $x_{S^\complement} \neq y_{S^\complement}$. Therefore, $\left(\Gamma^{\NN}_{S}\right)_i = \Gamma^{\NN}_{S \cap \{i\}}$. In particular, $\max_i \left \Vert \left(\Gamma^{\NN}_{S}\right)^{(b)}_i  \right \Vert = \max_i \left \Vert \left(\Gamma^{\NN}_i\right)^{(b)}  \right \Vert$.
Putting this together with \Cref{eqn:adversaryMaxBlockNN} and \Cref{eqn:adversaryInequalityNN}, we get
\begin{equation}
\left \Vert \Gamma^{\NN}_S \right \Vert \leq \left (\max_i \left \Vert \left (\Gamma^{\NN}_{i} \right )^{(b^*)} \right \Vert \right ) \left ( \max_{g \in \mathcal{F}_{p-\text{res}}^{(f)}} \Quant(g) \right ).
\label{eqn:blockInequalityNN}
\end{equation}

which when plugged into \Cref{eqn:redefParallelAdvQuant} proves the theorem.
\end{proof}

\begin{lemma}
Let \(\Gamma^{\NN}\) be a Nearest Neighbour adversary matrix (\Cref{eqn:adversaryMatrixNNReqs}) and let $S \subseteq [N]$ with $|S| = p$. Then,
\begin{enumerate}
\item \label{part:NN1} \(\Gamma^{\NN}_S\) is a block diagonal matrix with blocks of size \(2^p \times 2^p\). 
\item \label{part:NN2} Each block of \(\Gamma^{\NN}_S\) is an adversary matrix of a function \(g \in \mathcal{F}_{p-\text{res}}^{(f)}\) where \(\mathcal{F}_{p-\text{res}}^{(f)}\) is as defined in \Cref{subsec:NNAdversaries}.

\end{enumerate}
\label{lem:NearestNeighbourBlockDiagonal}
\end{lemma}
\begin{proof}
Since we can relabel the inputs arbitrarily, we can assume, without loss of generality, that $S$ contains the last $p$ indices, that is $S = [N-p+1, N]$. Thus, we can view $\Gamma^{\NN}_S$ as a block matrix where each block, indexed by an assignment of $S^\complement$, is of size $2^p \times 2^p$. Recall that $\Gamma^{\NN}_S[x,y] \neq 0$ only if $x$ and $y$ differ at a single bit $i$ with $i \in S$. This means that for any $x,y$ with $x_{S^\complement} \neq y_{S^\complement}$, we have $\Gamma^{\NN}_S[x,y] = 0$. This completes the proof of \cref{part:NN1}.

For any assignment $b$ of $S^\complement$, we define the induced function $f_b:\{0,1\}^p \to \{0,1\}$ as
\begin{equation*}
    f_b(z) \defeq f(z \concat b) 
\end{equation*}
where $z \concat b$ corresponds to an input $x$ with $x_S = z$ and $x_{S^\complement} =  b$.

Clearly, \(f_b \in \mathcal{F}_{p-\text{res}}^{(f)}\). Let \(\left ( \Gamma^{\NN}_S \right )^{(b)} \) be the diagonal block of \(\left ( \Gamma^{\NN}_S \right )\) associated with the assignment $b$. It remains to verify that \(\left ( \Gamma^{\NN}_S \right )^{(b)} \) is an adversary matrix for $f_b$. We know that \(\left (\Gamma_S^{\NN} \right )^{(b)}\) is symmetric since \(\Gamma_S^{\NN}\) is symmetric. Moreover, for any $z, z' \in \{0,1\}^p$, we have
\begin{equation*}
    f_{b}(z) = f_b(z') \implies f(z \concat b) = f(z' \concat b) \implies (\Gamma_S^{\NN})_{z \concat b, z' \concat b} = 0 \implies \left ((\Gamma_S^{\NN})^{(b)} \right )_{z,z'} = 0.
\end{equation*}
This proves \cref{part:NN2}.
\end{proof}

\section*{Acknowledgments}
The authors thank Laxman Dhulipala for suggesting to investigate the relationship between quantum algorithms and parallelism, Luke Schaeffer and Chaitanya Karamchedu for helpful discussions, and Andrew Childs for valuable feedback on an earlier draft. We also thank an anonymous reviewer for largely simplifying our constructions in \Cref{sec:unbounded_genuine_separations}. ASG additionally thank Stacey Jeffery and Ronald de Wolf for helpful discussions. 

JC is supported by the US Department of Energy grant no.\ DESC0020264.
ASG is supported by the U.S. Department of Energy, Office of Science, Office of Advanced Scientific Computing Research, Accelerated Research in Quantum Computing and Quantum Testbed Pathfinder programs (award numbers DE-SC0020312 and DE-SC0019040). MV is supported by the US Department of Energy grant no.\ DESC0020264 as well as the US Department of Energy Nuclear Energy University Programs award number DE-NE0009417. 
\newpage

\bibliography{main}

\newpage

\appendix
\addtocontents{toc}{\protect\setcounter{tocdepth}{1}} 

\section{Composition theorems for parallel query complexity}
\label{sec:AppendixCompositionTheorem}
In this section, we will present composition theorems for parallel quantum query complexity that will be useful for us.

\subsection{A composition theorem for deterministic parallel query complexity} \label{sec:det_comp_thm}

In this section, we prove a deterministic  parallel composition theorem that is useful for showing deterministic parallel query lower bounds in \Cref{sec:unbounded_genuine_separations}. In particular, we show a tight (up to constant factors) composition theorem for the case when the deterministic query complexity of the inner function is close to being full. The same technique would apply for the case when the outer function has $\Theta(n)$ deterministic query complexity. We leave (dis)proving a general deterministic parallel composition theorem as an open problem.

\begin{theorem} \label{thm:det_parallel_comp}
    Let $f: \mathcal{D}_f \to \{0,1\}$ and $g: \mathcal{D}_g \to \{0,1\}$ be (partial) functions, where \(\mathcal D_f \subseteq \{0,1\}^N, \mathcal D_g \subseteq \{0,1\}^M\) such that $\emph{\Det}(g) = \Theta(M)$ and $p \in [MN]$. Then, $\emph{\Det}^{p \parallel}(f \circ g) = \Theta\left(\ceil{\frac{M}{p}} \emph{\Det}^{\ceil{p/M} \parallel}(f) \right)$.
\end{theorem}

\begin{proof}
    We first describe a deterministic $p$-parallel algorithm. Let $\mathcal{A}_f$ be an optimal deterministic $\ceil{p/M}$-parallel query algorithm for $f$. We will make queries in rounds where in each round, we will choose $\ceil{p/M}$ inputs for $f$ that the algorithm $\mathcal{A}$ would have chosen to compute and make $\ceil{M/p} p$ non-adaptive queries to completely determine the inputs of $g$ corresponding to these $\ceil{p/M}$ inputs. Thus, each round takes $\ceil{M/p}$ $p$-parallel queries to compute these $\ceil{p/M}$ inputs to $f$. Due to the optimality of $\mathcal{A}$, the total number of rounds needed to compute $f \circ g$ is $\Det^{\ceil{p/M}}(f)$ and, therefore, $\ceil{M/p} \Det^{\ceil{p/M}}(f)$ queries suffices.  

    For the lower bound, we will use the adversarial argument against a deterministic $p$-parallel algorithm $\mathcal{A}$ computing $f \circ g$. We consider two cases: i) $p \leq M$ and, ii) $p > M$. The first case is trivial since any $p$-parallel deterministic query algorithm must make $\Omega\left(\frac{\Det(f) \cdot \Det(g)}{p}\right) = \Omega\left(\frac{M}{p} \Det(f)\right)$ $p$-parallel queries to compute $f \circ g$ against an adversarial oracle. For the second case, our oracle will allow $\mathcal{A}$ to make $2p$-parallel queries as long as for each such query, $\mathcal{A}$ uses $p$ of these non-adaptive queries to find all the inputs to $p/M$ of the $g$ functions of its choice and, thus, determining $p/M$ inputs to $f$. This oracle is at least as powerful as a regular $p$-parallel query oracle since it allows computing $p/M$ inputs to $f$ along with making $p$ additional non-adaptive queries to the inputs to $g$. Thus, it suffices to lower bound the number of queries to this oracle. Since $\mathcal{A}$ is required to (non-adaptively) compute $p/M$ inputs to $f$ for each $2p$-parallel query that it makes, the minimal number of total inputs to $f$ that $\mathcal{A}$ must compute is $p/M \cdot \Det^{p/M \parallel}$. This means that $\mathcal{A}$ must query at least $p/M \cdot \Det(g) \Det^{p/M \parallel} = \Omega(p \Det^{p/M \parallel}$ total bits of input of $f \circ g$. It follows that $\mathcal{A}$ would need $\Omega(\Det^{p/M \parallel}$ $2p$-parallel queries to compute $f \circ g$.   
\end{proof}

We state the implication of this theorem that we use in \Cref{sec:unbounded_genuine_separations}.

\begin{corollary}
    Let $f: \mathcal{D}_f \to \{0,1\}$ be a (partial) function, where \(\mathcal D_f \subseteq \{0,1\}^N\) and $p \in [N^3]$. Then, $\emph{\Det}^{p \parallel}(f \circ (\andor)_{N^2}) = \Theta\left(\ceil{\frac{N^2}{p}} \emph{\Det}^{\ceil{p/N^2}}(f) \right)$.
\end{corollary}

We believe that the above theorem could be generalized to prove the following conjecture.
\begin{conjecture}
    Let $f: \mathcal{D}_f \to \{0,1\}$ and $g: \mathcal{D}_g \to \{0,1\}$ be (partial) functions, where \(\mathcal D_f \subseteq \{0,1\}^N, \mathcal D_g \subseteq \{0,1\}^M\)  and $p \in [MN]$. Then, $\emph{\Det}^{p \parallel}(f \circ g) = \Theta\left(\min_{q \in [M]} \emph{\Det}^{q \parallel}(g) \emph{\Det}^{\ceil{p/q} \parallel}(f) \right)$.
\end{conjecture}

\subsection{A composition theorem for randomized parallel query complexity} \label{sec:rand_comp_thm}

Below, we provide a randomized parallel composition theorem that will help us establish our separations in \Cref{sec:unbounded_genuine_separations}. This result is analogous to and inspired by Theorem 5 of \cite{Aaronson16}, which shows the randomized composition theorem for sequential query complexity with $\andfunc$ or $\orfunc$ (or $\andor$) being the inner function. Even though we restrict the inner function to be $\ksum$ or $\blockksum$ (or $\bkk$), our proof technique works for any problem that the $\andfunc$ or $\orfunc$ problem could be directly reduced to. Intuitively, the idea is to construct an algorithm for computing $f$ using an algorithm that computes \(f \circ \bkk\) with lesser parallelism but not too many more queries by choosing the instance of $\bkk$ carefully. 

\begin{theorem} \label{thm:rand_comp}
Let \(f: \mathcal D \to \{0,1\}\) be a (partial) function where \(\mathcal D \subseteq \{0,1\}^N\), and let \(\ksum_M, \blockksum_M\) and \(\bkk_{M^2}\) be defined as in \Cref{prob:ksum}, \Cref{prob:blockKsum} and \Cref{prob:bkk}, taking $k=\log N$. Then, \(R^{p\parallel }(f \circ \ksum_M) = \tilde \Omega \left ( \left \lceil \frac M p \right \rceil R^{\lceil p/M \rceil \parallel }(f) \right )\) and \(R^{p\parallel }(f \circ \blockksum_M) = \tilde \Omega \left ( \left \lceil \frac M p \right \rceil R^{\lceil p/M \rceil \parallel }(f) \right )\). Therefore, \(R^{p\parallel }(f \circ \bkk_{M^2}) = \tilde \Omega \left ( \left \lceil \frac{M^2}{p} \right \rceil R^{\left \lceil p/M^2 \right \rceil \parallel }(f) \right )\).
\label{thm:bkkscrandparallellowerbound} 
\end{theorem}

\begin{proof} We give the proof for \(\ksum_M\) here; the proof for \(\blockksum_M\) proceeds exactly in the same way. Consider a randomized algorithm \(\mathcal A\) that computes \(f \circ \ksum_M\) on string \(Y\) in \(t\) many $p$-parallel queries to \(Y\), where \(\vert Y \vert = MN\). Then we claim that there exists a randomized algorithm \(\mathcal B\) which can compute \(f\) on any string \(X\) using \(\tilde O(tp/M)\) total queries to \(X\), where \(\vert X \vert = N\). To show this claim, consider an algorithm \(\mathcal A\) that can compute \(f \circ \ksum\) in \(t\) many \(p\)-parallel queries. 

On input \(X:=x_1x_2...x_N\), \(\mathcal B\) proceeds as follows:
\begin{enumerate}
\item Construct an input \(Y\), a bit string of length \(MN\) in the following way. Divide \(Y\) into \(N\) blocks \(Y_i\) each of size \(M\). Further divide each \(Y_i\) into sub blocks of size \(10 k \log M\) each and interpret each sub block as a number in the range \(1\) to \(\Omega(M^k)\). We represent the \(j^{th}\) sub block of \(Y_i\) as \(Y_{ij}\). For each \(i\) from \(1\) to \(N\), do the following:
\begin{enumerate}
    \item For each \(j \in \{1, 2, ..., k-1\}\), set \(Y_{ij}\) as \(0\).
    \item Select \(z\) from \(\left \{k, k+1, ..., \frac M {10 k \log M} \right \}\) uniformly at random. Set \(Y_{iz}\) as \(\overline{x_i}\)
    \item For all \(j \notin \{1, 2, ..., k-1\}\) and \(j \neq z\), set \(Y_{ij}\) as \(1\) 
\end{enumerate}
\item Run algorithm \(\mathcal A\) on \(Y\) and return the output.

\end{enumerate}

Notice that \(\ksum_M(Y_i) = x_i\) which implies \(f \circ \ksum_M(Y) = f(X)\). For every input \(\mathcal B\) prepares, \(\mathcal A\) correctly computes \(f \circ \ksum_M(Y)\) with probability \(\geq 2/3\), hence \(\mathcal B\) also correctly computes \(f(X)\) with probability \(2/3\) where the probability is taken over both \(\mathcal A\)'s randomness as well as the randomness in preparing the inputs.

Since algorithm \(\mathcal A\) makes \(t\) many \(p\)-parallel queries to \(Y\), the total number of sub blocks \(Y_{ij}\) queried is at most \(tp\). Let \(Y'\) be the string made from \(Y\) ignoring the first \(k-1\) sub blocks in every block. Now, we invoke \Cref{lem:Expectation_Y} for string \(Y'\) setting the variables \(n = N\), \(m = \frac {M}{10k \log M} - (k-1)\), \(l = tp\) and the \(*\)'s in the position of the \(x_i\)'s. It follows that with probability \(\geq 9/10\) over the inputs that \(\mathcal B\) prepares, \(\tilde O(tp/M)\) total queries are made to the \(x_i\)'s. Since for every input \(\mathcal B\) prepares, \(\mathcal A\) computes \(f \circ \ksum_M\) correctly with probability \(2/3\), the probability that \(\mathcal B\) computes \(f\) correctly, while also making \(\tilde O(tp/M)\) total queries to \(X\) is \(\frac 2 3 \times \frac 9 {10} = \frac 3 5\). This is constant away from 1/2 and this proves the claim.

It remains to show that this algorithm can be simulated in low depth with the requisite amount of parallelism. We have two cases, (1) \(p \leq M\) and (2) \(p > M\).
\paragraph{Case 1}
When \(p \leq M\), the statement becomes \(\Rand^{p\parallel }(f \circ \ksum_M) = \tilde \Omega \left ( \frac {M \Rand(f)} {p}\right )\). We know algorithm \(\mathcal B\) makes \(\tilde O (tp/M)\) total queries to \(X\) to compute \(f\), so we can just do them all sequentially. Thus, 
\begin{align*}
\frac{tp}{M} = \tilde \Omega \left (\Rand(f) \right ) \implies t = \tilde \Omega \left ( \frac {M\Rand(f)}{p} \right )
\end{align*}
which implies the statement.

\paragraph{Case 2} When \(p > M\), the statement becomes \(\Rand^{p\parallel }(f\circ \ksum_M) = \tilde \Omega \left ( \Rand^{p/M\parallel }(f) \right )\). Here, we simulate algorithm \(\mathcal B\) using \(t\) many \(p/M\)-parallel queries. Let \(r_1, r_2, ..., r_t\) be the number of queries in every layer of \(\mathcal B\) made to an element of \(X\). Observe that to simulate layer \(i\) using \(p/M\) parallelism, we require \(\left \lceil \frac {r_i}{p/M} \right \rceil\) many \(p/M\)-parallel queries. Thus, the total number of \(p/M\)-parallel queries we require is:
\begin{align*}
\sum_{i=1}^t \left \lceil \frac{r_i}{p/M} \right \rceil \leq \sum_{i = 1}^t \left (1 + \frac {r_i} {p/M} \right ) \leq t + \frac {\sum_i r_i} {p/M} \leq t + \frac {\tilde O (tp/M)} {p/M} = \tilde O(t).
\end{align*}
Thus, we must have \(t = \tilde \Omega \left (\Rand^{p/M\parallel }(f) \right ) \) which implies the statement.

\end{proof}

Below is a helper lemma we require to prove \Cref{thm:rand_comp}.

\begin{lemma}
\label{lem:Expectation_Y}
    Let $F_i \in [M]^m$ be any fixed string for $i \in [n]$, and construct $Y_i \in ([M] \cup \{*\})^m$ by replacing a random character in $F_i$ with a $*$. Further define $Y = Y_1\parallel ...\parallel Y_n$. Any randomized algorithm (even ones that depend on $F_i$) that makes $l$ queries to $Y$ will query at most $20l/m$ many $*$'s with probability at least $9/10$.
\end{lemma}
\begin{proof}
    Let $\mathcal{A}$ be some randomized algorithm making $l$ queries to $Y$, and consider a randomly chosen $Y_i$. WLOG we can assume that $\mathcal{A}$ does not query a $Y_i$ after it has queried on a $*$--- any algorithm that continued to query past this point could be given answers from $F_i$ without more queries to $Y_i$. If $\mathcal{A}$ makes $l$ queries total then in expectation it makes $l/n$ queries to $Y_i$.
    
    We can write the strategy of $\mathcal{A}$ on $Y_i$ as a distribution $\mathcal{D}_T$ over decision trees (where the randomness now ranges over the randomness of $\mathcal{A}$ as well as the $*$ placement in other $Y_{j \neq i}$). Further, every node $v$ in decision tree $T$ from $\mathcal{D}_T$ has only two possibilities: the queried value is a $*$, or it is from $F_i$. $\mathcal{A}$ stops querying after finding a $*$, so the tree is a caterpillar tree of height $h$, and over the randomness of placing the $*$ the expected queries will be $(h+1)/2$. We therefore have \begin{align*}
        \mathbb{E}_{T \sim \mathcal{D}_T, Y}[\text{height}(T)] &= 2\mathbb{E}_{T \sim \mathcal{D}_T, Y}[\text{queries}_T(Y_i)] -1 \\
        &\leq \frac{2l}{n}
    \end{align*}
    A fixed tree $T$ of height $h$ finds a $*$ with probability at most $h/m$ over random inputs, as the problem is an unstructured search problem. Hence the probability of finding the star in block $Y_i$ is at most \begin{align*}
        \Pr_{T \sim \mathcal{D}_T, Y}[\text{star found in }Y_i] &\leq \frac{\mathbb{E}_{T \sim \mathcal{D}_T, Y}[\text{height}(T)]}{m} \\
        &\leq \frac{2l}{nm}
    \end{align*}
    By linearity of expectation, the expected total number of $1$'s found over all blocks is at most $2l/m$. A straightforward application of Markov's inequality now proves the claim.
\end{proof}

\section{Parallel query complexity in the cheatsheet framework}
\label{sec:AppendixCheatsheetFramework}

In this section, we present lower bounds for query complexity measures in the cheatsheet framework.

\subsection{Deterministic parallel lower bound} \label{sec:cheatsheet_det_parallel}

In this section, we will show that having access to a cheatsheet for a function $f$ is not too helpful for a $p$-parallel deterministic algorithm for computing $f$. Intuitively, the size of the cheatsheet is big enough that searching through it will already be too expensive for any such algorithm. Our result is a straightforward adaptation of Lemma 21 of \cite{Aaronson16} but we repeat it here for completeness. 

\begin{theorem} \label{thm:cheatsheetdetparallelimpact}
    Let $f :\mathcal{D}\rightarrow \{0,1\}$ be a partial function, where $\mathcal{D}\subset [M]^N$, and let $f_{\cheatsheet}$ be the cheat sheet version of $f$ with $c=10 \log N$ independent copies of $f$. Then $\emph{\Det}^{p\parallel }(f_{\cheatsheet}) \geq \emph{\Det}^{p\parallel }(f)$.
\end{theorem}

\begin{proof}
    We will use the adversarial argument against a $p$-parallel deterministic algorithm $\mathcal{A}$ making computing $f_\cheatsheet$. Suppose that $\mathcal{A}$ makes fewer than $\Det^{p\parallel }(f)$ queries. For all the indices of a $p$-parallel query made by $\mathcal{A}$ that belongs to an input to one of the $c$ copies of $f$, we will answer them as an adversary against a $p$-parallel deterministic algorithm for $f$ would, and for all the indices of a $p$-parallel query made by $\mathcal{A}$ that belongs to a cheatsheet, we will answer a $0$. Since $p D^{p \parallel}(f) \leq N$ and the number of cheatsheets are $2^c = \Omega(N^{10})$, $\mathcal{A}$ must not have queried a single bit in most of the cheatsheets. Furthermore, since $\mathcal{A}$ made less than $\mathcal{A}$ $\Det^{p\parallel}(f)$ $p$-parallel queries in total, it must not have made at least $\Det^{p\parallel}(f)$ $p$-parallel queries to any input to any $f$. Thus, it must not know the output to any $f$. Therefore, the partially revealed instance for $f_\cheatsheet$ could be completed to form a $0$-instance and also completed to form a $1$-instance. It follows that $\mathcal{A}$ must make at least $\Det^{p\parallel}(f)$ queries. 
\end{proof}

\subsection{Cheatsheet randomized parallel lower bound} \label{sec:cheatsheet_rand_parallel}

Similar to the previous section, we prove in this section that having access to a cheatsheet for a function $f$ is not too helpful for a $p$-parallel randomized algorithm for computing $f$. 

Our proof of this result (\Cref{thm:cheatsheetrandparallelimpact}) is a straightforward adaptation of Lemma 6 of \cite{Aaronson16} but we repeat it here for completeness. We begin with the following observation.

\begin{observation} \label{obs:cheatsheet}
    Let $f :\mathcal{D}\rightarrow \{0,1\}$ be a partial function, where $\mathcal{D}\subset [M]^n$, and let $f_{\cheatsheet}$ be the cheat sheet version of $f$ with $c=10 \log N$ copies of $f$. Let $\mathcal{A}$ be a $p$-parallel bounded-error randomized algorithm for $f_{\cheatsheet}$. Let $x^1, x^2, \cdots, x^c \in \mathcal{D}$ and $z$ be an input to $f_{\cheatsheet}$ comprised of inputs $x^1, x^2, \cdots, x^c \in \mathcal{D}$ to $f$ with an all-zero cheat sheet\footnote{For any input $z$ to $f_{\cheatsheet}$, we assume, without loss of generality, that if the cheat sheet corresponding to $f^{\times c}(z)$ is all-zero, then $f_{\cheatsheet}(z) = 0$. As argued in \cite{Aaronson16}, this can be ensured by requiring all valid cheat sheets to begin with a $1$.}. Let $\ell$ index the cheat sheet associated with the string $f(x^1)$, $f(x^2), \cdots f(x^c)$. Then, $\mathcal{A}$ must query at least 1 bit in the $\ell^{th}$ cheat sheet with probability at least $1/3$.  
\end{observation}

\begin{proof}
    Let $y$ be an input to $f_{\cheatsheet}$ comprised of inputs $x_1, x^2, \cdots, x^c \in \mathcal{D}$ to $f$ and the same cheat sheet as $z$ except for the $\ell^{th}$ cheat sheet being valid. That is, $f_{\cheatsheet}(y) = 1$. Since $\mathcal{A}$ computes $f_{\cheatsheet}$ with probability at least $2/3$, $f_{\cheatsheet}(y) \neq f_{\cheatsheet}(z)$ and $y$ and $z$ only differ at the $\ell^{th}$ cheat sheet, our desired claim follows.
\end{proof}

\begin{theorem} \label{thm:cheatsheetrandparallelimpact}
    Let $f :\mathcal{D}\rightarrow \{0,1\}$ be a partial function, where $\mathcal{D}\subset [M]^n$, and let $f_{\cheatsheet}$ be the cheat sheet version of $f$ with $c=10 \log N$ copies of $f$. Then $\emph{\Rand}^{p\parallel }(f_{\cheatsheet})=\Omega(\emph{\Rand}^{p\parallel }(f)/c^2) = \tilde{\Omega}(\emph{\Rand}^{p\parallel }(f))$.
\end{theorem}

\begin{proof}
    Let $\mathcal{D}^0$ and $\mathcal{D}^1$ be $0$-input and $1$-input distributions for $f$ that is hard to distinguish for a $p$-parallel bounded-error randomized algorithm. In particular, the $p$-parallel randomized query complexity of distinguishing them with probability at least $1/2 + 1/12c$ is $\Omega(R^{p\parallel }(f)/c^2)$. Let $\mathcal{A}$ be any $p$-parallel bounded-error randomized algorithm for $f_{\cheatsheet}$. 

    For each input $z$ to $f_{\cheatsheet}$ with an all-zero cheat sheet and for each index $i \in [2^c]$ of a cheat sheet, let $p^z_i$ be the probability that $\mathcal{A}$ queries the $i^{th}$ cheat sheet on input $z$. Let $\ell_z$ index the cheat sheet associated with the string $f^{\times c}(z)$. From \Cref{obs:cheatsheet}, we know that $p^z_{\ell_z} \geq 1/3$. Furthermore, $\sum_{i \in [2^c]} p^z_i \leq p \Rand^{p\parallel }(f)$.

    For each $j \in [2^c]$, let $\mathcal{D}^j = \mathcal{D}^{j_1} \times \mathcal{D}^{j_2} \times \cdots \times \mathcal{D}^{j_c}$. Let $q^j_i = \mathbb{E}_{z \sim \mathcal{D}^j} p^z_i$ for all $i,j \in [2^c]$. By the arguments made above, we can infer that $q^j_j \geq 1/3$ and $\sum_{i \in [2^c]} q^j_i \leq pR^{p\parallel }(f)$. This means that for all but $2^{c/2} = N^5$ many $i \in [2^c]$, $q^j_i \leq 1/(p\Rand^{p\parallel }(f))^5$. In particular, for each $j \in [2^c]$, there is an $i \in [2^c]$ such that $q^j_i \leq 1/6$.

    Let $k \in [2^c]$ be an index such that $q^{0^c}_k \leq 1/6$. We argued above that $q^k_k \geq 1/3$. Let $0^c = k_0, k_1, \cdots, k_c = k$ be a sequence of indices such that $|k_{i-1} \oplus k_i| \leq 1$ for all $i \in [c]$. Therefore, there must exist an $i \in [c]$ with $q^{k_i}_k - q^{k_{i-1}}_k \geq 1/6c$ and $q^{k_{i-1}}_k \leq 1/3$. Without loss of generality, we can assume that $|k_i| -|k_{i-1}| = 1$.

    Using $\mathcal{A}$, we now give an algorithm to distinguish $\mathcal{D}^0$ and $\mathcal{D}^1$ with probability at least $1/2 + 1/12c$. Given an input from $\mathcal{D}_0$ or $\mathcal{D}_1$, our algorithm samples an input from $\mathcal{D}^{k_{i-1}}$ or $\mathcal{D}^{k_i}$; this can be done since $|k_{i-1} \oplus k_i| \leq 1$. Therefore, deciding if the sampled input is from $\mathcal{D}^{k_{i-1}}$ or from $\mathcal{D}^{k_i}$ lets us decide whether the given input is from $\mathcal{D}_0$ or from $\mathcal{D}_1$.

    We run the algorithm $\mathcal{A}$ on the sampled input with an all-zero cheat sheet, and if it queries a bit in the $k^{th}$ cheat sheet, we output $1$ and otherwise, we output $1$ with probability $p = \max\{0, (1-a-b)/(2-a-b)\}$ where $a = q^{k_i}_k$ and $b = q^{k_{i-1}}_k$. If $p = 0$, then since $b \leq 1/3$, we have $a \geq 2/3$ so our algorithm succeeds with probability at least $2/3$. Otherwise, our algorithm succeeds with probability at least $1/2 + (a-b)/2 \geq 1/2 + 1/12c$: if the given input is sampled from $\mathcal{D}^1$, then the success probability is $a + (1-a)\frac{1-a-b}{2-a-b} \geq (1+a-b)/2$; otherwise, it is $(1-b) \frac{1}{2-a-b} \geq (1+a-b)/2$. But this means that our algorithm must make $\Omega(\Rand^{p\parallel }(f)/c^2)$ $p$-parallel randomized queries. Hence, $\mathcal{A}$ needs to make $\Omega(\Rand^{p\parallel }(f)/c^2)$ $p$-parallel randomized queries, which implies the desired result. 
\end{proof}

\section{Pointer Chasing bounds}
\label{sec:AppendixPointer-chasing-proofs}

In this section, we present a proof for each of the parts of \Cref{thm:pointer_chasing}, beginning with a deterministic algorithm.

\paragraph{Deterministic parallel upper bound}
\begin{lemma}
    \label{lemma:pointer-chasing-det-ub}
    $\emph{\Det}^{p\parallel}(\textsc{Pointer Chasing}) \leq \min(k, N/p)$.
\end{lemma}

\begin{proof}
    We will consider two cases.
    \begin{enumerate}
        \item $k \leq N/p$. For each $i \in [k]$, we can iteratively compute $X^i$ using $k$ sequential queries to the oracle for the input $X$. Thus, we will know the last bit of $X^k(0^n)$.     
        \item $k > N/p$. Using $N/p$ $p$-parallel queries, we can compute the whole input $X$ (and in particular, can compute the last bit of $X^k(0^n)$). 
    \end{enumerate}
\end{proof}
\vspace{-1pc}

\paragraph{Deterministic parallel lower bound}
\begin{lemma}
    \label{lemma:pointer-chasing-det-lb}
    $\emph{\Det}^{p\parallel}(\textsc{Pointer Chasing}) = \Omega(\min(k, N/p))$.
\end{lemma}

\begin{proof}
    Let $\mathcal A$ be a deterministic algorithm for pointer chasing. We will construct two instances $f, g$ of pointer chasing such that $f$ is a YES instance and $g$ is a NO instance, but $\mathcal A$ does not distinguish $f$ and $g$ supposing that it makes fewer than $\min(k/10, N/10p)$ queries.

    We will iteratively construct both instances, maintaining the invariant that the algorithm cannot ``jump ahead'' in the chain. In either case, the chain will be a path with no cycles. To answer a round of $p$-parallel queries, suppose that $\mathcal A$ has queried input function $h$ at index/value pairs $(0, h(0)), (h(0), h(h(0))), ..., (h^{l-1}(0), h^{l}(0))$, as well as some other set $S$ of index/value pairs that is of size smaller than $\min(k/10, N/10p)$. Suppose by induction that $S$ contains only pairs of the form $(x, 0)$ (this is clearly true for an algorithm making no queries, and we will maintain it here). Let $B$ denote the set of positions which are being queried in this round. We will then pick the value $y = h^l(0)$ such that $y$ is not any of the indices in $B$, nor in $S$, nor does it create a chain in the path (this is always possible so long as fewer than $\min(k/10, N/10p) < N/2$ queries have been made). We thus maintain the invariant that after $q$ rounds of queries, a chain of length at most $q$ has been revealed. Therefore, if $\mathcal A$ makes $k-1$ rounds of queries (and supposing that it makes fewer than $N/10p$ total $p$-parallel queries), we can construct $f$ in the manner outlined aboved, and finally choose $f^k(0)$ so the last bit is $1$. Similarly for $g$, but choosing $g^k(0)$ such that the last bit is $0$. Any input/value pairs not fixed by this procedure we can simply set the value to $0$.
\end{proof}

\paragraph{Randomized parallel lower bound}
\begin{lemma}
    \label{lemma:pointer-chasing-randomized-lb}
    $\emph{\Rand}^{p\parallel}(\textsc{Pointer Chasing}) = \Omega(\min(k, N/pk))$.
\end{lemma}

\begin{proof}
    From Yao's lemma, it suffices to consider deterministic algorithms on some distribution of inputs. We will consider inputs that are random permutations, i.e. let $f:\{0,1\}^n \rightarrow \{0,1\}^n$ be a random bijective function. To reason about a deterministic algorithm we will consider it's query transcript, i.e. a list $(x, f(x))$ of all points that it has queried. After $q$ rounds of $p$-parallel queries, this transcript is of size $pq$. 
    
    It suffices to consider $k-1$-round $p$-parallel algorithms; we will show that such algorithms must make at least $\Omega(N/k)$ total queries to succeed with constant probability. Suppose that the chain $(0, f(0)), ..., (f^{l-1}(0), f^l(0))$ of length $l$ appears in the transcript, and further that no $f^j(0)$ for $l < j \leq k$ appears (e.g. the algorithm has learned the first $l$ elements of the chain, but no more). Because $f^k(0)$ determines the final output, clearly such a deterministic algorithm cannot decide $f$ with probability exceeding $1/2$. Note that this invariant form of the transcript trivially holds with $l=0$ for an algorithm making no queries. In a round of $p$-parallel queries, we will bound the probability that the algorithm transitions from a transcript satisfying the invariant, to one that does not.
    
    A round of $p$-parallel queries selects $p$ positions to query, out of $N$ total. The remaining $k-l$ elements in the chain do not occur in the list of $q$ many queried elements, and so we can consider placing them randomly among the remaining $N-q$, which satisfies $N-q \geq N/2$ by assumption. We therefore have probability $2p(k-l) / N \leq 2pk / N$ probability that the remaining elements in this chain do not intersect with the query set. If there is no intersection, then the chain length will increase by exactly $1$ in this round of queries. Hence, in $k-1$ rounds we will break the invariant with probability $(k-1) * 2pk /N = O(pk^2 / N)$ by the union bound. For this to be a constant we need $p = \Omega(N/k^2)$, which gives $\Omega(N/k)$ total queries.
\end{proof}

\paragraph{Quantum lower bound}

\begin{lemma}
    \label{lemma:pointer-chasing-quant-lb}
    $\emph{\Quant}(\textsc{Pointer Chasing}) = \Omega(k)$.
\end{lemma}

\begin{proof}
    We can prove this result by a straightforward reduction to parity on $k/2$ bits (let $k$ even), which is known to require $\Omega(k)$ quantum queries \cite{Beals98}. This reduction is inspired by the well known no-fast forwarding theorem \cite{berry07hamiltonian}. Let $X \in \{0,1\}^{k/2}$ denote such an instance of parity. We will construct $f$ the first $k$ index/value pairs of $f$ as follows, for all $i \in [0, ..., k/2-1]$. \begin{align}
         f(2i) &= \begin{cases}
            2i+2 & \text{ if $X_i = 0$} \\
            2i+3 & \text{ if $X_i = 1$}
        \end{cases} \\
         f(2i+1) &= \begin{cases}
            2i+3 & \text{ if $X_i = 0$} \\
            2i+2 & \text{ if $X_i = 1$}
        \end{cases}
    \end{align}
    and set the remaining inputs to be all $0$. In this way, we have that $f^{k/2}(0)$ is equal to the parity of $X$. In particular, there is a path $P$ of length $k/2$ starting at $(0, f(0))$. The $i$-th link in this path goes from an index of even/odd parity to one of (a) even/odd parity if $X_i=0$, and (b) odd/even parity if $X_i=1$. Hence the parity of the full string $X$ is the same as the last bit of $f^k(0)$.
\end{proof}

\end{document}